\documentclass[12pt]{article}
\usepackage{amsmath}
\usepackage{graphicx}
\usepackage{natbib}
\usepackage{url} 

\usepackage{xcolor}      
\usepackage{todonotes}   

\usepackage{comment}
\usepackage[utf8]{inputenc}
\usepackage{diagbox} 
\usepackage{float}   

\usepackage{makecell} 
\usepackage{booktabs}
\usepackage{algorithm}
\usepackage{algpseudocode}
\usepackage{amsmath} 
\usepackage{amsthm}   
\usepackage{mathrsfs}

\usepackage{amssymb} 
\usepackage{bm}      

\usepackage{xr}

\DeclareMathOperator*{\argmin}{arg\,min}
\newtheorem{theorem}{Theorem}[section]

\newcommand{\blind}{0}

\begin{document}

\def\spacingset#1{\renewcommand{\baselinestretch}%
{#1}\small\normalsize} \spacingset{1}


\if0\blind
{
  \title{\bf Regularized estimation for highly multivariate spatial Gaussian random fields}
  \author{Francisco Cuevas-Pacheco,
    francisco.cuevas@usm.cl\hspace{.2cm}\\
    Departamento de Matemática, Universidad Técnica Federico Santa María, \\ Valparaíso, Chile.\\
    \\
    Gabriel Riffo \\
    Instituts für Mathematik, Technische Universität Berlin, \\ Berlin, Germany \\
    and \\
    Xavier Emery \\
    Department of Mining Engineering, Universidad de Chile,\\
    Santiago, Chile, \\
    Advanced Mining Technology Center, Universidad de Chile, \\ 
    Santiago, Chile.}
    
  \maketitle
} \fi

\if \blind
{
  \bigskip
  \bigskip
  \bigskip
  \begin{center}
    {\LARGE\bf Title}
\end{center}
  \medskip
} \fi

\bigskip
\begin{abstract}
Estimating covariance parameters for multivariate spatial Gaussian random fields is 
computationally challenging, as the number of parameters grows rapidly with the number 
of variables, and likelihood evaluation requires operations of order $\mathcal{O}((np)^3)$. 
In many applications, however, not all cross-dependencies between variables are relevant, suggesting that sparse covariance structures may be both statistically advantageous and practically necessary. We propose a LASSO-penalized estimation framework that induces sparsity in the Cholesky factor of the multivariate Mat\'{e}rn correlation matrix, enabling automatic identification of uncorrelated variable pairs while preserving positive semidefiniteness. Estimation is carried out via a projected block coordinate descent algorithm that decomposes the optimization into tractable subproblems, with constraints enforced at each iteration through appropriate projections. Regularization parameter selection is discussed for both the likelihood and composite likelihood approaches. We conduct a simulation study demonstrating the ability of the method to recover sparse correlation structures and reduce estimation error relative to unpenalized approaches. We illustrate our procedure through an application to a geochemical dataset with $p = 36$ variables and $n = 3998$ spatial locations, showing the practical impact of the method and making spatial prediction feasible in a setting where standard approaches fail entirely.
\end{abstract}

\noindent%

\spacingset{1.75} 

\section{Introduction}

The analysis of multivariate spatial data is increasingly common in the environmental sciences and natural resources engineering, among other fields. 
Modeling the spatial dependencies among 
coregionalized variables is essential for accurate spatial prediction through cokriging, but faces significant computational and statistical challenges as the number of variables grows.

The modeling of the spatial correlation structure generally considers a 
covariance function $C(\cdot)$ 
belonging to a parametric family of matrix-valued covariances $\{C(\cdot;\theta), \theta \in \Theta \}$, where $\Theta$ is the parameter space, and the function $C(\cdot;\theta)$ is positive semidefinite. In the literature, several parametric models for the covariance function have been proposed, such as the separable model, the linear model of coregionalization \citep{wackernagel2003multivariate}, asymmetric models \citep{li2011approach}, and the multivariate Matérn model \citep{gneiting2010matern, apanasovich2012valid}. For a thorough review, the reader is referred to \citet{genton2015cross}, with their exhaustive list of references.

Despite these advances, a fundamental problem persists: for a multivariate random field, 
with $p$ observed characteristics on $n$ spatial locations, the number of covariance parameters grows 
very fast. Therefore, maximum likelihood estimation requires computing the determinant and 
inverse of $np \times np$ covariance matrices with computational complexity 
$\mathcal{O}((np)^3)$, making estimation infeasible for large $p$ or $n$. Composite 
likelihood approaches \citep{bevilacqua2012estimating, bevilacqua2016composite} reduce 
computational costs by using lower-dimensional likelihoods based on pairs of observations. 
Similarly, Vecchia approximations \citep{vecchia1988estimation, katzfuss2021general} achieve 
scalability by conditioning each observation on a small set of neighbors, but neither approach addresses the fundamental issue of parameter dimensionality.

In many applications, not all $\mathcal{O}(p^2)$ cross-dependencies are necessary. Certain variable pairs may exhibit negligible spatial cross-correlation, suggesting sparse correlation structures where some cross-covariances are exactly zero. In mining exploration, for example, concentrations of certain elements may be uncorrelated despite being measured at the same locations. Identifying such sparse structures offers statistical advantages as fewer parameters, reduced overfitting, and improved interpretability, together with computational benefits as lower memory requirements, and/or faster predictions.

Regularization methods such as the LASSO \citep{tibshirani1996regression} have proven effective for inducing sparsity in high-dimensional statistical models by adding an $L_1$ penalty to the objective function. However, applying LASSO to multivariate spatial covariance estimation presents unique challenges: covariance parameters must satisfy complex positive semidefiniteness constraints, and the optimization problem is highly nonlinear. 

In this work, we propose a LASSO-penalized estimation framework that induces sparsity in a block of the parameter vector. Estimation is performed through a projected block coordinate descent algorithm that exploits the natural grouping of parameters, decomposing the optimization into smaller tractable subproblems, each corresponding to a subset of parameters, while ensuring convergence under mild regularity conditions and all constraints are satisfied at each iteration. To select the regularization parameter, we employ AIC for likelihood-based estimation and CLIC for composite likelihood approaches \citep{varin2005note}. We apply this method to the Cholesky factor of the multivariate Matérn correlation matrix to automatically identify uncorrelated variable pairs while maintaining positive semidefiniteness. Although developed for the multivariate Matérn model, the framework can be readily adapted to other covariance families. Through extensive simulation studies and a real application to geochemical data with $p=36$ variables and $n=3998$ locations, we demonstrate that the method successfully identifies sparse correlation structures, reduces memory requirements from over $130$ GB to $1.31$ GB, and makes spatial prediction computationally feasible in settings where standard methods fail.

The remainder of this paper is organized as follows. Section 2 provides background on multivariate spatial random fields and the multivariate Matérn model. Section 3 presents the LASSO-penalized estimation framework and the block coordinate descent algorithm. Section 4 reports simulation results. Section 5 presents the geochemical application. Section 6 concludes with discussion and future directions.

\section{Background}

\subsection{Stationary and isotropic multivariate random fields}

Let $\bold{Z}=\{(Z_1(s),\dots,Z_p(s))^{\top}:s\in\mathbb{R}^{d}\}$ be a $p$-variate 
random field where $s$ denotes the spatial location, and $Z_i(s)$ represents 
the $i$-th component at location $s$. Furthermore, assume that $\bold{Z}$ is a Gaussian 
random field, that is, for any positive integer $n$ and any finite set of spatial locations 
$\{s_1,\dots,s_n\}\subset \mathbb{R}^d$, the vector $(\bold{Z}(s_1),\dots,\bold{Z}(s_n))$ 
follows a multivariate normal distribution.

The Gaussian assumption implies that the first- and second-order moments completely 
characterize the finite-dimensional distributions of the random field, i.e., the mean vector 
function $\mu(s) = [\mu_{1}(s), \ldots, \mu_{p}(s) ]$, with $\mu_{i}(s) = 
\mathbb{E}[Z_{i}(s)]$, and the covariance function ${\rm{Cov}}[Z_{i}(s), Z_{j}(s^{\prime})] 
= C_{ij}(s,s^{\prime})$, for all $s, s^{\prime} \in \mathbb{R}^{d}$ and $i, j \in \{1 , 
\ldots, p \}$. In order to propose parametric models for both the mean and covariance 
functions, we recall that the former is an otherwise unconstrained measurable function, while 
the latter must satisfy the positive semidefiniteness condition, that is, for any finite 
collection of locations $s_1, \ldots, s_n \in \mathbb{R}^d$ and any vector $a \in 
\mathbb{R}^{np}$, it must hold that
\begin{equation}
a^\top \Sigma a \geq 0, \label{eq:positivedefinite}
\end{equation}
where $\Sigma$ is the block covariance matrix with $(i,j)$-th block given by $C_{ij}(s_k, 
s_l)$ for $k,l = 1, \ldots, n$. The matrix-valued function $C=[C_{ij}]_{i,j=1}^p$ is said 
to be positive definite if inequality \eqref{eq:positivedefinite} is strict unless $a$ is a 
zero vector. Both conditions on the mean and covariance functions are necessary and sufficient 
for the multivariate normal distribution to be well-defined.

To reduce model complexity and facilitate estimation, we assume that the random field is 
stationary in a multivariate sense. Formally, a multivariate random field $\{\bold{Z}(s): 
s \in \mathbb{R}^{d}\}$ is said to be stationary if its mean function is constant, i.e., 
$\mu_i(s) = \mu_i$ for all $s \in \mathbb{R}^d$ and $i = 1, \ldots, p$, and its covariance 
function depends only on the spatial lag $s - s'$, i.e., $C_{ij}(s, s') = C_{ij}(s - s')$ 
for all $s, s' \in \mathbb{R}^d$, and $i, j = 1, \ldots , p$. The covariance is furthermore 
isotropic if it only depends on the norm of the spatial lag $h := \| s - s^{\prime} \|$.

\subsection{The multivariate Matérn covariance model}

A very well-known isotropic model is the multivariate Matérn covariance function 
\citep{gneiting2010matern}, which is given by:
\begin{equation}
\begin{split}
    \mathcal{C}_{ij}(h; \sigma_{ij}, \alpha_{ij}, \nu_{ij}) &= \sigma_{ij} 
    \mathcal{M}(h; \alpha_{ij}, \nu_{ij}) \\&:= \sigma_{ij}^2 
    \frac{2^{1-\nu_{ij}}}{\Gamma(\nu_{ij})}(\sqrt{2\nu_{ij}} \alpha_{ij} h)^{\nu_{ij}} 
    K_{\nu_{ij}}(\sqrt{2\nu_{ij}} \alpha_{ij} h), \label{eq:matern_model}
\end{split}
\end{equation}

\noindent where $K_{\nu}(\cdot)$ is the modified Bessel function of the second kind 
\citep{abramowitz1965handbook}, $\sigma_{ij} \in \mathbb{R}$ is a parameter that controls 
the covariance between the $i$-th and $j$-th random field components (with $\sigma_{ii} > 0$ 
for the marginal variances), $\alpha_{ij} > 0$ is a parameter that controls the practical 
range between the $i$-th and $j$-th random field components, and $\nu_{ij} > 0$ is the 
smoothness parameter, with $\nu_{ii}$ controlling the smoothness of the $i$-th random field 
component \citep{stein1999interpolation}.

To ensure that the multivariate model \eqref{eq:matern_model} is a valid covariance function, 
some conditions on the parametric space must be imposed. The following result from 
\citet{apanasovich2012valid} provides sufficient conditions to ensure the positive 
semidefiniteness of \eqref{eq:matern_model}.

\begin{theorem}
\label{Apanaso} The multivariate Matérn model \eqref{eq:matern_model} is a valid covariance 
function on $\mathbb{R}^{d} \times \mathbb{R}^{p}$ if there exist matrices $R_A$, $R_B$, 
and $\rho$ such that:

\begin{itemize} 
    \item $\nu_{ij} = \dfrac{\nu_{ii} + \nu_{jj}}{2} + \Delta_A(1 - R_{A,ij})$, where 
    $0 \leq R_{A,ij} \leq 1$, $R_A$ is a correlation matrix, and $\Delta_A \geq 0$. 
    \item $\alpha_{ij}^2 = \dfrac{\alpha_{ii}^2 + \alpha_{jj}^2}{2} + \Delta_B(1 - 
    R_{B,ij})$, where $0 \leq R_{B,ij} \leq 1$, $R_B$ is a correlation matrix, and 
    $\Delta_B \geq 0$. 
    \item $\sigma_{ij} = \sqrt{\sigma_{ii}^2 \sigma_{jj}^2} \rho_{ij} 
    \dfrac{\alpha_{ii}^{\nu_{ii}} \alpha_{jj}^{\nu_{jj}}}{\alpha_{ij}^{2\nu_{ij}}} 
    \dfrac{\Gamma(\nu_{ij})}{\sqrt{\Gamma(\nu_{ii}) \Gamma(\nu_{jj})}}$, where $\rho$ is 
    a correlation matrix. 
\end{itemize}

\end{theorem}

\noindent In particular, the matrices $R_A$ and $R_B$ can be interpreted as a set of 
parameters adding a deviation of $\nu_{ij}$ and $\alpha_{ij}^2$ with respect to 
$(\nu_{ii} + \nu_{jj})/2$ and $(\alpha_{ii}^2 + \alpha_{jj}^2)/2$, respectively. The 
parameters $\Delta_A$ and $\Delta_B$ represent the maximum deviation of $\nu_{ij}$ from 
$(\nu_{ii} + \nu_{jj})/2$ and of $\alpha_{ij}^2$ from $(\alpha_{ii}^2 + \alpha_{jj}^2)/2$.

Note that when $\Delta_B = 0$, the matrix $R_B$ becomes unidentifiable. Moreover, if $R_B$ 
is an equicorrelation matrix, the model also becomes unidentifiable. For a small number $p$ 
of random field components, \citet{apanasovich2012valid} recommend modeling $R_B$ as an 
equicorrelation matrix, defining $\Delta_B' := \Delta_B(1 - R_{B,ij})$, and then setting:
\begin{equation*}
\alpha_{ij}^2 = \dfrac{\alpha_{ii}^2 + \alpha_{jj}^2}{2} + \Delta_B'.    
\end{equation*}

\noindent Based on Theorem 1, \citet{apanasovich2012valid} propose to plug-in the 
restrictions into Equation \eqref{eq:matern_model}, obtaining the following 
parameterization:
\begin{equation*}
C_{ij}(h; \sigma, \rho_{ij}, \alpha_{ij}, \nu_{ij}) := \sqrt{\sigma_i^2 \sigma_j^2} 
\rho_{ij} \dfrac{\alpha_{ii}^{\nu_{ii}} \alpha_{jj}^{\nu_{jj}}}{\alpha_{ij}^{2\nu_{ij}}} 
\dfrac{\Gamma(\nu_{ij})}{\sqrt{\Gamma(\nu_{ii}) \Gamma(\nu_{jj})}} 
\mathcal{M}(h; \alpha_{ij}, \nu_{ij}). \label{ApanaMatern}
\end{equation*}

Notice that the number of parameters is $3(p^2-p)/2$. This makes estimation procedures 
such as maximum likelihood unfeasible when $p$ is large. In particular, the works of 
\citet{gneiting2010matern}, \citet{apanasovich2012valid} and \citet{emery2022new} use up 
to $p = 5$. In what follows, we tackle this problem by exploring penalized estimation.

\subsection{Estimation procedures}

Let $z_{i}(s_k)$ be an observation of the $i$-th variable of a multivariate Gaussian 
random field at location $s_{k} \in \mathbb{R}^d$, and let $Z_n = (z_1(s_1), \dots, 
z_p(s_1), \dots, z_1(s_n), \dots, z_p(s_n))$ be the stack of the observations of each 
one of the random field components. Assume that the covariance function of the random field, 
$C(\cdot; \theta)$, belongs to a parametric family of covariance functions, where $\theta$ 
denotes the parameter vector which belongs to the parameter space $\boldsymbol{\Theta}$. 
To estimate the parameter vector $\theta$, we can proceed by solving the following problem:

\begin{equation}
    \min_{\theta \in \boldsymbol{\Theta}} f(\theta; Z_n), \label{estimation}
\end{equation}
where $f$ is an objective function. Under this framework, we obtain the maximum likelihood 
estimator by choosing $f$ as the negative log-Gaussian multivariate density function, i.e.:
\begin{equation*}
    \ell(\theta; Z_n) = -\frac{np}{2}\log(2\pi) - \frac{1}{2} \log(|\Sigma_n (\theta)|) 
    - \frac{1}{2} (Z_n - \mu)^{\top} \Sigma_n(\theta)^{-1} (Z_n - \mu),
\end{equation*}
where $\mu = (\mu_1, \dots, \mu_p, \dots, \mu_1, \dots, \mu_p)^\top$ is the mean vector 
and $\Sigma_n(\theta)$ is the covariance matrix defined in \eqref{eq:positivedefinite} 
with entries given by the multivariate Matérn covariance function with parameter vector 
$\theta$. The estimator obtained by maximizing this function (or minimizing $-2\ell(\theta)$) 
is called the maximum likelihood estimator. From the statistical point of view, the maximum 
likelihood estimator is consistent and asymptotically normal under certain conditions 
specified in \citet{gaetan2010spatial}. However, for large datasets the computation of 
$\Sigma_n(\theta)^{-1}$ requires $\mathcal{O}((np)^3)$ operations and $\mathcal{O}((np)^2)$ 
memory, making maximum likelihood estimation computationally expensive or infeasible. 
Therefore, alternatives that reduce computational cost while preserving statistical 
properties are appealing.

Among the different alternatives to reduce computational complexity in terms of both runtime 
and memory is the composite likelihood approach \citep{bevilacqua2016composite}. In this 
regard, we consider $Z_{n} := [Z_{n}^{1}, \ldots, Z_{n}^{K}]$ a partition of the random 
vector $Z_{n}$, and we compute the composite likelihood function given by

\begin{equation*}
    CL(\theta) = \sum_{k=1}^{K-1} \sum_{l=k+1}^{K} {cl}_{kl}(\theta) \, 
    w_{kl},\label{plp likelihood}
\end{equation*}

\noindent where ${cl}_{kl}(\theta)$ denotes the joint density function of the random vector 
$(Z_{n}^{k}, Z_{n}^{l})$. In this work, we propose an adapted version of the composite 
likelihood for multivariate Gaussian random fields. Let $k,l\in \{1,\dots,n\}$, and define 
the matrix $\Sigma_{kl}(\theta) = \text{Cov}(\mathbf{Z}(s_k), \mathbf{Z}(s_l))$, and the 
function $cl_{kl}(\cdot)$ such that:

\begin{equation*}
    {cl}_{kl}(\theta) = -\frac{1}{2} \left(\log |\Sigma_{kl}(\theta)| + 
    (\mathbf{Z}(s_k), \mathbf{Z}(s_l))^{\top} \Sigma_{kl}(\theta)^{-1} 
    (\mathbf{Z}(s_k), \mathbf{Z}(s_l)) \right).
\end{equation*}

This approach operates on $2p \times 2p$ matrices, whose determinant and inverse 
computations each have a cost of $\mathcal{O}((2p)^3)$. With the nearest-neighbor weights 
defined below, summing over all contributing pairs reduces the total computational complexity 
to $\mathcal{O}((2p)^3 nv)$ operations and $\mathcal{O}((2p)^2)$ memory, where $v$ is the 
number of nearest neighbors. This yields substantial computational savings when $n$ is large 
relative to $p$.

The weights $w_{kl}$ represent a balance between computation time and estimation efficiency. 
Here we use the proposal from \citet{caamano2024nearest}:

\begin{equation*}
    w_{kl} := w(s_{k},s_{l}) =  
    \begin{cases}
        1, & s_{l} \in N_{k}(v); \\
        0, & s_{l} \notin N_{k}(v),
    \end{cases}
\end{equation*}
where $N_{k}(v)$ denotes the set of the $v$ nearest neighbors to $s_k$ in the sample. For 
other choices of weights $w_{kl}$, see \citet{joe2009weighting, bevilacqua2012estimating, 
davis2011comments}. The study of optimal weight selection is beyond the scope of this paper.

Composite likelihood approaches are well suited to large datasets, but as with full 
likelihood, most covariance families for large $p$ require a large number of parameters 
subject to complex constraints. This can lead to estimates that violate the parameter space, 
breaking positive semi-definiteness of the estimated covariance function and producing invalid 
models. To address both the parameter dimensionality and the constraint enforcement problems 
simultaneously, we consider a penalized version of the objective function 
\eqref{estimation}. Specifically, letting $\theta_L \subseteq \theta$ denote the subset 
of parameters for which sparsity is desired, we propose:
\begin{equation}
    \min_{\theta \in \boldsymbol{\Theta}} f(\theta; Z_n) + \lambda \| \theta_L \|_{1}, 
    \label{estimation_penalty}
\end{equation}

\noindent where $\lambda \geq 0$ is the regularization parameter controlling the degree of 
sparsity. In the next sections, we describe the choice of $\theta_L$ and the algorithm 
used to solve \eqref{estimation_penalty}.

\section{Main results}

\subsection{Projected block coordinate descent}\label{propuesta}

In what follows, we solve Equation \eqref{estimation} by considering a block coordinate 
gradient descent algorithm with proximal and projected steps 
\citep{tseng2001convergence, wright2015coordinate}. This algorithm takes advantage of 
the natural grouping of parameters into subsets that can be updated efficiently and, in 
some cases, independently. Specifically, we partition the parameter vector into $\kappa$ 
blocks, i.e., $\theta = (\theta^{(1)}, \ldots, \theta^{(\kappa)})$, where each block 
$\theta^{(b)}$ is an element of the projected parameter space $\Theta^{(b)}$, i.e. 
$\Theta = \Theta^{(1)} \times \cdots \times \Theta^{(\kappa)}$. This block structure 
allows us to exploit separability on the objective function and impose constraints on 
the parameter vector, which can improve convergence speed and scalability to 
high-dimensional problems \citep{wright2015coordinate, tseng2001convergence}. Now, let 
$\theta^{-(b)}$ be the vector $\theta$ without the $b$-th block. Then, the block 
gradient update is computed as:

\begin{align*}
    \theta^{(1)}_{k} &= \mathcal{P}^{1}\left( \theta^{(1)}_{k-1} - \tau_{k} 
    \nabla f\left(\theta^{(1)}_{k-1} \mid \theta^{-(1)}_{k}\right) \right) \\
    \vdots &= \vdots \\ 
    \theta^{(\kappa)}_{k} &= \mathcal{P}^{\kappa}\left( \theta^{(\kappa)}_{k-1} - 
    \tau_{k} \nabla f\left(\theta^{(\kappa)}_{k-1} \mid \theta^{-(\kappa)}_{k}\right) 
    \right), 
\end{align*}
\noindent where $f(\theta^{(b)} \mid \theta^{-(b)})$ is the objective function evaluated 
on the block $\theta^{(b)}$ while fixing the remaining blocks $\theta^{-(b)}$, $\tau_k > 0$ 
is the step size at iteration $k$, and $\mathcal{P}^{b}(x)$ is the orthogonal projection 
of the vector $x$ onto the parameter space $\Theta^{(b)}$, i.e.,
$$\mathcal{P}^{b}(x) = \argmin_{y \in \Theta^{(b)}} \| x - y \|^{2},$$

\noindent provided that $\Theta^{(b)}$ is a convex set. The specific form of the gradient 
and the projection operators for the multivariate Matérn model are detailed in Appendix A.

\subsection{Multivariate Matérn case}\label{matern_case}

Here we consider a parsimonious version of the multivariate Matérn covariance function 
by fixing the smoothness parameter $\nu_{ij} = \nu$ for all $i,j = 1, \ldots, p$. This 
choice reduces the complexity of the model and avoids the joint estimation of $\nu$, 
which is challenging and unstable \citep{zhang2004inconsistent}. Fixing $\nu$ also allows 
the inference to focus on the scale and range parameters, mitigating identifiability issues 
between $\nu$ and the range parameter documented in the literature 
\citep[see for example][]{stein1999interpolation, kaufman2013role}.

First, we parameterize $\sigma_{i} := \exp(s_{i})$ and $\alpha_{i} := \exp(a_{i})$. Then, 
we split the parameter vector as $\theta = (\boldsymbol{\sigma}^{2}, \boldsymbol{\alpha}, 
\Delta_{B}, \mathbf{L}, \mathbf{R}_{B})$, where the vectors $\boldsymbol{\sigma}^{2} = (\sigma^{2}_{11}, 
\ldots, \sigma^{2}_{pp})$ and $\boldsymbol{\alpha} = (\alpha_{11}, \ldots, \alpha_{pp})$ 
are the parameters of the marginal random fields. This selection of blocks makes the 
parametric space of the Matérn covariance function a product of convex spaces. 
Specifically,


\begin{equation*}
\Theta = \mathbb{R}^{p} \times \mathbb{R}^{p} \times \mathbb{R}^{+}_{0} \times \text{Tri}(p) 
\times \text{CDN}(p),
\end{equation*} where \begin{equation*}
\text{Tri}(p) = \{ \mathbf{L} \in \mathbb{R}^{p \times p} : \mathbf{L} \text{ is lower 
triangular with } L_{ii} > 0 \},
\end{equation*}
\begin{equation*}
\text{CDN}(p) = \{ \mathbf{C} \in \mathbb{R}^{p \times p} : \mathbf{C} \text{ is conditionally 
negative semi-definite and } \mathrm{diag}(\mathbf{C}) = \mathbf{1} \}.
\end{equation*}

\noindent The projections needed for the block coordinate descent method exist and are detailed in Appendix A. In particular, the vectors $\boldsymbol{\sigma}^{2}$ and $\boldsymbol{\alpha}$ are estimated under independence in Step 1 of Algorithm \ref{alg:gen}, decoupling their estimation from the cross-covariance parameters and reducing the problem to a set of $p$ independent univariate optimizations.

We adapt the block coordinate descent method to solve the penalized problem 
\eqref{estimation_penalty}. This modification, known as the penalized block coordinate 
descent, updates each block by solving the corresponding penalized subproblem, i.e., by 
minimizing the objective function restricted to the block including the penalty term. We 
focus the penalization on the Cholesky matrix $\mathbf{L}$, as its off-diagonal entries encode the 
correlation structure between random fields; setting $\mathbf{L}_{ij} = 0$ for $i \neq j$ 
corresponds to zero cross-correlation between the $i$-th and $j$-th fields.

Before presenting the algorithm, we recall that the soft-thresholding operator 
$S_{\lambda}: \mathbb{R} \to \mathbb{R}$ is defined componentwise as:
\begin{equation*}
S_{\lambda}(x) = \text{sign}(x) \max(|x| - \lambda, 0),
\end{equation*}
which sets components of $x$ with absolute value smaller than $\lambda$ to zero, inducing 
sparsity in the updated block. Specifically, we propose Algorithm \ref{alg:gen} to solve this problem. Therefore, for a given $\lambda$ we obtain an estimate via
\begin{equation*}
    \hat{\theta}_{\lambda} = \argmin_{\theta \in \Theta} f_{\lambda}(\theta).
\end{equation*}

The selection of the hyperparameter $\lambda$ affects the performance of the estimation 
and so, some proposal must be given.

\begin{algorithm}
    \caption{Projected proximal gradient descent algorithm for the multivariate Matérn 
    covariance function}\label{alg:gen}
    \begin{algorithmic}[1]
        \State Estimate the marginal parameters $\theta_i = (\sigma_{ii}^2, \alpha_{ii})$ 
        under independence. That is, let $(Z_n)_{i} = (Z_i(s_1), \dots, Z_i(s_n))^{\top}$ 
        be the vector of observations for the $i$-th random field. Then define:
        \begin{equation*}
            \hat{\theta}_i = \argmin_{\theta} f((Z_n)_{i}; \theta_i).
        \end{equation*}
        \State Let $\theta = (\mathbf{L}, \Delta_B, \mathbf{R}_B)$ and define the initial parameters 
        $\theta^{(0)}$.
        \For{$m$ in $1:M$}
        \State Update $\mathbf{L}^{(m)}$ such that
        \begin{equation*}
            \mathbf{L}^{(m)} = S_{\lambda}\left(\mathbf{L}^{(m-1)} - \tau_m \dfrac{\partial f}{\partial 
            \mathbf{L}}\bigg|_{\theta^{(m-1)}} \right)
        \end{equation*}
        subject to $(\mathbf{LL}^{\top})_{ii} = \sigma_{ii}^2$. See Appendix A for the enforcement 
        of this constraint.
        \State Update $\Delta_B^{(m)}$ such that
        \begin{equation*}
            \Delta_B^{(m)} = \mathcal{P}^{+}\left(\Delta_B^{(m-1)} - \tau_m 
            \dfrac{\partial f}{\partial \Delta_B}\bigg|_{(\mathbf{L}^{(m)}, \Delta_B^{(m-1)}, 
            \mathbf{R_B}^{(m-1)})}\right)
        \end{equation*}
        subject to $\Delta_B \geq 0$, where $\mathcal{P}^{+}$ denotes projection onto 
        $\mathbb{R}^{+}_{0}$.
        \State Update $\mathbf{R}_B^{(m)}$ such that
        \begin{equation*}
            \mathbf{R}_B^{(m)} = \mathcal{P}^{\text{CDN}}\left(\mathbf{R}_B^{(m-1)} - \tau_m \dfrac{\partial f}
            {\partial \mathbf{R}_B}\bigg|_{(\mathbf{L}^{(m)}, \Delta_B^{(m)}, \mathbf{R}_B^{(m-1)})}\right)
        \end{equation*}
        subject to $0 \leq \mathbf{R}_B \leq 1$ and $\mathbf{R}_B \in \text{CDN}(p)$, where $\mathcal{P}^{\text{CDN}}$ 
        denotes projection onto $\text{CDN}(p)$.
        \EndFor
    \end{algorithmic}
\end{algorithm}

\subsection{Hyperparameter selection via information criteria}\label{lam_optimo}

Different options to choose the optimal hyperparameter can be found in the literature 
(see for example \citet{hastie2009elements, efron2004least, meinshausen2010stability}). 
Since we are dealing with a large number of spatial locations and covariates, we look for 
a strategy that can reuse the obtained computations. Therefore, we follow 
\citet{tibshirani1996regression} and propose a grid search starting from the most sparse 
case, computing the solution for different values of $\lambda$ using warm starts 
\citep{friedman2010regularization}. In this regard, we notice that the most sparse 
solution is obtained when $\lambda = \lambda_{\max}$, where

\begin{equation*}
    \lambda_{\max} = \max \left\{ \left| \frac{\partial f}{\partial l_{i,j}}
    (\theta^{(0)}) \right| : i, j = 1, \dots, p \; \text{and} \; i \neq j \right\}.
\end{equation*}

\noindent For both the composite and the maximum likelihood, $\lambda_{\max}$ is computed directly from the gradient of the objective function at the initial parameter estimate $\theta^{(0)}$, as given by the expression above. We also follow \citet{friedman2010regularization} and set $\lambda_{\min} = 10^{-8}\lambda_{\max}$.

We then generate an equally spaced sequence of $\lambda$ values on the log scale between 
$\lambda_{\min}$ and $\lambda_{\max}$. Based on empirical experience, this sequence should 
have at least $(p^2 - p)$ elements, i.e., at least as many grid points as off-diagonal 
elements in the matrix $\mathbf{L}$ to be penalized, though this is a
 rule of thumb rather than a theoretical requirement.

To select an optimal value of $\lambda$ we consider tailored information criteria 
depending on whether we use the likelihood function or the composite likelihood. In the 
case of the likelihood function, we use the Akaike Information Criterion, given by:

\begin{equation*}
    \text{AIC}(\theta) = -2\ell(\theta) + 4 \left|\{(i,j): i,j=1,\dots,p \;\; 
    \text{and} \;\; (\Psi)_{ij} \neq 0\}\right|,
\end{equation*}
where the factor of $4$ rather than the standard $2$ accounts for the symmetry 
of $\Psi$: since $\Psi_{ij} = \Psi_{ji}$, each free off-diagonal parameter 
$L_{ij}$ ($i > j$) contributes two nonzero entries to the set, so the cardinality 
above is twice the number of free parameters being penalized. For the composite likelihood function $CL(\theta)$, we implement the composite likelihood 
information criterion \citep[CLIC,][]{varin2005note, bevilacqua2016composite}, given by:
\begin{equation}
    \text{CLIC}(\theta) = -pl_{p}(\theta) - tr\left(J_{p}(\theta) 
    H_{p}^{-1}(\theta)\right) \label{clic},
\end{equation}
where $J_p(\theta)$ and $H_p(\theta)$ are matrices defined in Appendix 
\ref{Eq:H and J definition}.

The validity of \eqref{clic} requires that the composite likelihood estimator be 
consistent and asymptotically normal, as established in \citet{varin2005note}. In the 
appendix, we verify that these conditions hold for the estimator proposed in Section 
\ref{propuesta}. Calculating the CLIC requires consistent estimation of the matrices 
$H_p(\theta)$ and $J_p(\theta)$. This can be achieved using plug-in estimates 
$H_p(\hat{\theta}_{\lambda})$ and $J_p(\hat{\theta}_{\lambda})$. However, computing 
this last expression has a computational complexity of $\mathcal{O}((np)^4)$, making it 
infeasible for large datasets. To estimate $J_p(\hat{\theta}_{\lambda})$, we extend the 
subsampling method described in \citet{heagerty2000window} to the multivariate case. 
Under the assumption that $W^{-1} J_p(\hat{\theta}_{\lambda})$ converges to a matrix 
$J^*_p$ as $n \to \infty$, where $W$ is the sum of the weights involved in the 
estimation, we use the subsampling method to obtain an estimate $\hat{J}^*_{p}$ of 
$J^*_p$ and then estimate $J_p(\hat{\theta}_{\lambda})$ via $W\hat{J}^*_{p}$. Given 
subsets $S_1, \ldots, S_M$ of $\{1, 2, \ldots, n\}$, associated with spatial locations 
$\{s_1, s_2, \ldots, s_n\}$, the estimator is defined as:
\begin{equation*}
    \hat{J}^*_{p}(\theta) = \frac{1}{M} \sum_{m=1}^M \frac{1}{W^{(m)}_p} 
    \sum_{\substack{(k,l) \in S_m \\ (k',l') \in S_m}} [\nabla cl^p_{kl}(\theta)] 
    [\nabla cl^{p}_{k'l'}(\theta)]^\top w_{kl} w_{k'l'},
\end{equation*}
where $W^{(m)}_p = \sum_{(k,l) \in S_m} w_{kl}$, and pairs $(k,l)$ and $(k',l')$ both 
belong to $S_m$.

On the other hand, calculating the matrix $H_p(\theta)$ has a computational complexity 
of $\mathcal{O}(p^4 n^2)$, which also incurs high computational costs growing rapidly 
as solutions become denser. In cases with a large number of point pairs, we propose 
adjusting the weights to reduce the number of point pairs and/or taking a bootstrap 
subsample to compute $H_p(\theta)$.

\section{Simulation study}

We conduct three simulation studies to analyze the behavior of the penalized estimation method presented in Section \ref{propuesta}. We first illustrate its capability to identify zero correlations in the matrix $\Sigma$. Next, we will conduct a simulation study to compare the performance of the estimate obtained using likelihood and composite likelihood functions. Finally, we conduct a simulation study using the composite likelihood function and assess the behavior of the hyperparameter selection via CLIC. All computations presented in this section were performed on a 2x 8-Core CPU (Intel Xeon E5-2689 @ 1200-3600 MHz), with 128 GB RAM, 931 GB storage, and a 5.4.0-182-generic x86 64 kernel.

\subsection{Illustrative Simulation}\label{Simul_ilust}

This section illustrates how to implement the proposed estimation method. We simulate a random field of dimension $p = 5$ across $n = 500$ spatial locations. For the latter, we consider 500 samples from a uniform distribution over the square $ [0,1] $. Regarding the parameters of the random field, we set:
\begin{equation*}
	\sigma^2 = (0.5, 1.0, 1.5, 2.0, 2.5) \quad \text{and} \quad \alpha = (10, 6.67, 5, 4, 3.33).
\end{equation*}
Additionally, we define the matrix $ \mathbf{R} $ as:
\begin{equation*}
	\mathbf{R}_{ij} = \begin{cases}
		1, & \text{if } i = j, \\
		0.5, & \text{if } |i - j| = 1, \\
		0, & \text{otherwise}.
	\end{cases}
\end{equation*} 
We then set the parameters $ \rho = \mathbf{R}$, $ \mathbf{R}_B = \mathbf{R}$, and $ \Delta_B = 60 $.

\textbf{Remark:}	Although in this example we set $ \rho = \mathbf{R}_B = \mathbf{R} $, in general $ \rho \neq \mathbf{R}_B $, and therefore these parameters are typically estimated separately.

With the Matérn model parameters set, we simulate a random field using the Cholesky decomposition method. This involves simulating $ X_{n} \sim N(0, I_{np})$ and then defining our sample as $ Z_n = \text{Chol}(\Sigma(\theta)) X_{n}$. This procedure generates a sample from a multivariate random field with a 
Matérn covariance model, as shown in Figure \ref{Campos_Simul}. Figure \ref{Variograma_simul} displays the cross-variograms between the different random field components. Note that some cross-variograms exhibit an irregular behavior; for instance, the cross-variograms of variables $ (Z_n)_{1} $ and $ (Z_n)_{5}$, $ (Z_n)_{2}$ and $ (Z_n)_5 $, or $ (Z_n)_{2}$ and $ (Z_n)_4 $. This irregularity occurs because these components are uncorrelated, so their theoretical variogram should be constant and identically equal to zero.

\begin{figure}[H]
	\centering
	\begin{minipage}{.5\textwidth}
		\centering
		\includegraphics[scale=0.25]{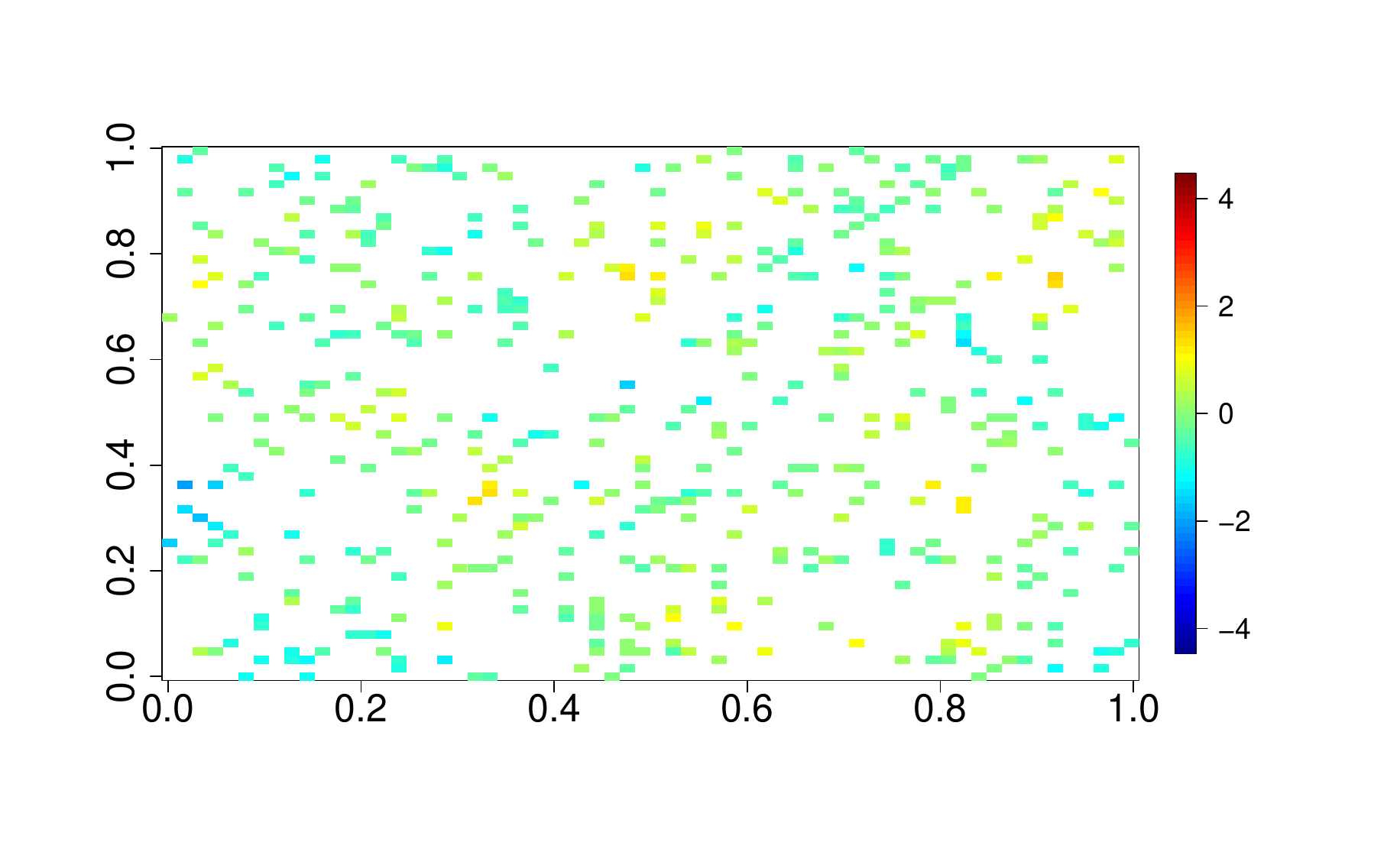}
	\end{minipage}%
	\begin{minipage}{.5\textwidth}
		\centering
		\includegraphics[scale=0.25]{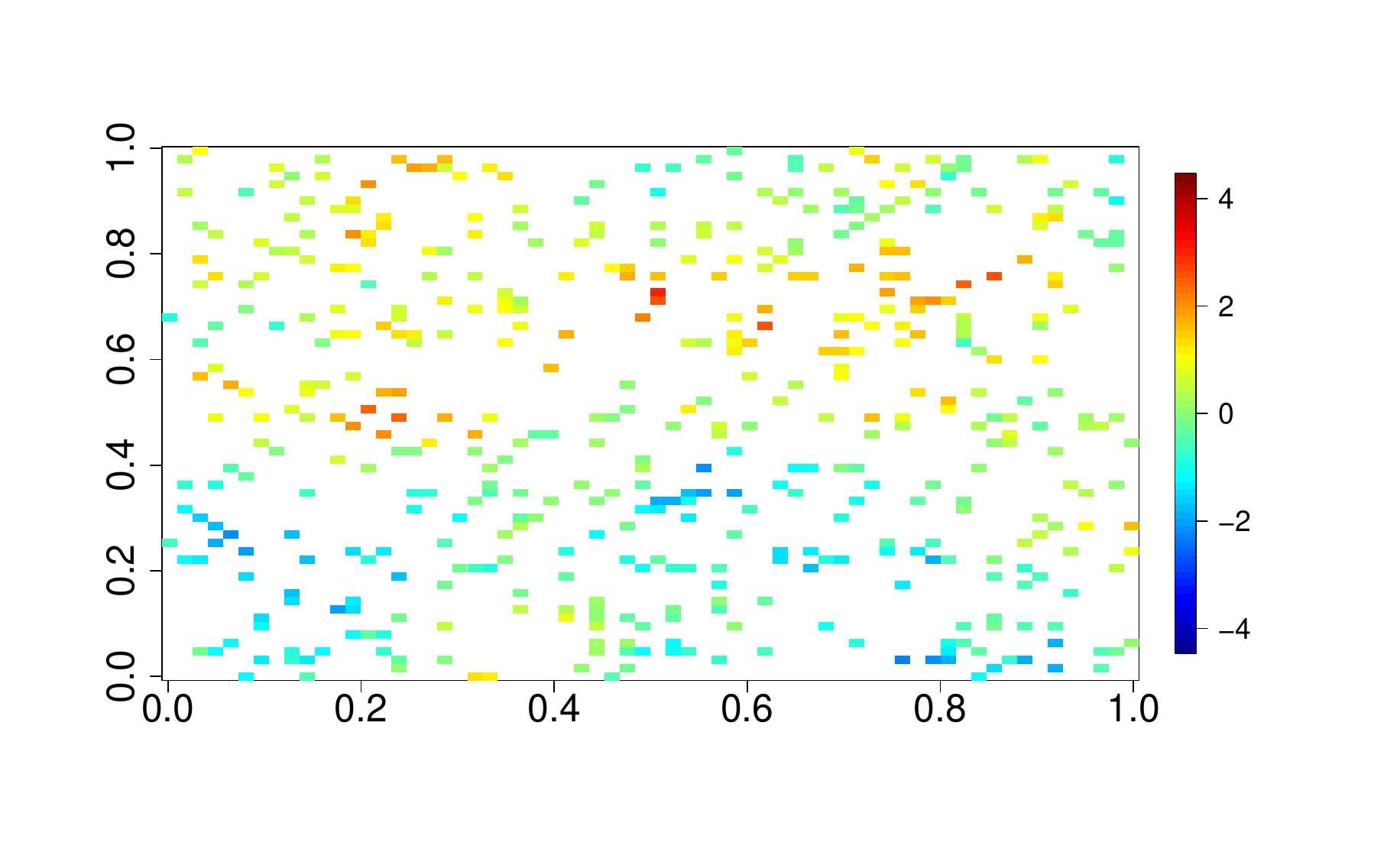}

	\end{minipage}
	
	\begin{minipage}{.5\textwidth}
		\centering
		\includegraphics[scale=0.25]{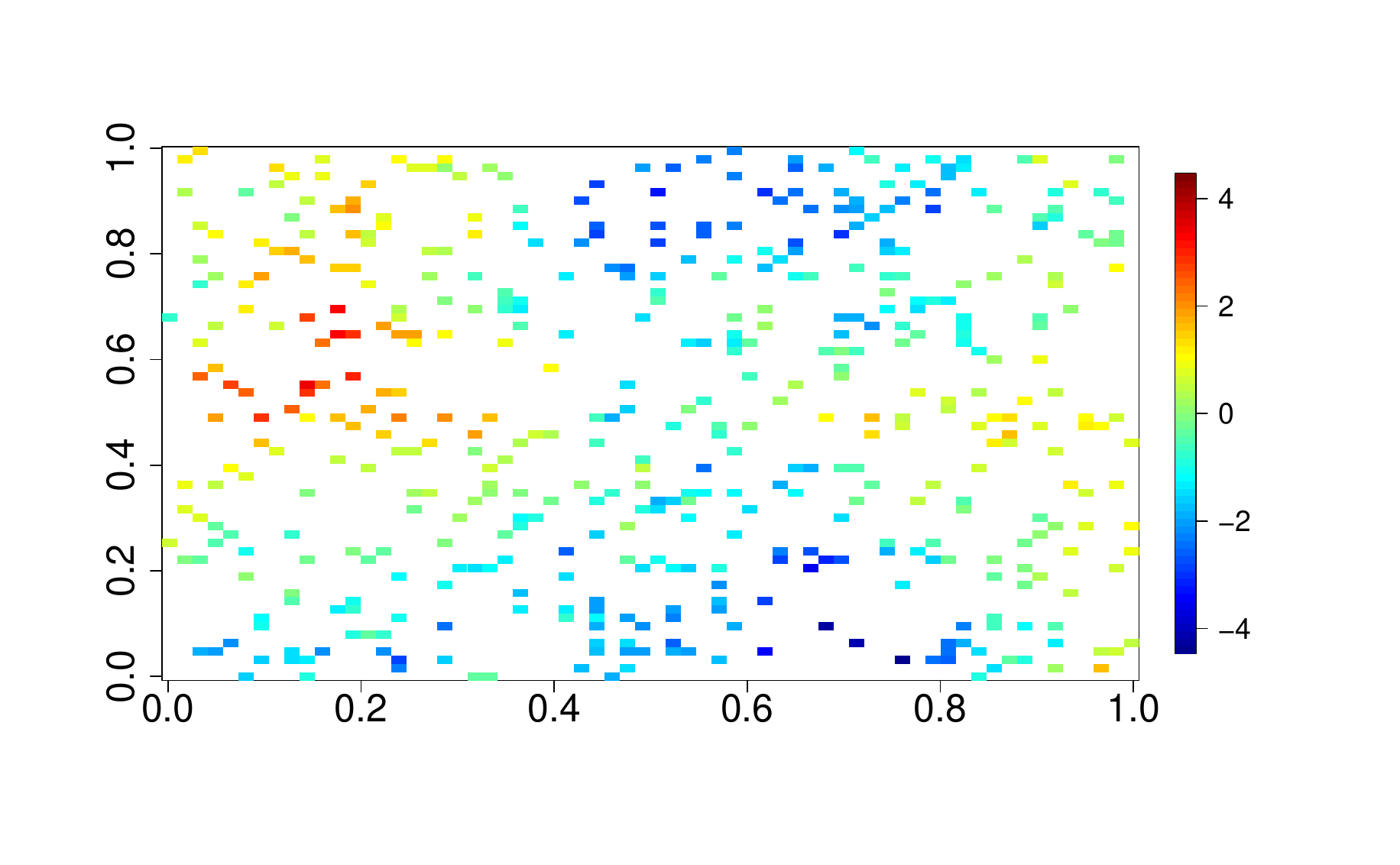}
	\end{minipage}%
	\begin{minipage}{.5\textwidth}
		\centering
		\includegraphics[scale=0.25]{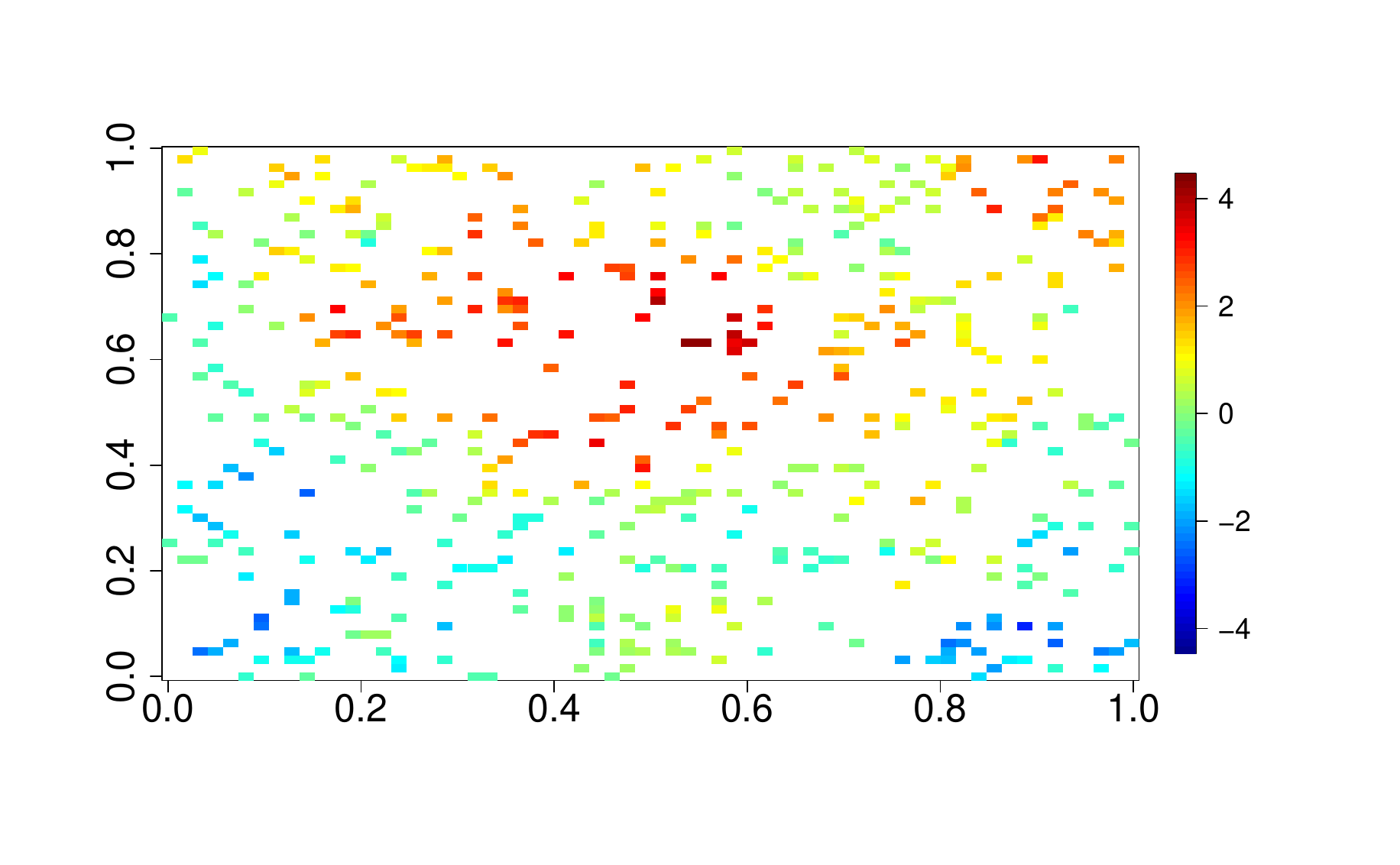}
	\end{minipage}
	
	\begin{minipage}{.5\textwidth}
		\centering
		\includegraphics[scale=0.25]{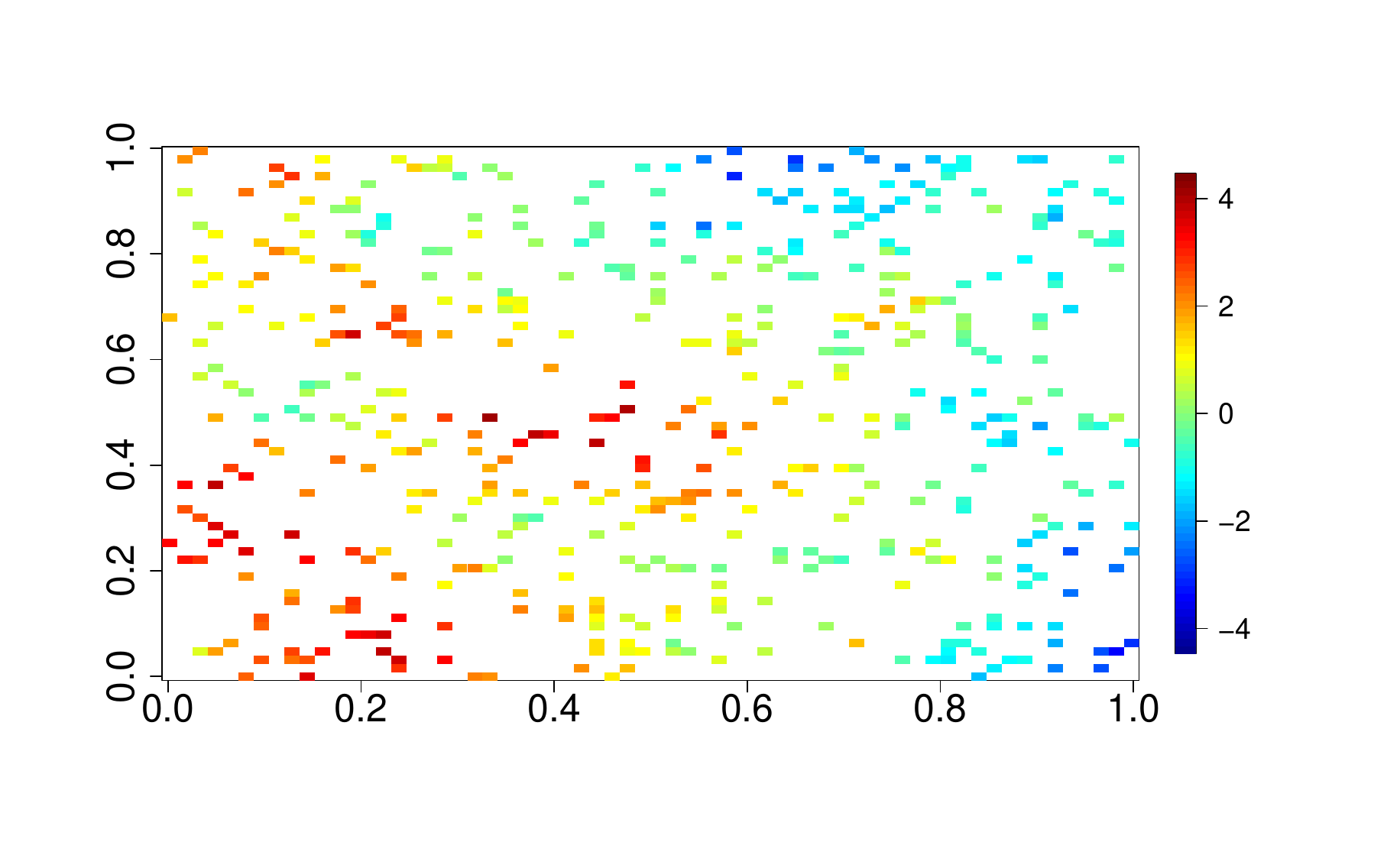}
	\end{minipage}
	
	\caption{Visualization of the simulated random field components.}\label{Campos_Simul}
\end{figure}

\begin{figure}[H]
	\centering 
	\includegraphics[scale=0.5]{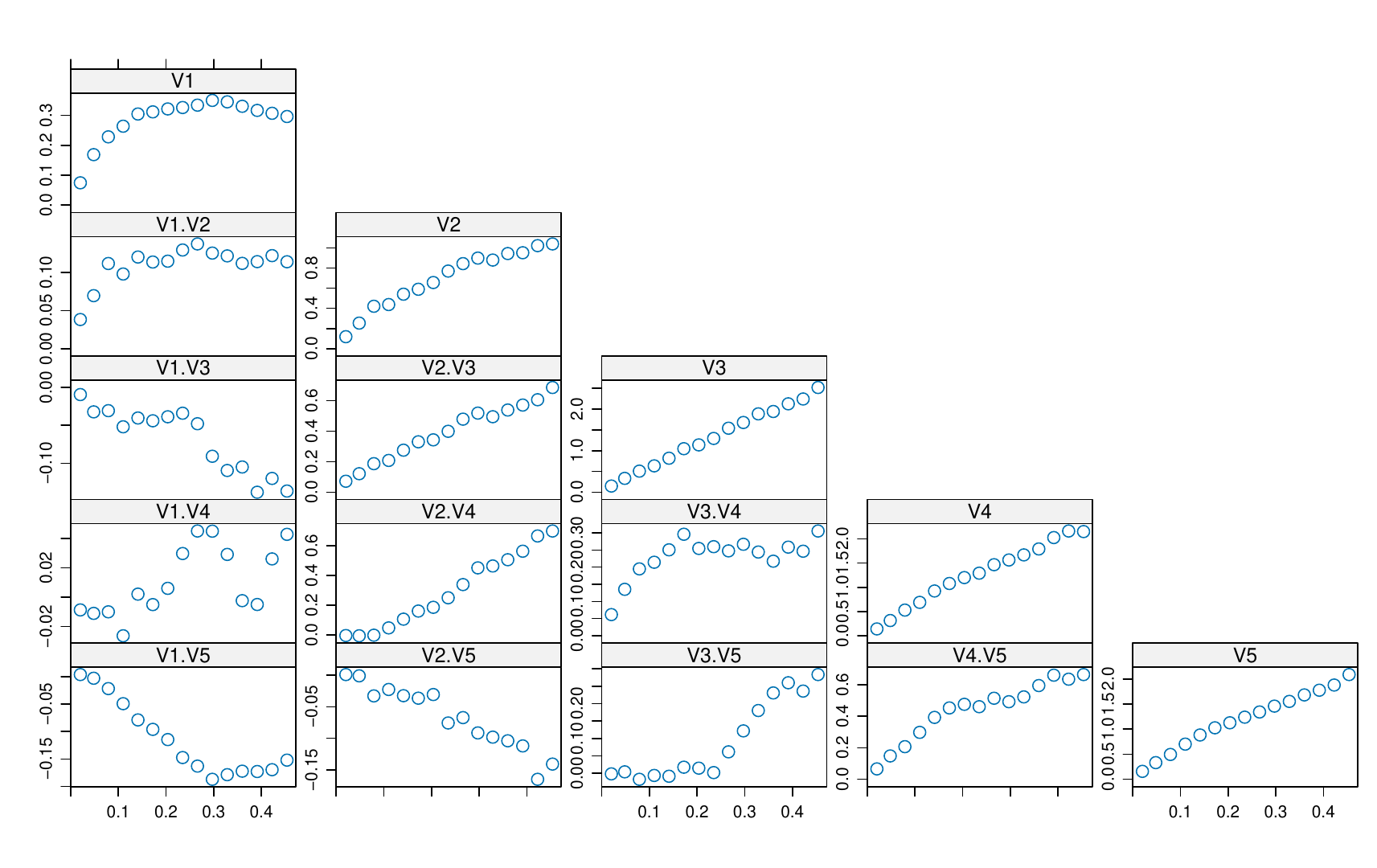}
	\caption{Cross-variograms of the simulated random field components.}
	\label{Variograma_simul}
\end{figure}

We then estimate the parameters of the multivariate Matérn model using maximum likelihood. To generate the sequence of $ \lambda$ values, we suggest using a number of $ \lambda$ values that is twice the number of correlations to be identified. For example, for a 5-dimensional random field, there are 10 correlations to identify, so we recommend a sequence of $ \lambda $ values with at least 20 elements. For this example, we use a sequence of $ \lambda$ with 100 elements to obtain a smoother graph of the AIC criterion with respect to $ \lambda $.

In Figure \ref{Sparse_simul}, we plot the percentage of zeros in the matrices $ L$ and $ \Psi $ with respect to $ \lambda$. Note that a greater number of zeros in the matrix $ L$ also implies a greater number of zeros in the matrix $ \Psi $, which is the objective of the penalization. Additionally, in Figure \ref{lineas_sparse}, we plot the solutions for the different values of $ L_{ij}$ and $ \Psi_{ij}$ with $ i \neq j$. We observe that the solutions for the parameters that are truly equal to zero decay very quickly to zero as $ \lambda $ increases, in contrast to the other parameters, which require much higher values of $ \lambda $ to become zero. These two groups of parameters are easily identifiable from the graph and correspond exactly to the two groups of parameters we sought to identify.

As mentioned previously, to determine the optimal penalty, we compute the CLIC criterion with respect to $\lambda$, as shown in Figure \ref{AIC_simul}. The value of $ \lambda $ with the lowest AIC, marked in red, is chosen as the optimal $ \lambda $.

After selecting the optimal $ \lambda $, we use the parameters $ \hat{\theta}_{\lambda}$ corresponding to the chosen $ \lambda $ for spatial prediction. For this, we implement ordinary cokriging on a $ 100 \times 100 $ grid. The results 
are displayed in Figure \ref{simul_pred_var}. Additionally, to illustrate how the errorvariances behave relative to the observations, we plot the cokriging variances of the third random field component, along with the observation locations, as shown in Figure \ref{variances_Kriging}.

In this example, we used the likelihood function, but the procedure is analogous if the composite likelihood function is used for estimation. The decision on which function to use depends on the estimation time required by each function. Generally, the likelihood function has a higher computational complexity, thus requiring more time for estimation. Specifically, Table \ref{tab:comp_time} presents the estimation times required using the likelihood and composite likelihood functions, respectively.

\begin{figure}[H]
	\begin{minipage}{.5\textwidth}
		\centering
		\includegraphics[scale=0.3]{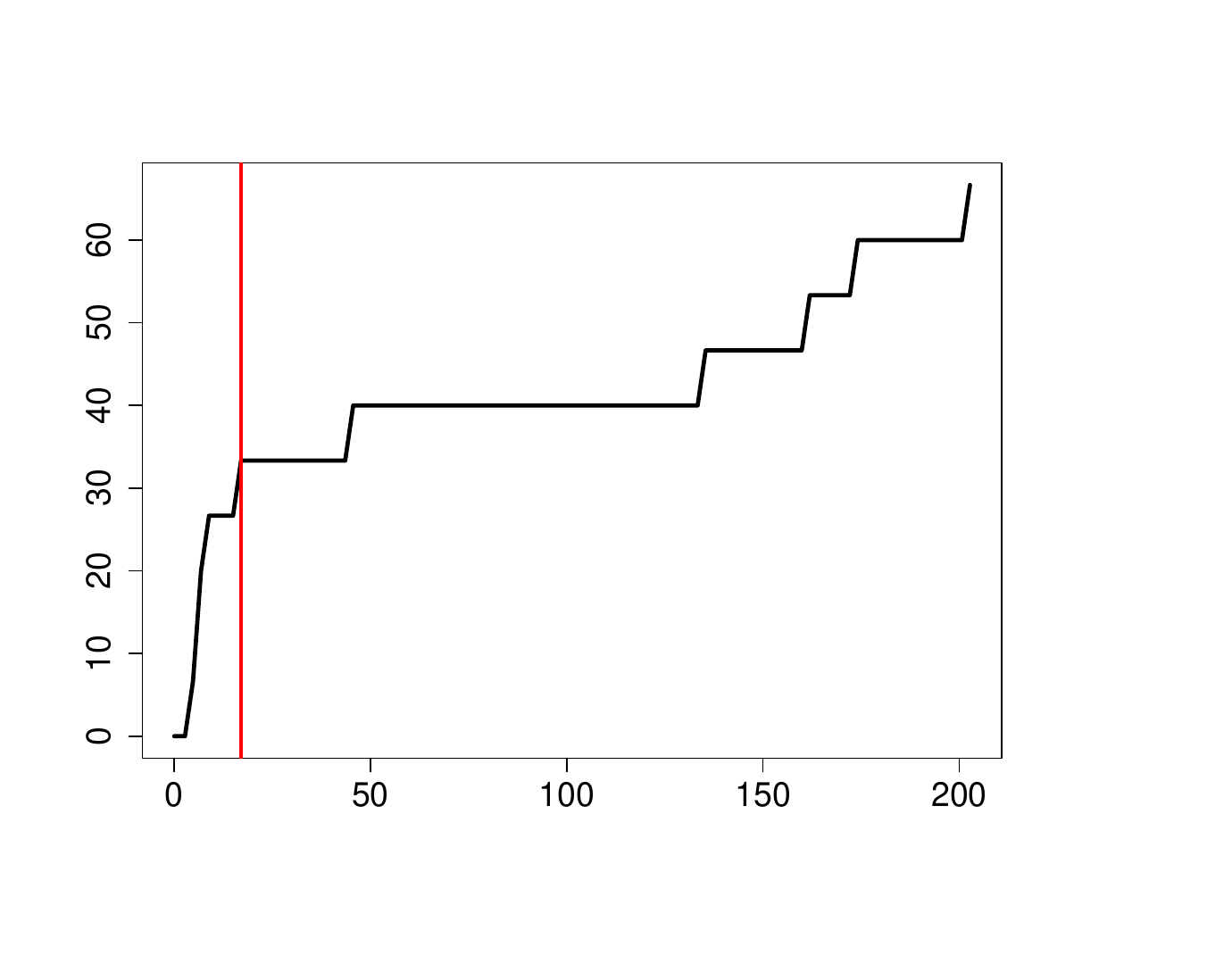}
	\end{minipage}%
	\hspace{0.5cm}
	\begin{minipage}{.5\textwidth}
		\centering
		\includegraphics[scale=0.3]{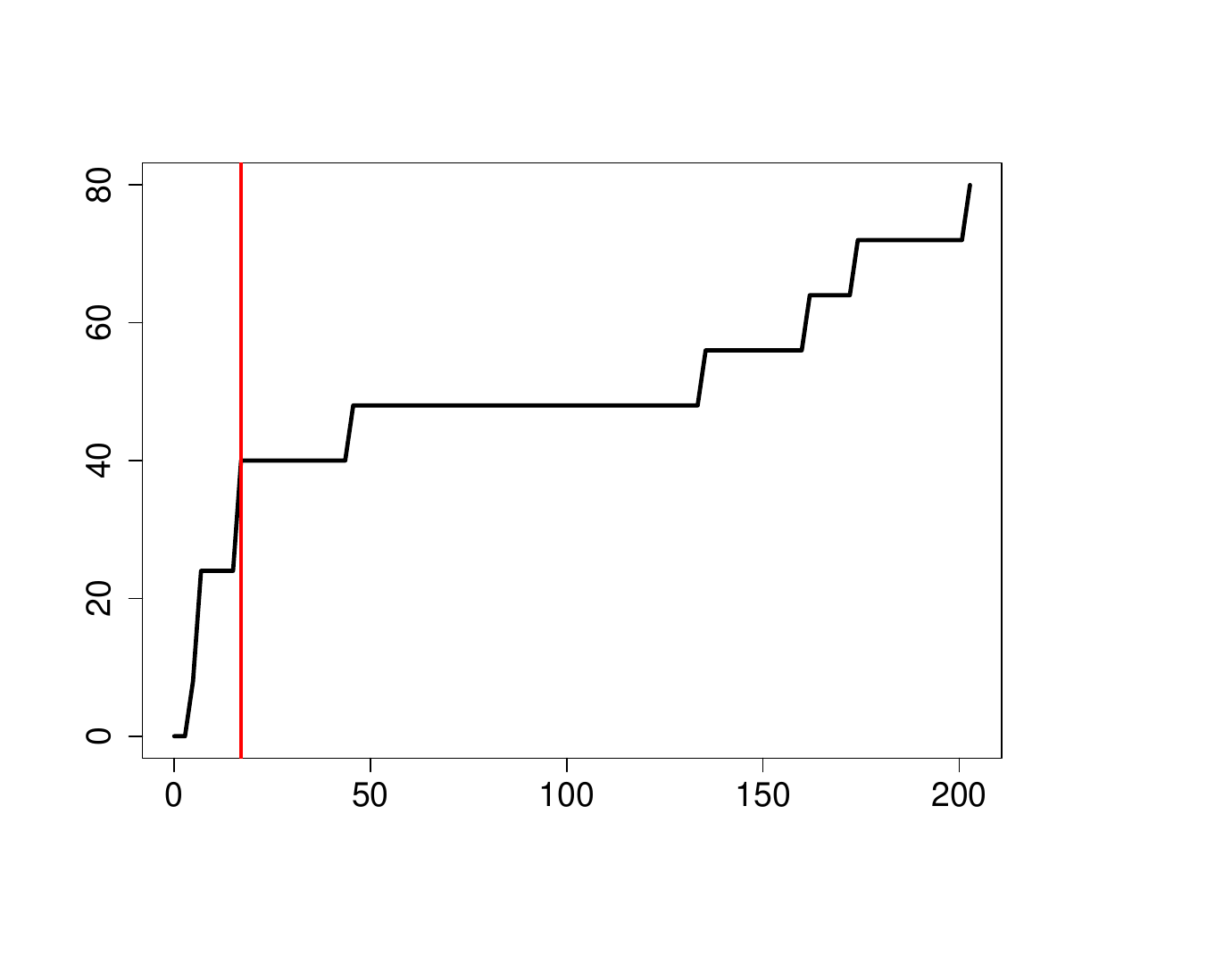}
	\end{minipage}
	\caption{Percentage of zeros with respect to $\lambda$. Left: Matrix $ L $. Right: Matrix $ \Psi $.}\label{Sparse_simul}
\end{figure}

\begin{figure}[H]
	\begin{minipage}{.5\textwidth}
		\centering
		\includegraphics[scale=0.24]{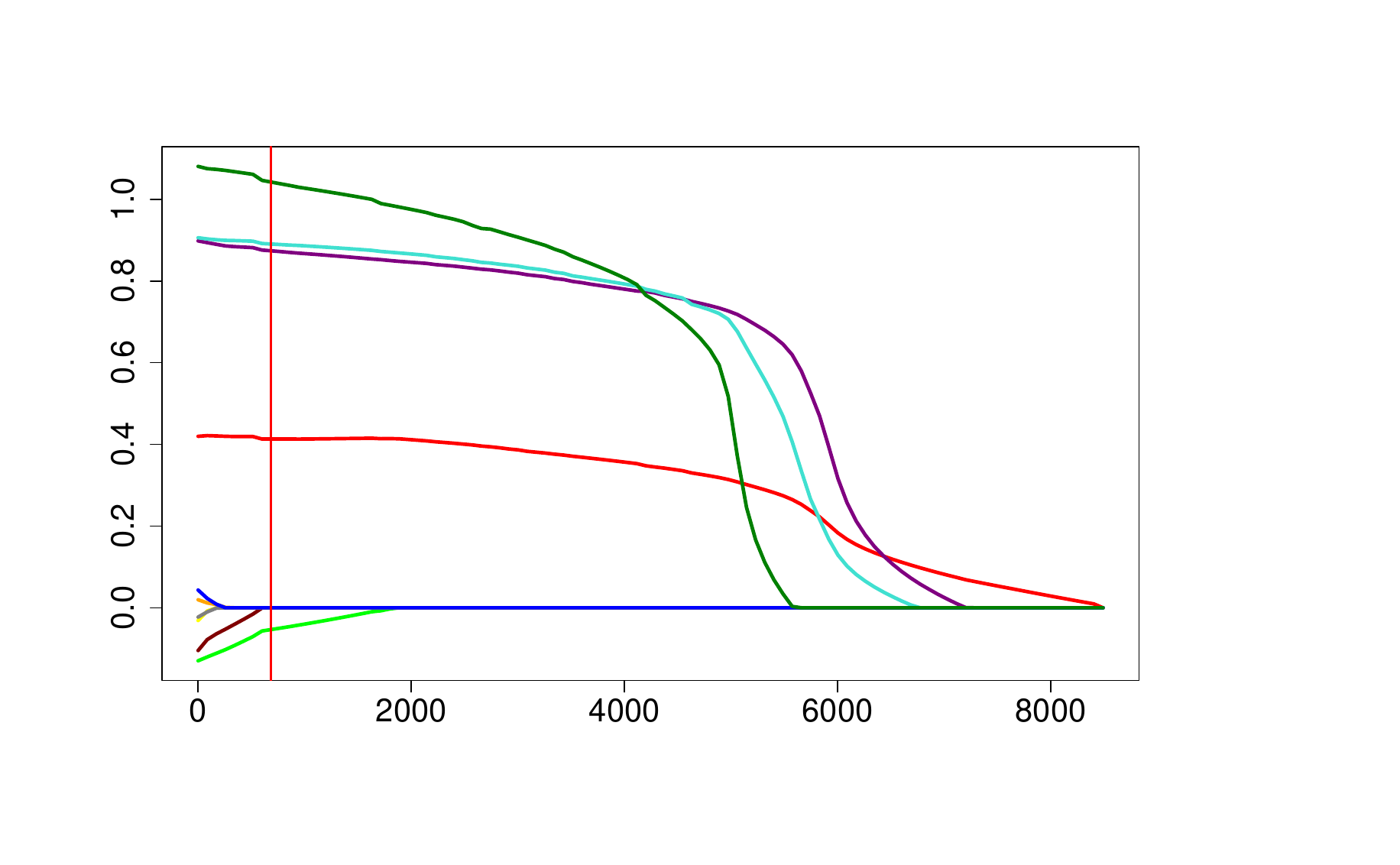}
	\end{minipage}%
	\hspace{0.5cm}
	\begin{minipage}{.5\textwidth}
		\centering
		\includegraphics[scale=0.24]{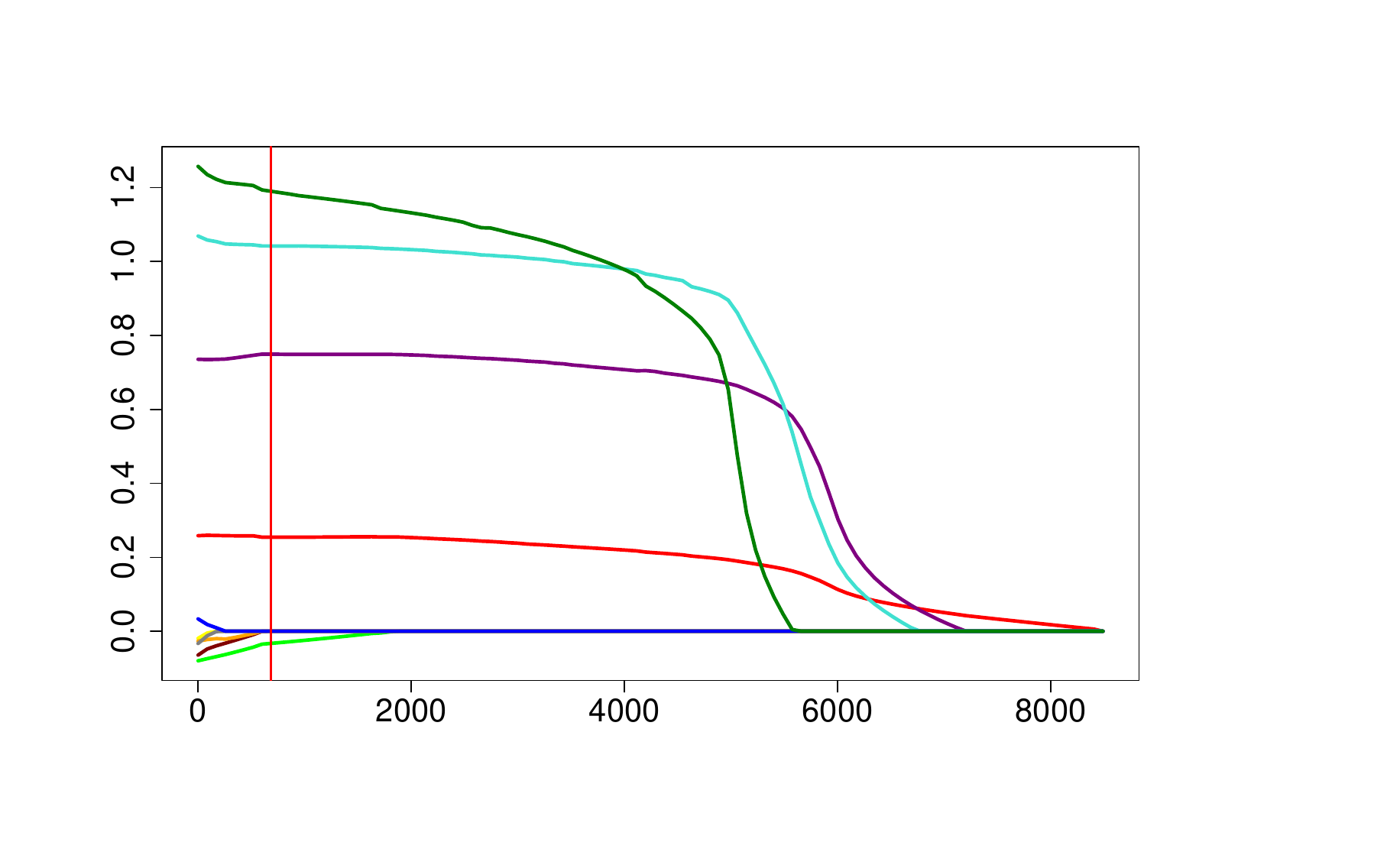}
	\end{minipage}
	\caption{Solution paths with respect to $ \lambda$. Each colored line represents an entry of the matrix. Left: Matrix $ L $. Right: Matrix $ \Psi$.  }\label{lineas_sparse}
\end{figure}

\begin{figure}[H]
	\centering
	\includegraphics[scale=0.35]{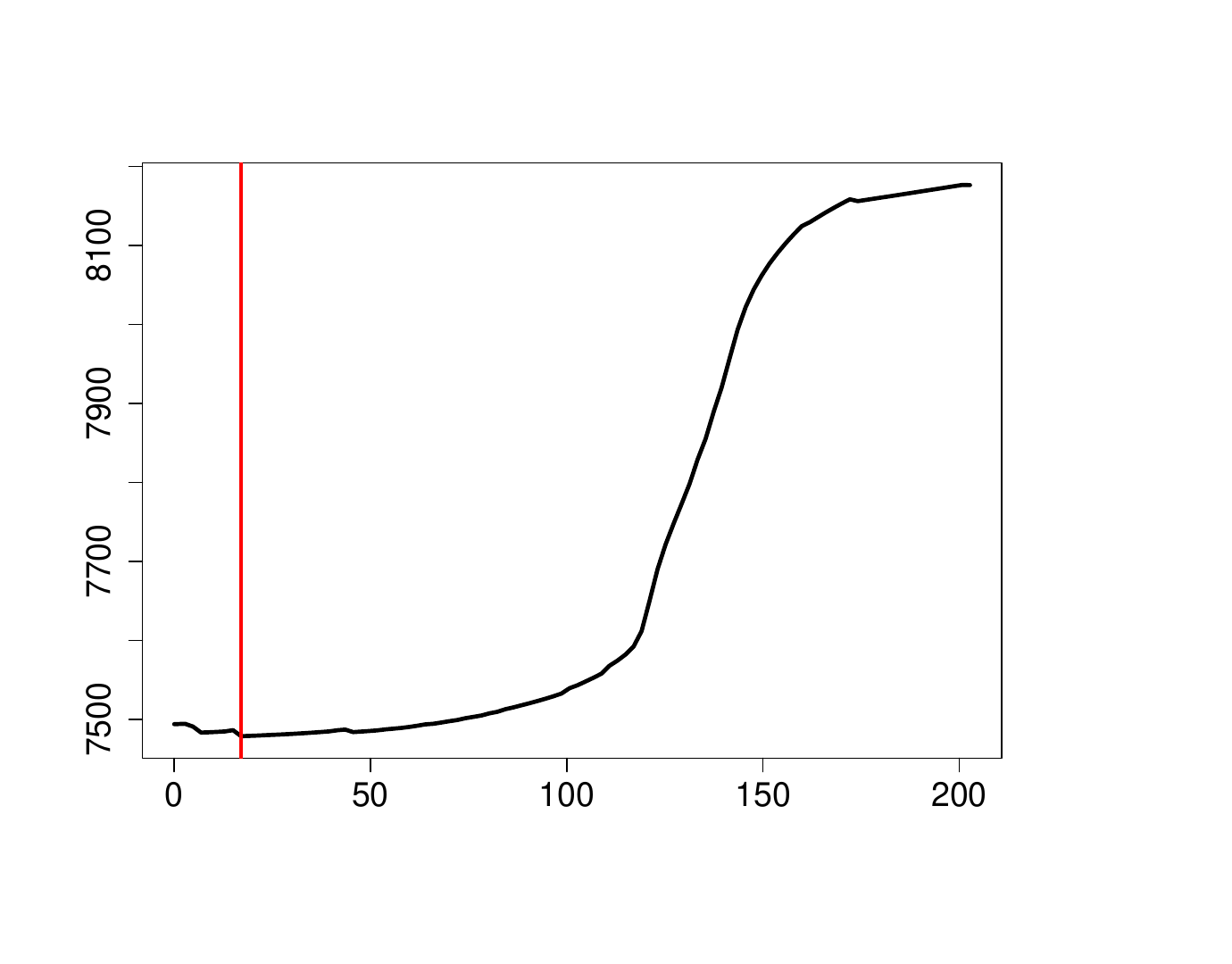}
	\caption{AIC with respect to $\lambda$. In red, the optimal $ \lambda$ value.}\label{AIC_simul}
\end{figure}

\begin{figure}[H]
	\centering
	\begin{minipage}{.5\textwidth}
		\centering
		\includegraphics[scale=0.25]{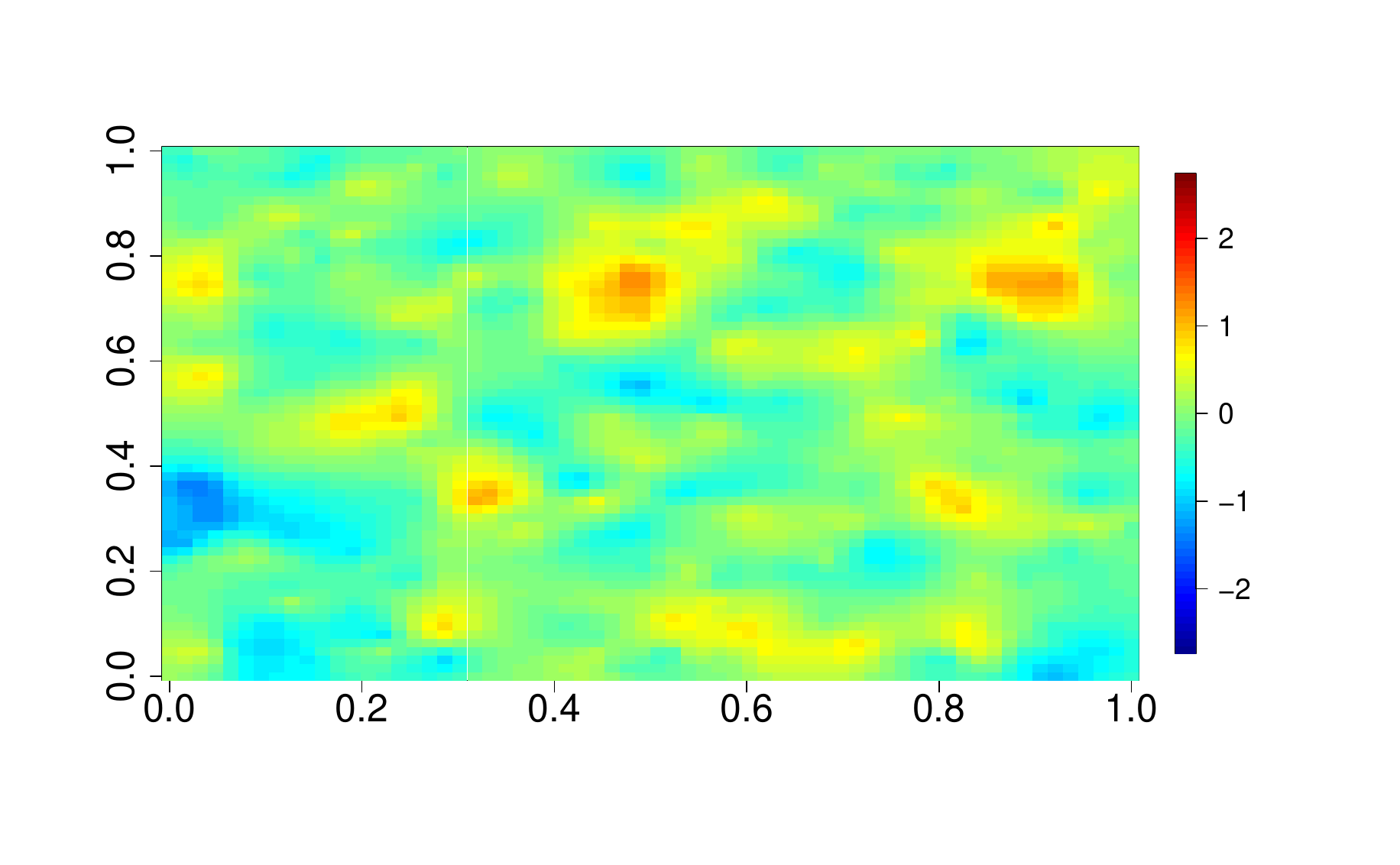}
	\end{minipage}%
	\begin{minipage}{.5\textwidth}
		\centering
		\includegraphics[scale=0.25]{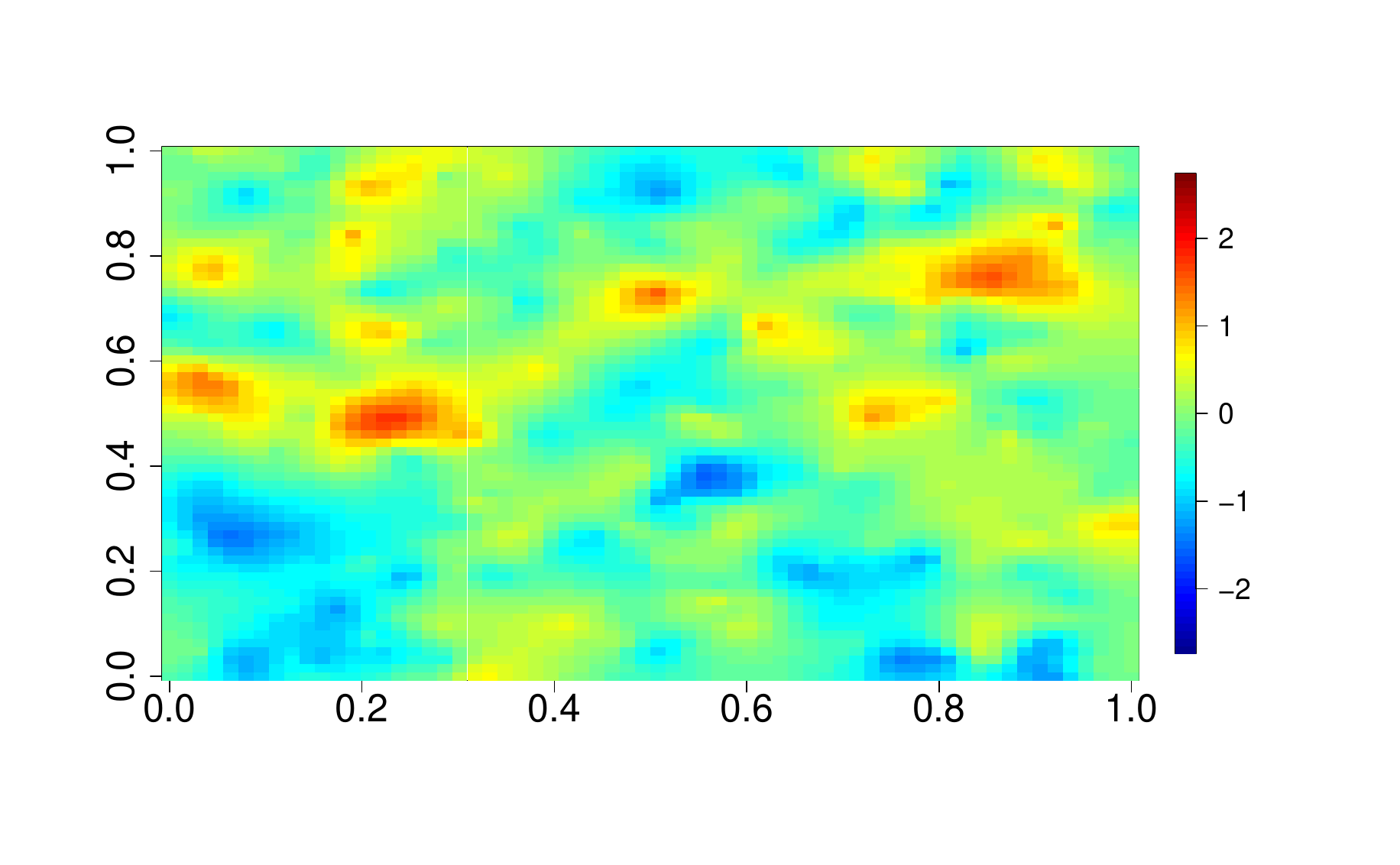}

	\end{minipage}
	
	\begin{minipage}{.5\textwidth}
		\centering
		\includegraphics[scale=0.25]{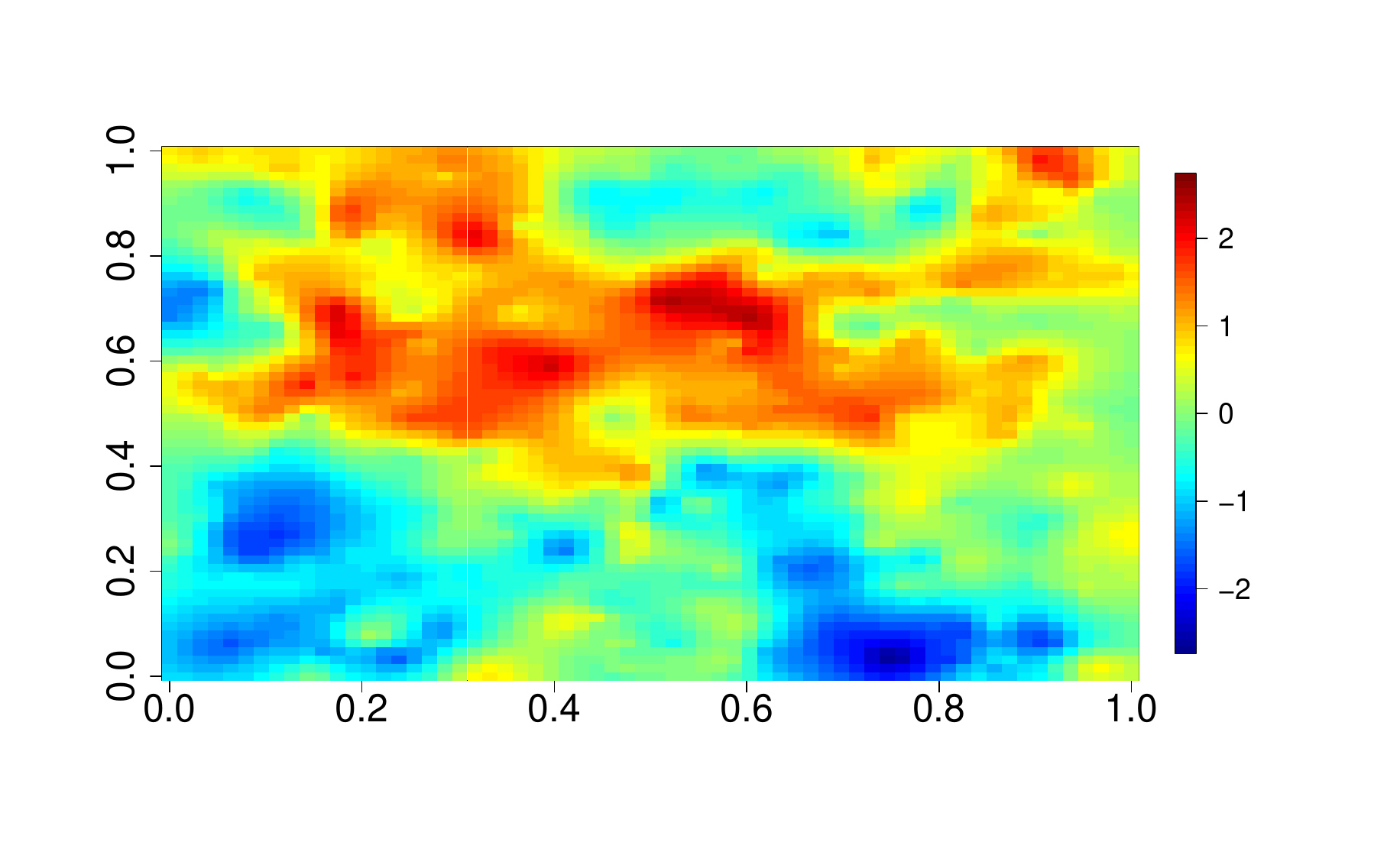}
	\end{minipage}%
	\begin{minipage}{.5\textwidth}
		\centering
		\includegraphics[scale=0.25]{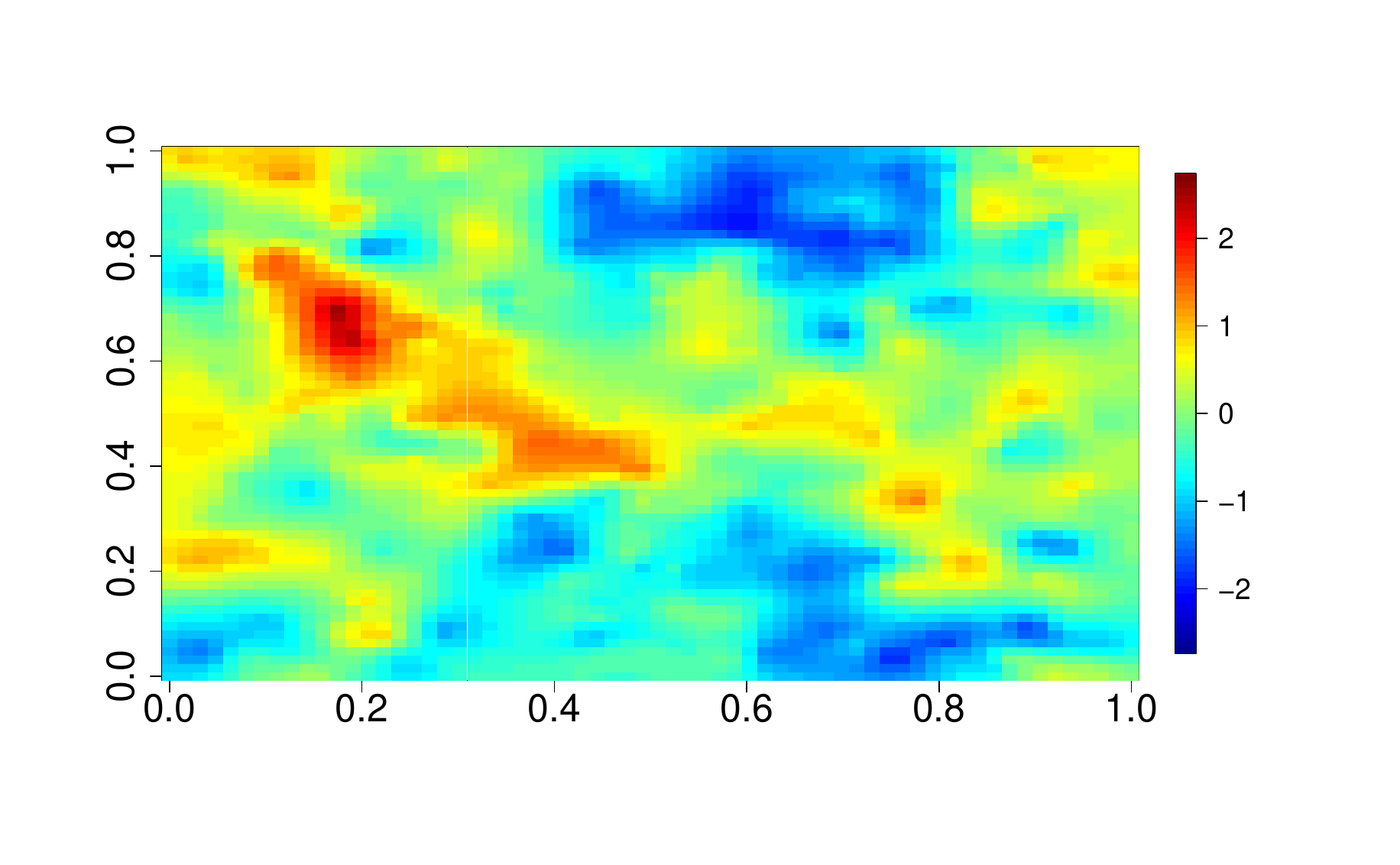}
	\end{minipage}
	
	\begin{minipage}{.5\textwidth}
		\centering
		\includegraphics[scale=0.25]{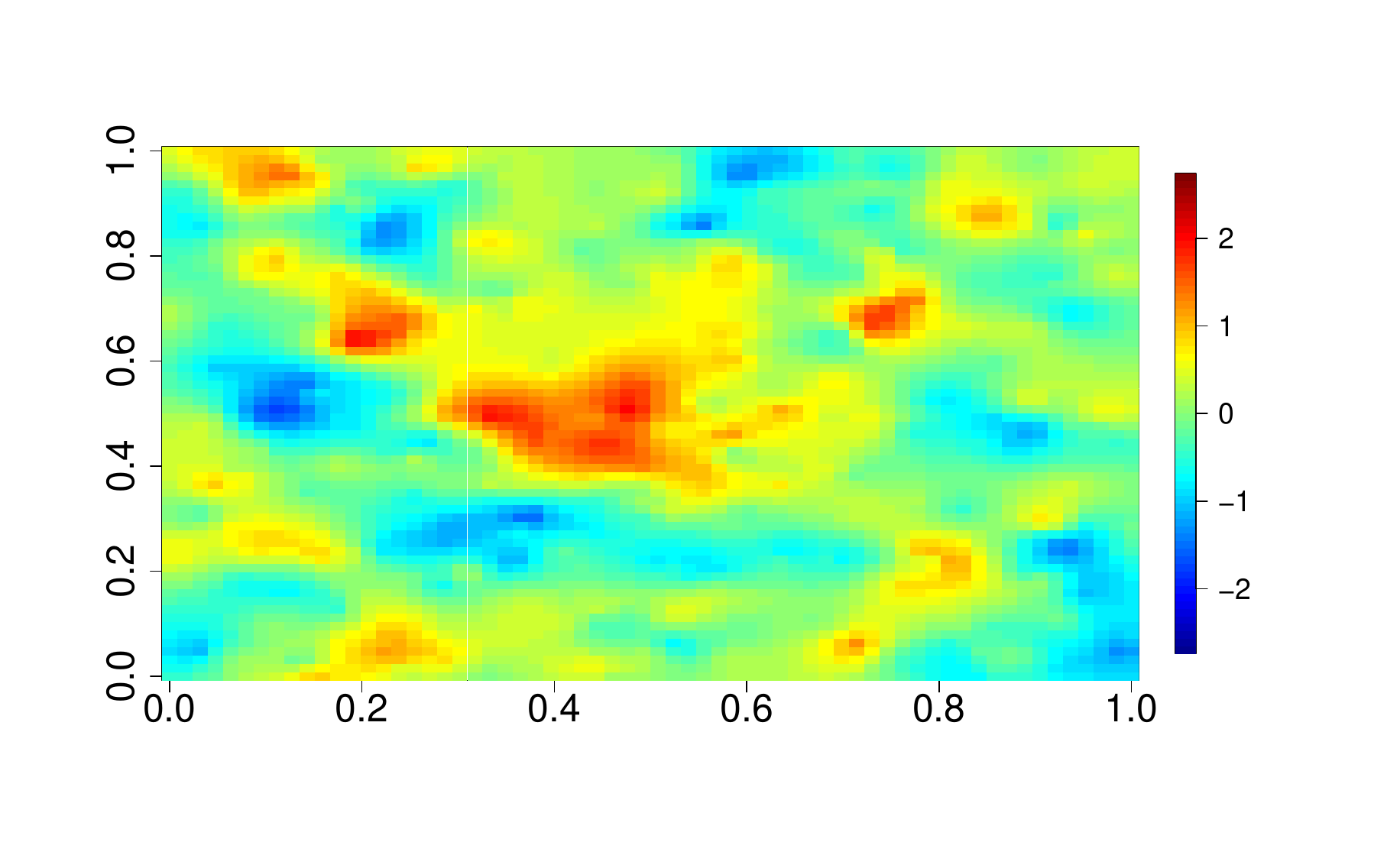}	\end{minipage}
	
	\caption{Cokriging predictions for the 5 random field components.}
	\label{simul_pred_var}
\end{figure}

\begin{figure}[H]
	\centering
	\includegraphics[scale=0.4]{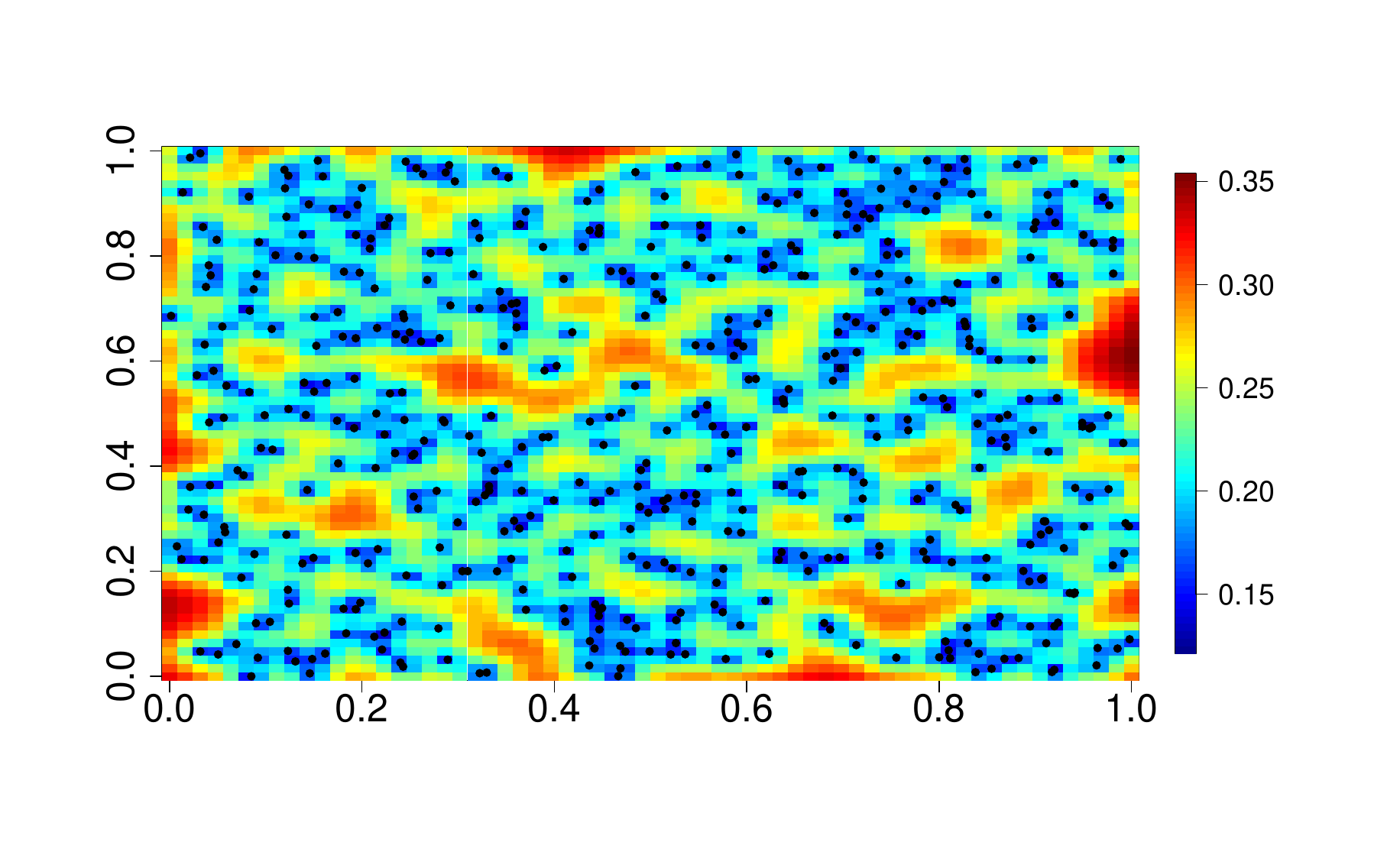}
	\caption{Cokriging variance for the third random field component.}\label{variances_Kriging}
\end{figure}

\subsection{Simulation Experiment}

The goal of this experiment is to computationally determine whether the LASSO penalty can detect zero coefficients in the correlation matrix, as well as to compare estimation results obtained using the likelihood and composite likelihood functions.

To achieve this, we chose the same parameterization as in Section \ref{Simul_ilust} to simulate 500 multivariate random fields. For each simulated random field, the covariance parameters were estimated using the likelihood and composite likelihood methods, both with and without penalization. Specifically, for the LASSO estimation, a sequence of 20 values for $\lambda$ was generated, and the optimal $\lambda $  was selected using the methods described in Section \ref{lam_optimo}.

\begin{figure}[H]
	\centering
	\begin{minipage}{.5\textwidth}
		\centering
		\includegraphics[scale=0.2]{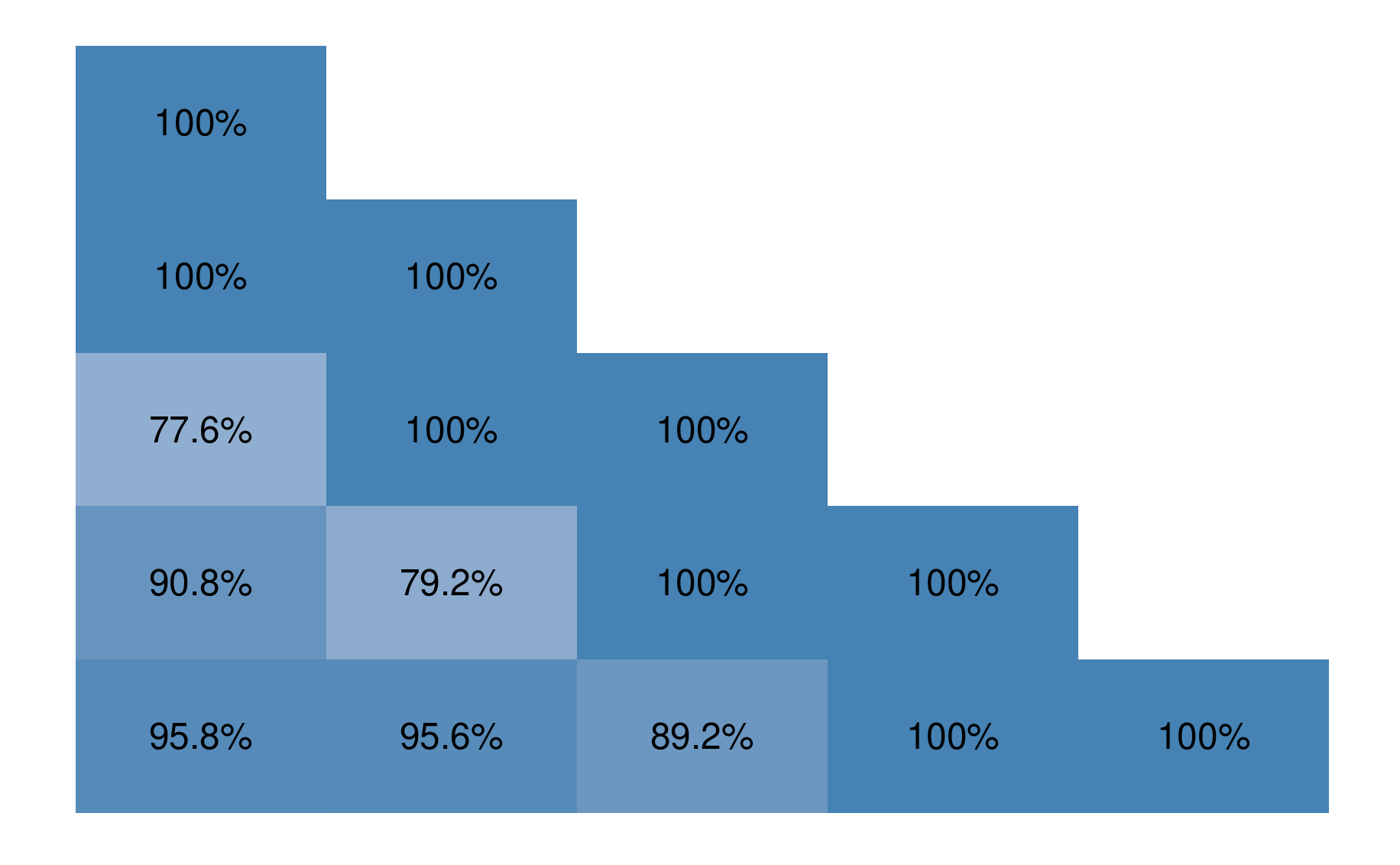}
        (a)
	\end{minipage}%
	\begin{minipage}{.5\textwidth}
		\centering
		\includegraphics[scale=0.2]{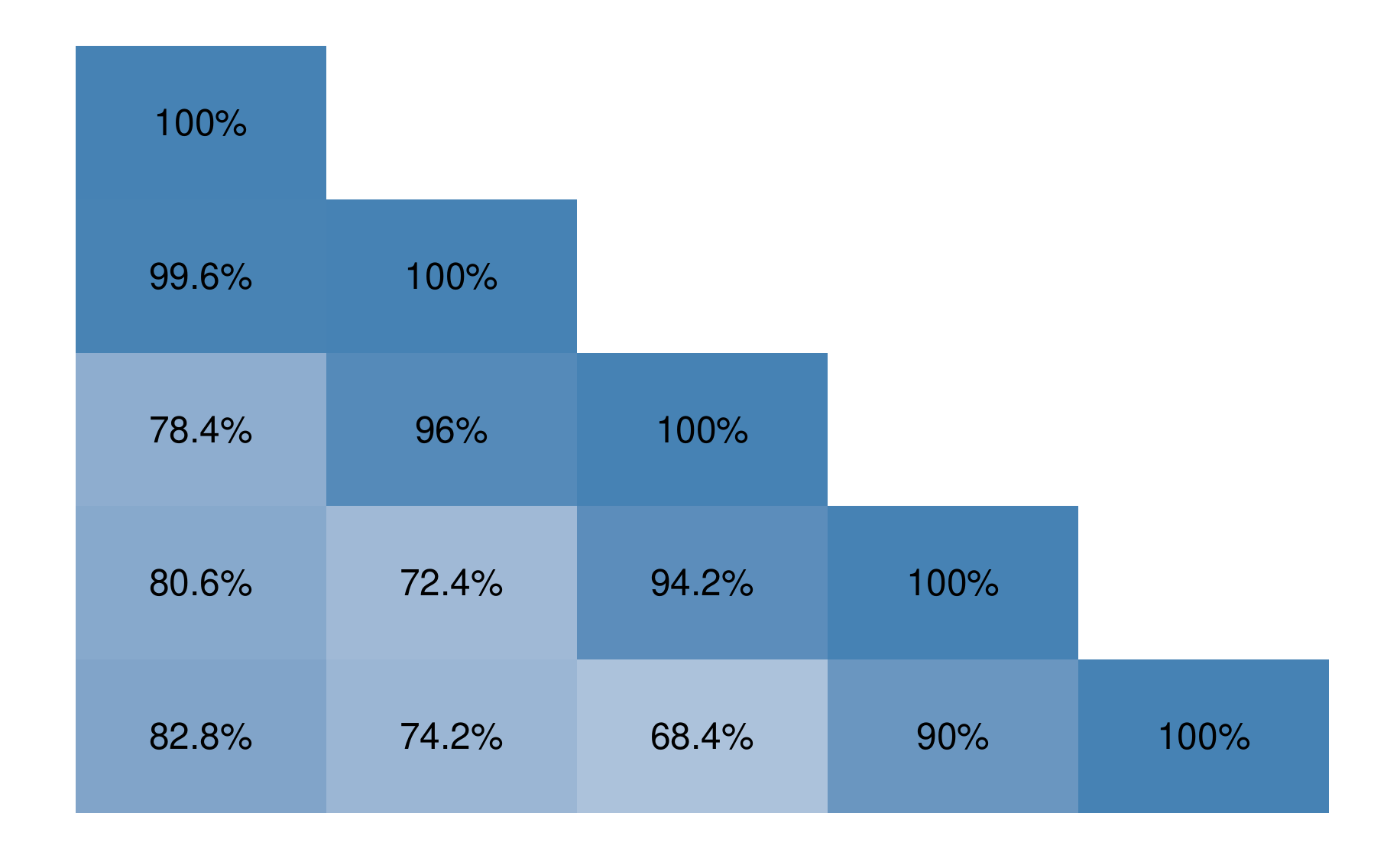}
		(b)
	\end{minipage}
	\begin{minipage}{.5\textwidth}
		\centering
		\includegraphics[scale=0.2]{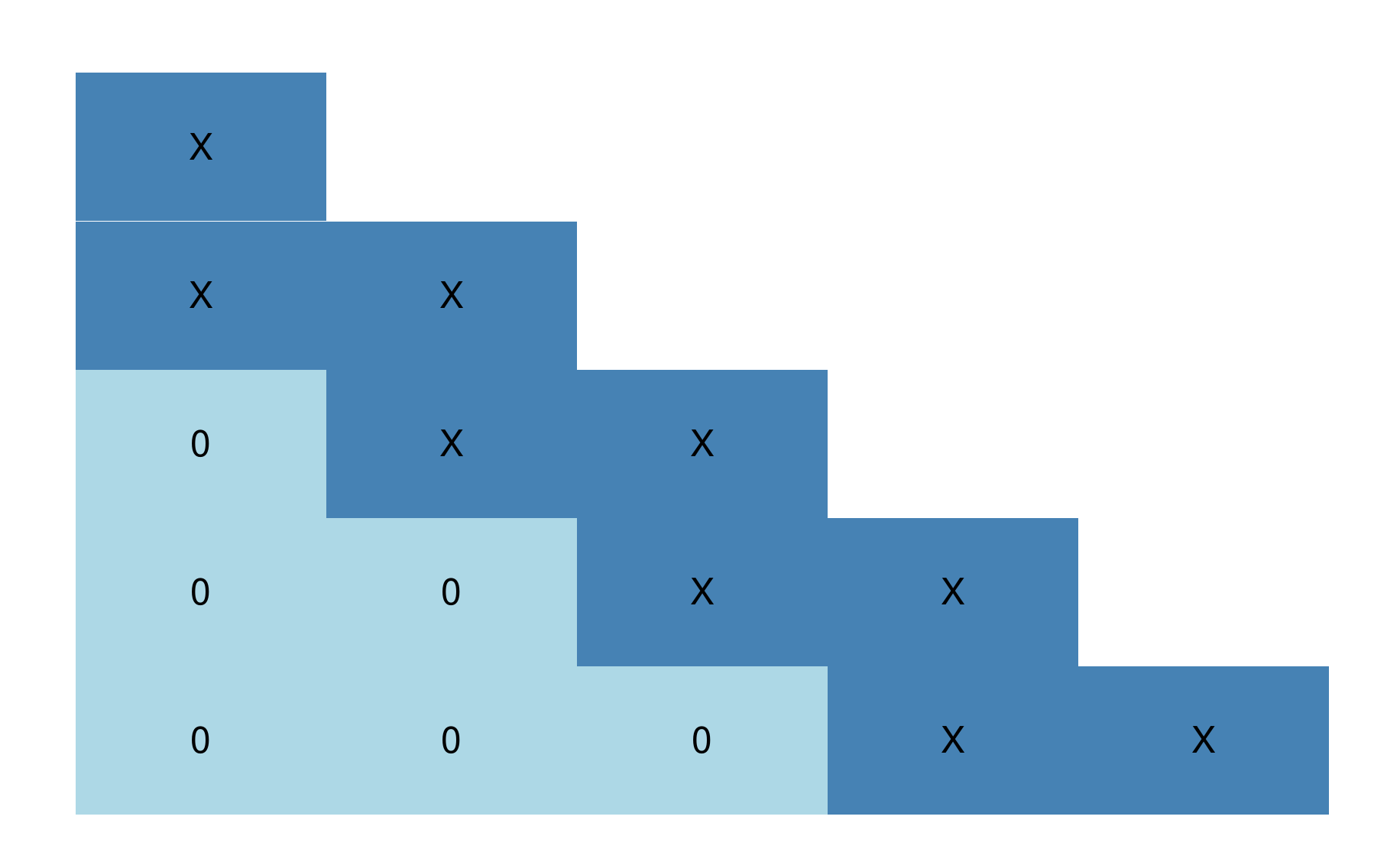}
		(c)
	\end{minipage}
	\caption{Percentage of correct zero identifications in matrix $ L $. (a) Likelihood. (b) Composite likelihood. (c) True zero entries in the matrix.}
	\label{lassovscomp}
\end{figure}

\begin{figure}[H]
	\centering
	\includegraphics[scale=0.4]{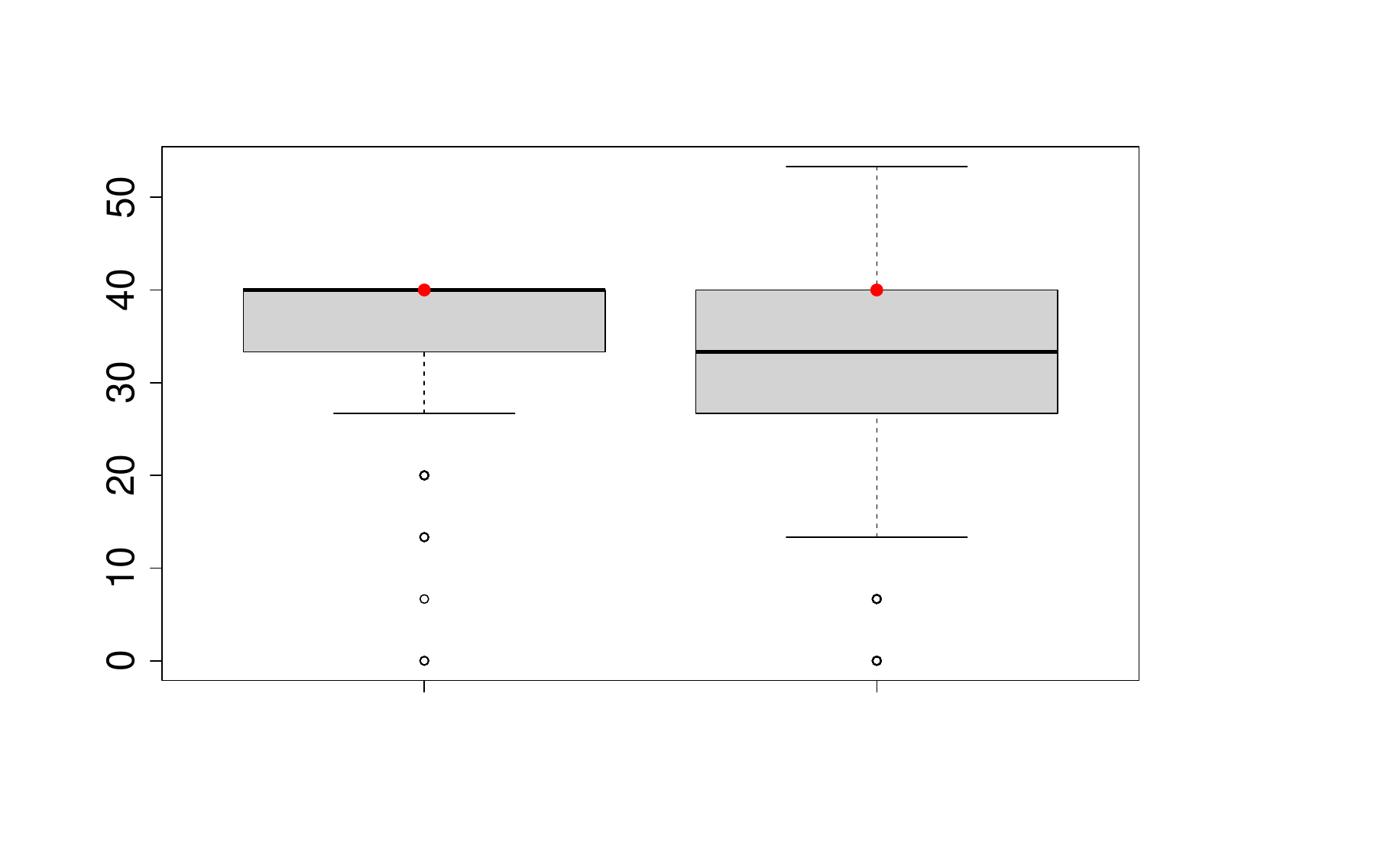}
	\caption{Boxplot of the percentage of zeros in the estimated $ L $ matrix. Left: Likelihood. Right: Composite likelihood.}
	\label{Sparsidad_vs}
\end{figure}

\begin{table}[H]
\centering
\begin{minipage}{.45\textwidth}
\centering
\begin{tabular}{|c|c|c|}
\hline
 & \multicolumn{2}{|c|}{Predicted} \\
\hline
True & Non-zero & Zero \\
\hline
Non-zero & 100 \%  & 11.97 \% \\
Zero     & 0.00\% & 88.03 \% \\
\hline
\end{tabular}
\caption{Likelihood}
\end{minipage}
\hfill
\begin{minipage}{.45\textwidth}
\centering
\begin{tabular}{|c|c|c|}
\hline
 & \multicolumn{2}{|c|}{Predicted} \\
\hline
True & Non-zero & Zero \\
\hline
Non-zero & 97.76 & 23.87 \\
Zero     & 3.37 & 76.13 \\
\hline
\end{tabular}
\caption{Composite likelihood}
\end{minipage}

\caption{Confusion matrices for zero detection in matrix $L$. Positive (negative) denotes non-zero (zero) correlations, respectively. Left: Likelihood. Right: Composite likelihood.}
\label{tab:confusionp5}
\end{table}

\begin{figure}[H]
	\centering
	\begin{minipage}{.5\textwidth}
		\centering
		\includegraphics[scale=0.2]{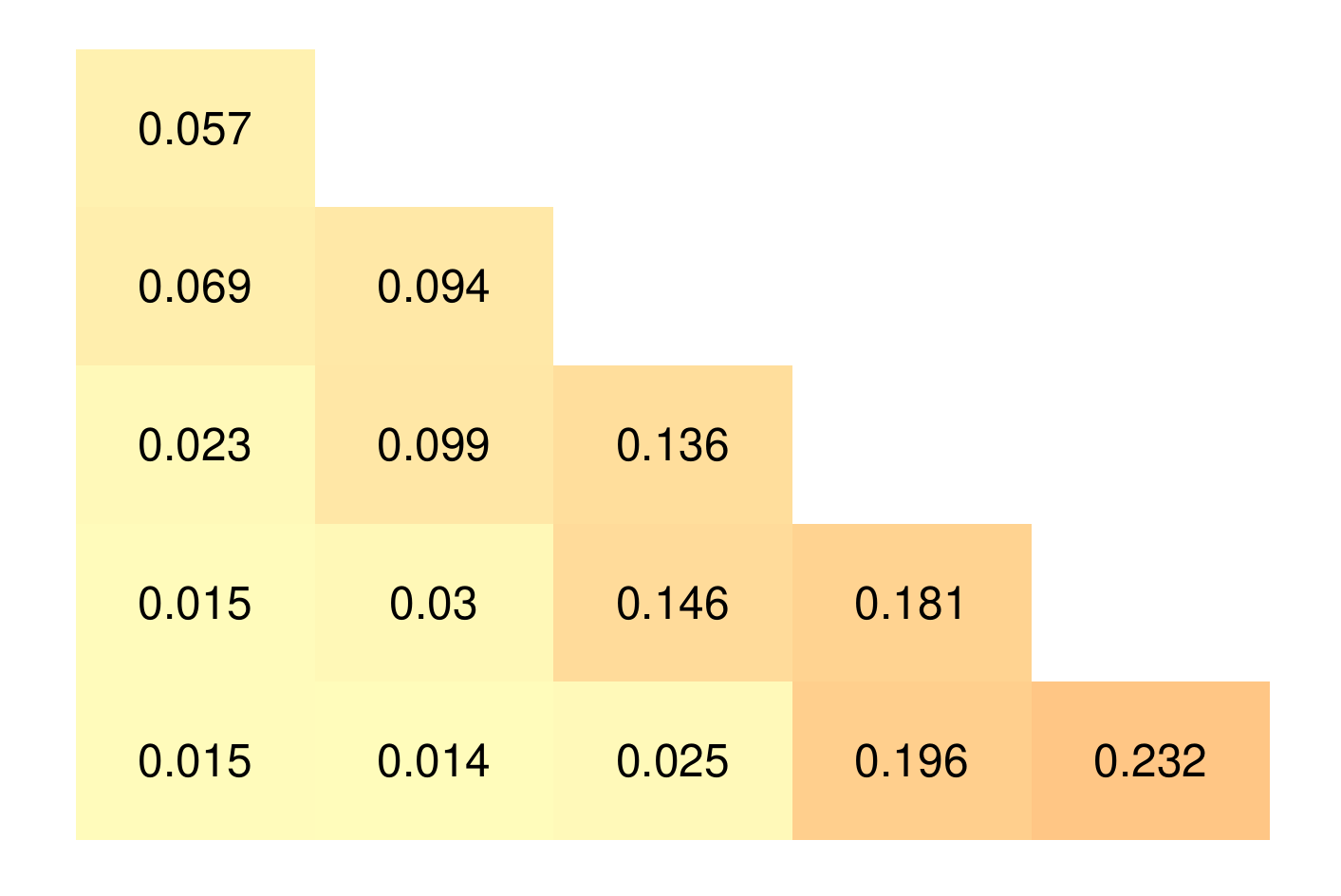}
		(a)
	\end{minipage}%
	\begin{minipage}{.5\textwidth}
		\centering
		\includegraphics[scale=0.2]{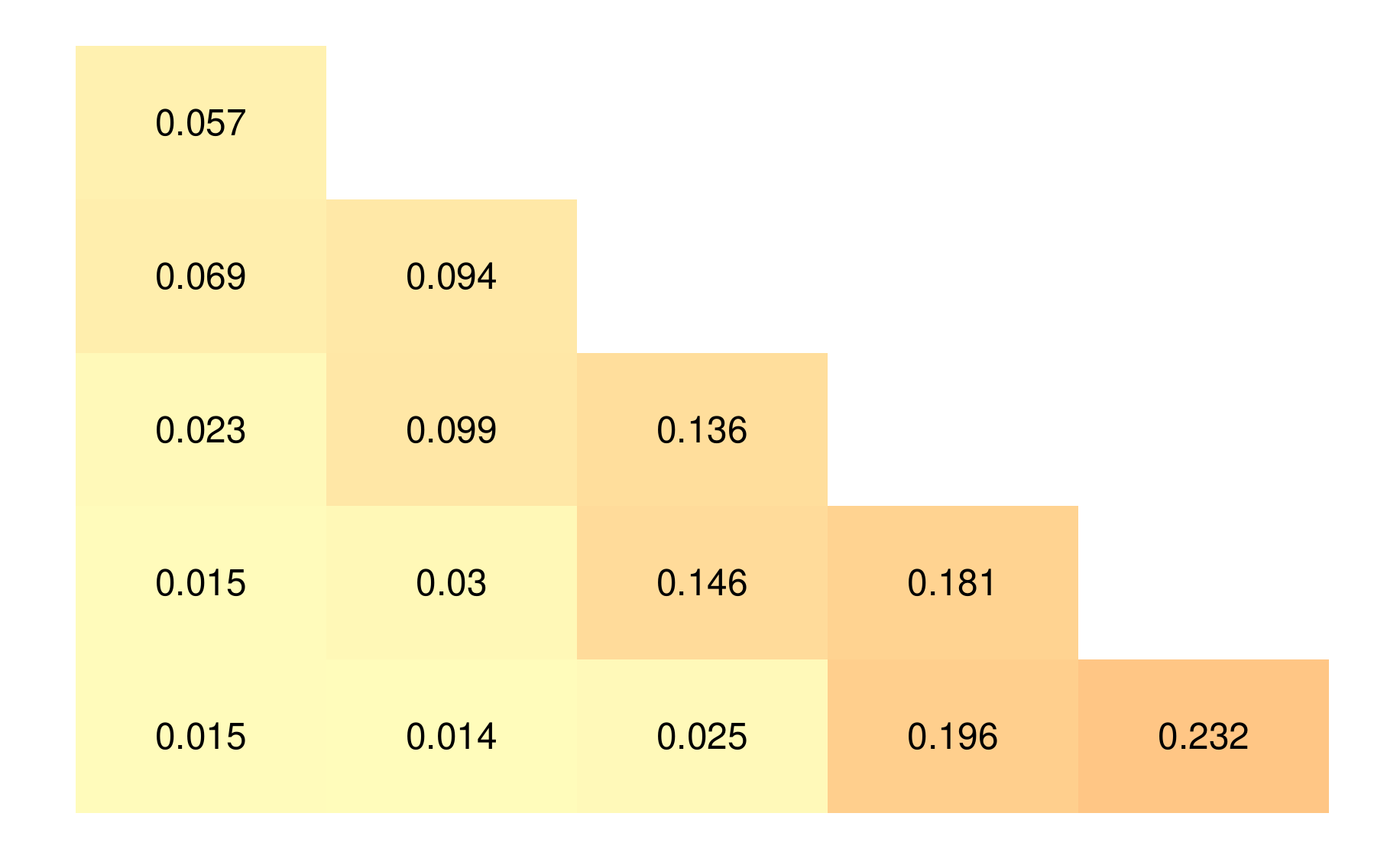}
		(b)
	\end{minipage}
	
	\begin{minipage}{.5\textwidth}
		\centering
		\includegraphics[scale=0.2]{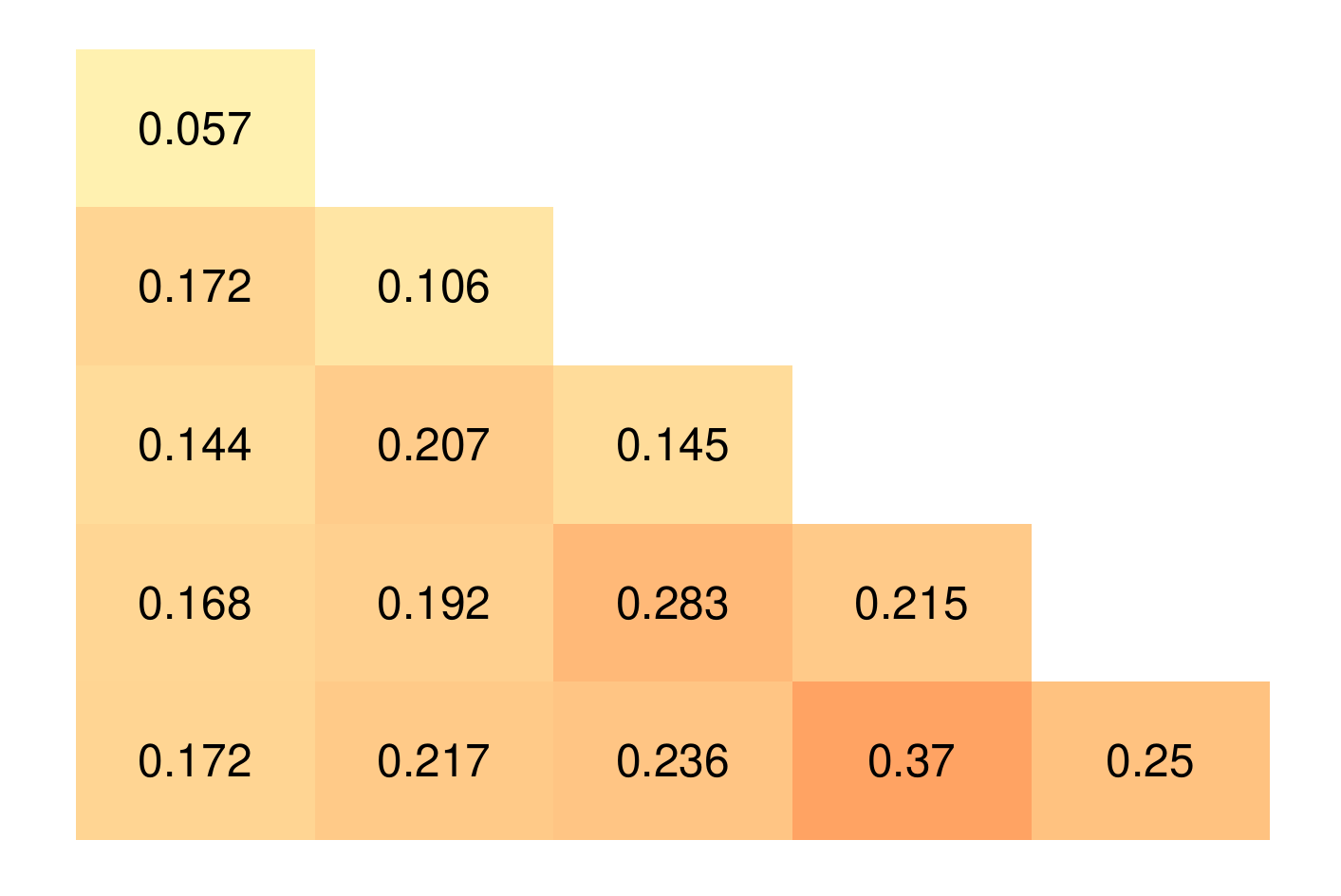}
		(c)
	\end{minipage}%
	\begin{minipage}{.5\textwidth}
		\centering
		\includegraphics[scale=0.2]{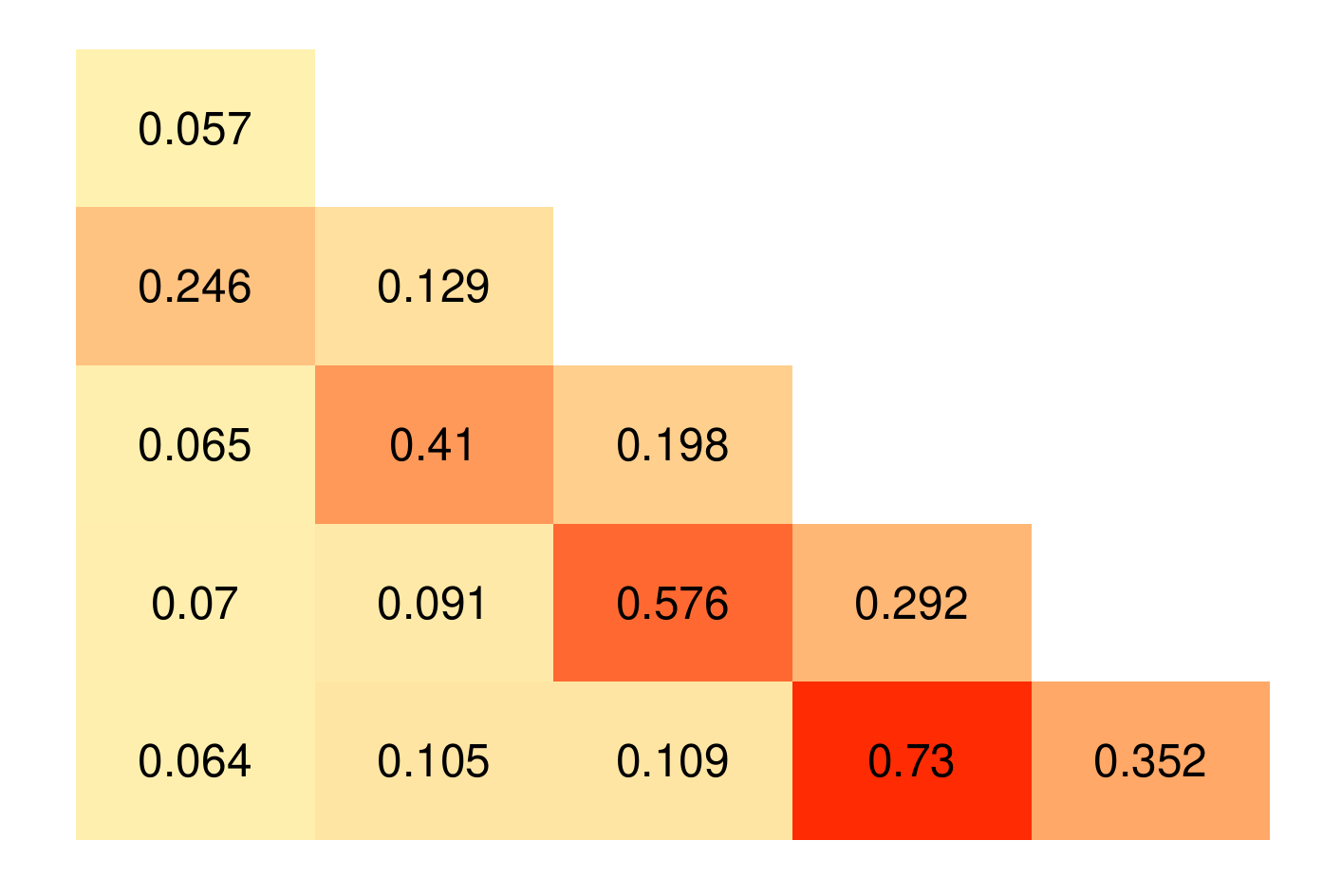}
		(d)
	\end{minipage}
	
	\caption{Root mean square error (RMSE) of the estimated $ L$ matrices. (a) Likelihood without penalization. (b) Likelihood with penalization. (c) Composite likelihood without penalization. (d) Composite likelihood with penalization.}
	\label{Rmse_500sim}
\end{figure}

In Figure \ref{lassovscomp}, we report the accuracy in identifying the zeros in matrix $ L$ when estimating with likelihood and composite likelihood, both with a LASSO penalty. An entry $ (i,j)$ is considered correctly identified if, in the true matrix, the value in that position is zero (non-zero) and the estimated matrix also has a zero (non-zero) at that position, respectively. 

It is observed that the LASSO estimation method with likelihood perfectly identifies the first and second diagonals of matrix $ L $, where the non-zero values are located, while the composite likelihood estimation has a slight error rate in detecting these parameters. On the other hand, the LASSO method does not always identify all the zeros in matrix $ L $, whether estimated with likelihood or composite likelihood, resulting in matrices that generally contain fewer zeros than the true matrix. This is precisely what is seen in the boxplot in Figure \ref{Sparsidad_vs}.

In Table \ref{tab:confusionp5}, the confusion matrices for zero detection in matrix $ L $ are shown, comparing estimation with likelihood and composite likelihood. Notably, the likelihood approach does not yield false negatives, meaning if a value is detected as zero, it is indeed zero. Errors in the likelihood estimation are due to null correlations that the estimator detects as non-zero. Conversely, composite likelihood increases the false negative error rate and generally has a higher error rate in zero detection overall.

In Figure \ref{Rmse_500sim}, the root mean square error (RMSE) of the matrix $ L$ is reported for likelihood estimation (top) without penalization (left) and with penalization (right), and for composite likelihood estimation (bottom) without penalization (left) and with penalization (right). 
It is observed that when estimating with LASSO, the RMSE increases on the first and second diagonals—i.e., where the matrix has non-zero coefficients—and decreases elsewhere. Summing all the RMSE values, we obtain a total RMSE of $1.75$ and $1.34$ for maximum likelihood estimation without and with penalization, respectively, and an RMSE of $2.96$ and $3.5$ for composite likelihood estimation without and with penalization, respectively.

\begin{table}[H]
	\centering
	\begin{tabular}{|c|c|c|c|}
		\hline
		\diagbox[dir=NW]{Method}{$n$ observations} & 100 & 500 & 1000 \\ \hline
		Maximum Likelihood   & $0.61$ &  $5.66$ & $528.40$ \\ \hline
		  Composite Likelihood & $0.45$  & $1.99$ & $5.57$ \\ \hline
	\end{tabular}
	\caption{Computation times (in minutes) for estimation using Maximum Likelihood (first row) and Composite Likelihood (second row).}
	\label{tab:comp_time}
\end{table}

\section{Application}

We now apply our method to a dataset from a mineral exploration campaign conducted in 
a prospective area located in southern Ecuador, consisting of a surface lithogeochemical 
survey with a spacing varying from $100\text{m} \times 400\text{m}$ to $50\text{m} 
\times 50\text{m}$ \citep[the data that support the findings of this study are available on request from the corresponding author, XE]{guartan2021predictive, guartan2021regionalized}. The dataset 
totalizes $3998$ soil or rock samples, each containing information on the concentrations 
of $9$ major elements (reported in \%) and $27$ trace elements (reported in ppm) (Table 
\ref{cuadrito}), as well as on the prevailing lithology (Figure \ref{mapofdata}). The 
dimensions of the study area are approximately $5500\text{m}$ in the east-west direction 
and $6500\text{m}$ in the north-south direction, with elevations ranging from $2490\text{m}$ 
to $4022\text{m}$ above mean sea level.

\begin{figure}[h]
\label{mapofdata}
    \centering
    \includegraphics[scale=0.5]{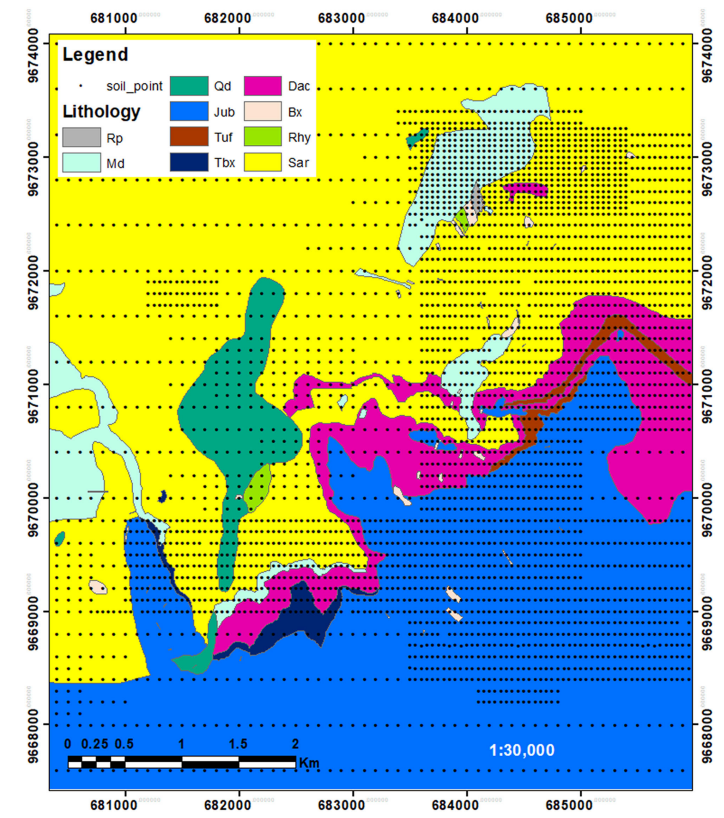}
    \caption{Spatial locations of the surface samples, along with rock types 
    \citep[Figure obtained from][]{guartan2021predictive, guartan2021regionalized}.}
\end{figure}

An empirical normal score transformation was applied to the data to achieve a normal 
distribution for each variable, standardized to have a mean of 0 and variance of 1. 
We focus on the normal score transforms of all $p = 36$ variables, with the primary 
objective of performing cokriging to map the variables of interest: copper (Cu), iron 
(Fe), cobalt (Co), and aluminum (Al). A secondary objective is to identify and remove 
non-relevant cross-correlations, simplifying the cokriging model without compromising 
its predictive capability and reducing the computational burden in estimation and 
prediction processes.

All calculations presented below were performed on a 2×8-Core CPU (Intel Xeon E5-2689 @ 1200--3600 MHz), with 128 GB of RAM, 931 GB of storage, and a 5.4.0-182-generic x86\_64 kernel. To estimate the parameters of the Matérn model, we used the composite likelihood function, as the full likelihood function was computationally infeasible for this dataset size. For the weights $w_{kl}$, we considered the 5 nearest neighbors $NN(5)$.

To account for the measurement error apparent in the data, the covariance model for the $i$-th component is specified as
\begin{equation*}
    \widetilde{C}_{ii}(h) = \tau_i^2 \mathbf{1}_{\{h = 0\}} + C_{ii}(h;\, \sigma_i^2, 
    \alpha_i, \nu),
\end{equation*}
where $\tau_i^2 \geq 0$ is the nugget parameter for the $i$-th variable and $C_{ii}(\cdot)$ is the marginal Mat\'{e}rn covariance as in \eqref{eq:matern_model}. The nugget parameter $\tau_i^2$ is estimated marginally alongside $\sigma_i^2$ and $\alpha_i$ in an extended version of Step~1 of Algorithm~\ref{alg:gen}, by maximizing the individual log-likelihood of each component under independence. For the cross-covariance terms $C_{ij}$ with $i \neq j$, no nugget component is included, as the nugget is treated as a purely idiosyncratic measurement error that is uncorrelated across variables. This assumption is standard in geochemical applications \citep{wackernagel2003multivariate}, where measurement errors for distinct elements are generally produced by independent analytical procedures. The smoothness parameter is fixed at $\nu = 0.5$ for all components, consistent with the fast increase of the direct variograms at small distances
(Figure~\ref{variograms_interest}). Table~\ref{Tab:Varianza_y_Nugget} of the supplementary material reports the estimated nugget $\hat{\tau}_i^2$, partial sill $\hat{\sigma}_i^2$, and nugget-to-sill ratio $\hat{\tau}_i^2/(\hat{\tau}_i^2 + \hat{\sigma}_i^2)$ for each of the $36$ variables. In particular, the four variables of interest --- Cu, Fe, Co, and Al --- exhibit nugget-to-sill ratios of $6.3\%$, $65.7\%$, 
$13.9\%$, and $35.9\%$, respectively. The large nugget for $Fe$ suggests that a substantial fraction of its total variability is attributable to measurement error or sub-sample-spacing spatial structure, which has direct implications for the cokriging prediction variances of that variable.

Prior to the estimation procedure, the dataset was partitioned by randomly selecting 800 spatial locations to be used for prediction (test set), while the remaining observations (approximately $80\%$ of the data) were used for model estimation.

Regarding the LASSO penalty, we generated a sequence of 100 $\lambda$ values and estimated the parameters of the multivariate Matérn model, which took a total of $21.82$ hours. The CLIC criterion was used to select the optimal $\lambda$. To obtain the bootstrap estimator of the matrix $\hat{J}(\hat{\theta})$, we initially took 120 subsamples of 50 pairs of points each for the first 50 estimations. Due to high computational costs, we then reduced each subsample to only 10 pairs of points and took a total of 100 subsamples for the remaining 50 estimations. This process took a total of 13.325 days.

In Figure \ref{Lineas}, the solution paths for the covariances of the 36 variables with 
the 4 variables of interest are shown as a function of $\lambda$. The composite likelihood 
function clearly separates variables whose cross-covariance with the variables of interest 
decays rapidly to zero from those that remain non-zero at the optimal $\lambda$, providing 
an interpretable sparsity structure. Figure \ref{Sparsidad_tot} shows the percentage of 
zeros in matrices $L$ and $\Psi$ and the CLIC criterion as a function of $\lambda$. The 
optimal $\lambda$ was selected when matrices $L$ and $\Psi$ had $89.78\%$ and $52.31\%$ 
zero entries, respectively.

It should be noted that these data could not be analyzed without penalization because 
storing the full covariance matrix would have required over 130 GB, whereas the penalized 
estimation at the optimal $\lambda$ required a maximum of only 1.3 GB, demonstrating the 
critical role of sparsity in making this problem computationally tractable.

\begin{figure}
    \centering
    \begin{minipage}{.5\textwidth}
        \centering
        \includegraphics[scale=0.22]{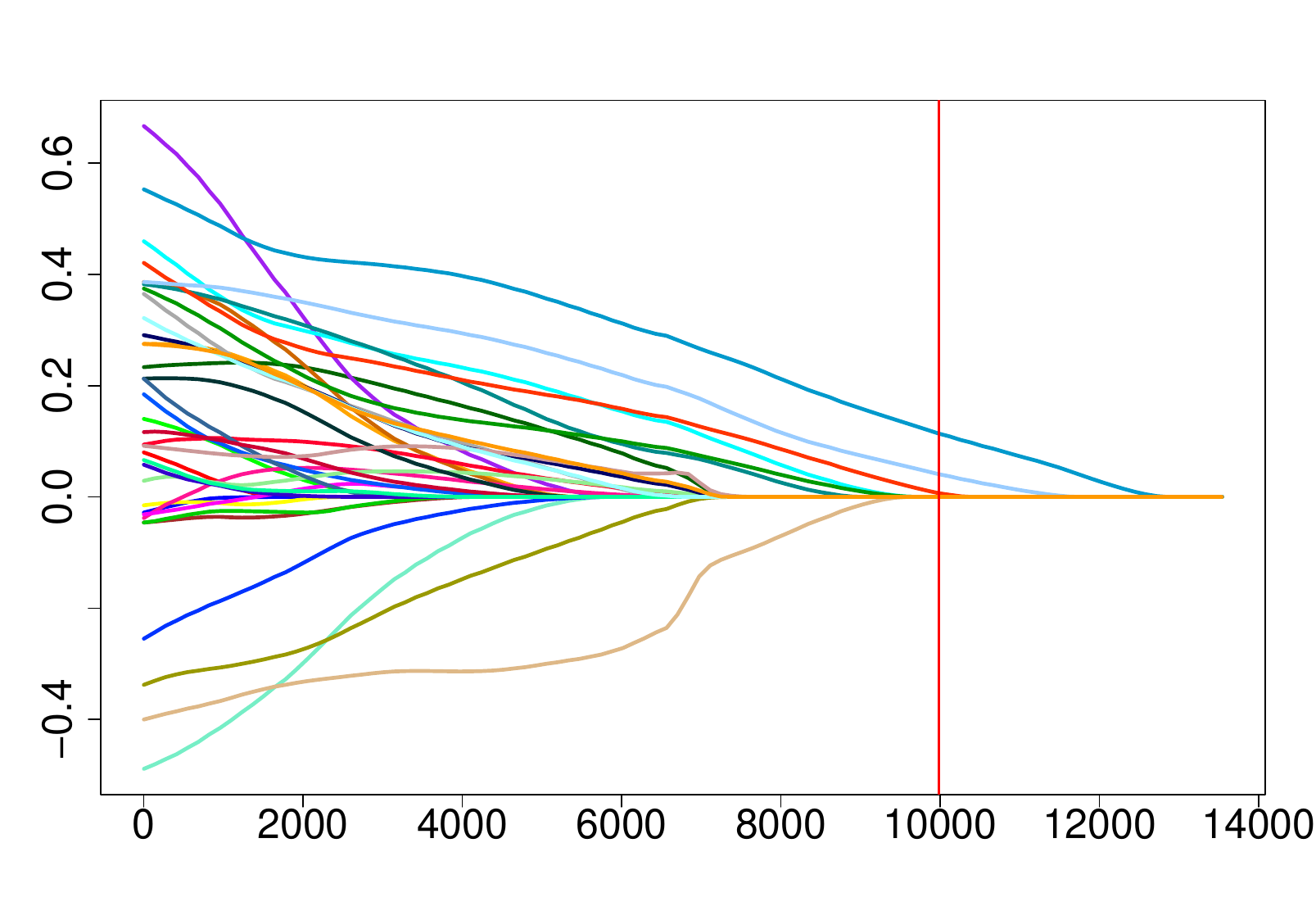}
        (a) Cu
    \end{minipage}%
    \begin{minipage}{.5\textwidth}
        \centering
        \includegraphics[scale=0.22]{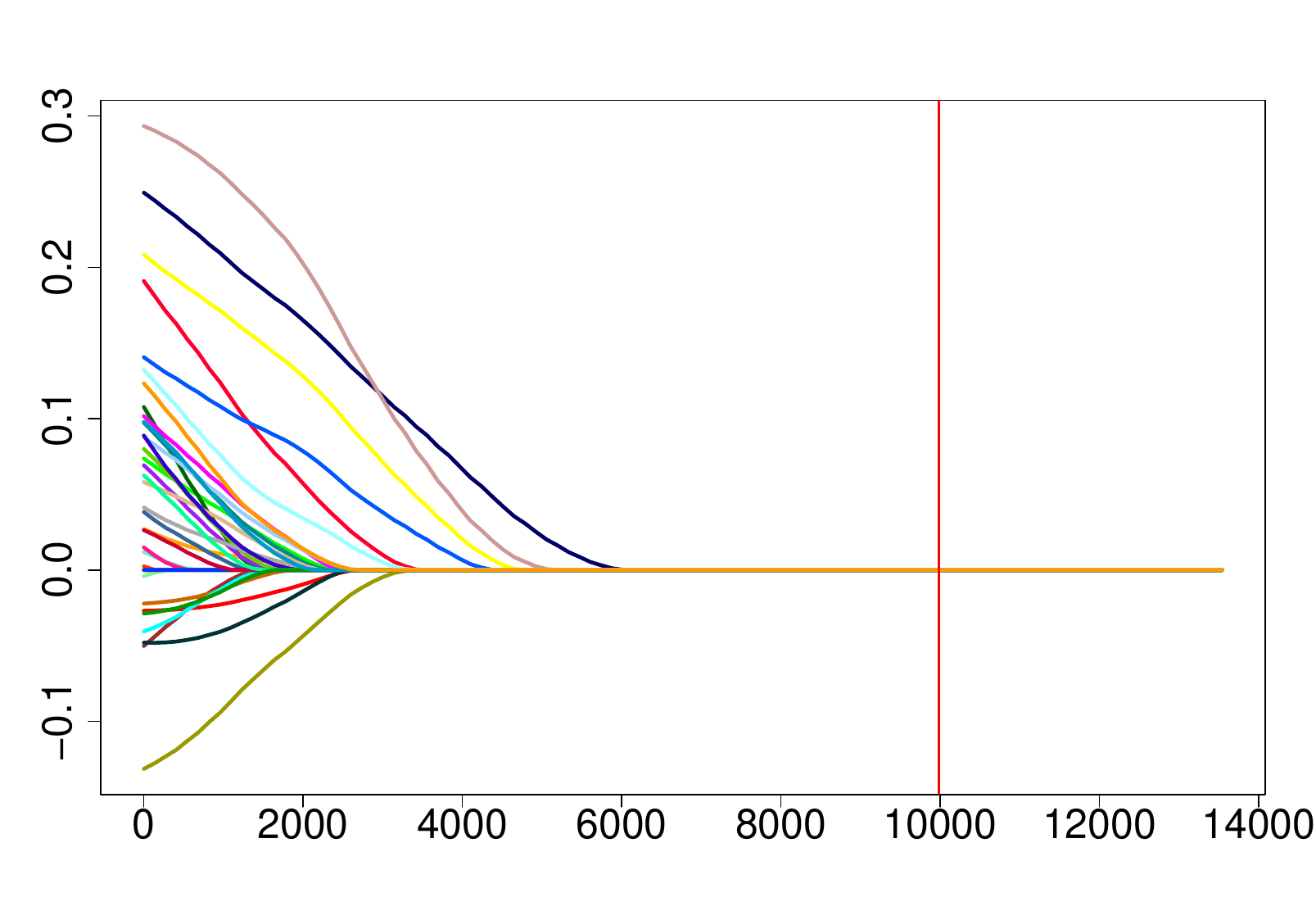}
        (b) Fe
    \end{minipage}
    
    \begin{minipage}{.5\textwidth}
        \centering
        \includegraphics[scale=0.22]{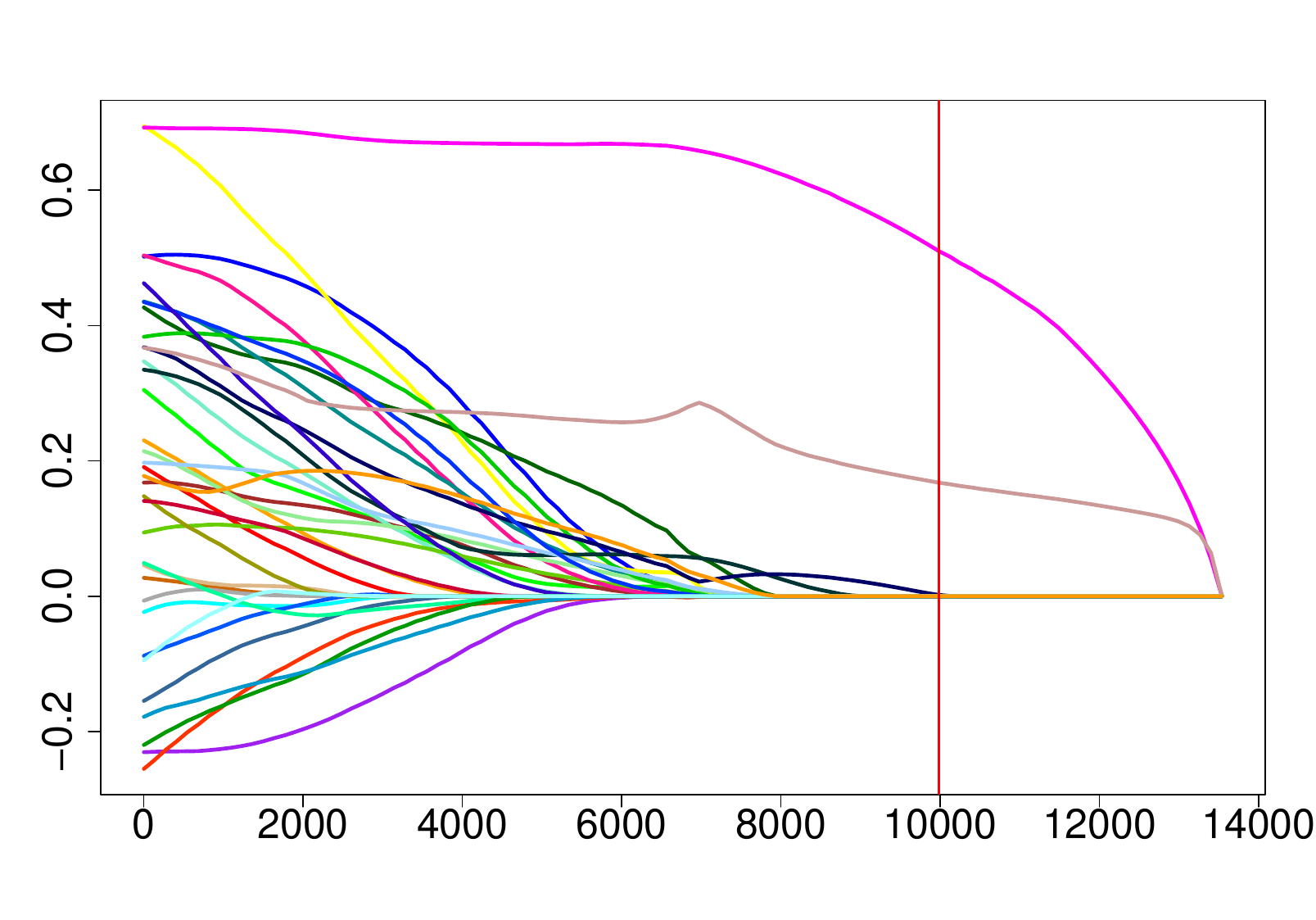}
        (c) Co
    \end{minipage}%
    \begin{minipage}{.5\textwidth}
        \centering
        \includegraphics[scale=0.22]{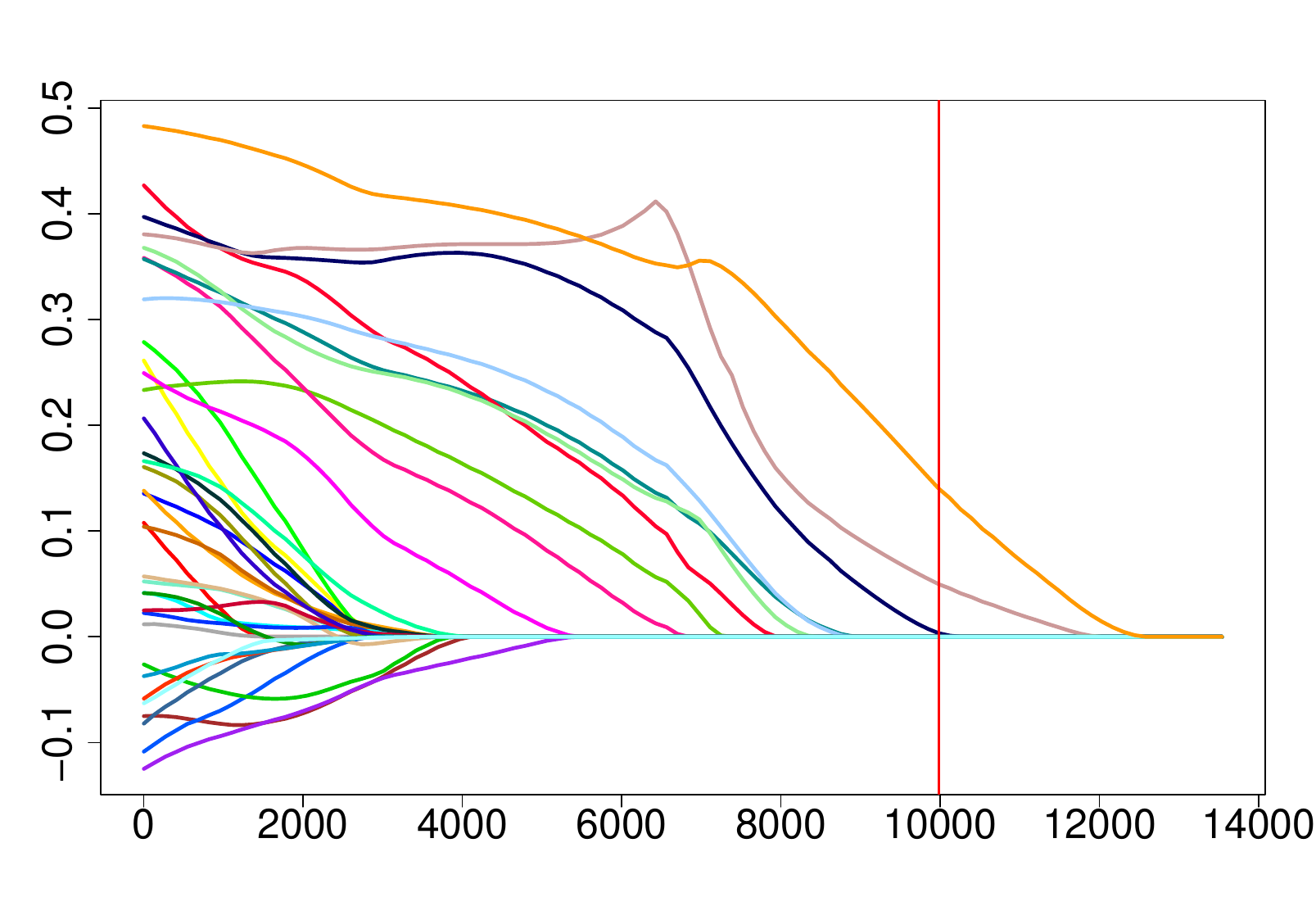}
        (d) Al
    \end{minipage}
    
    \caption{Solution paths of covariances of the variables of interest with respect to 
    other variables, estimated using the composite likelihood function. The optimal 
    $\lambda$ is shown in red.}\label{Lineas}
\end{figure}

\begin{figure}
    \centering
    \begin{minipage}{.32\textwidth}
        \centering
        \includegraphics[scale=0.15]{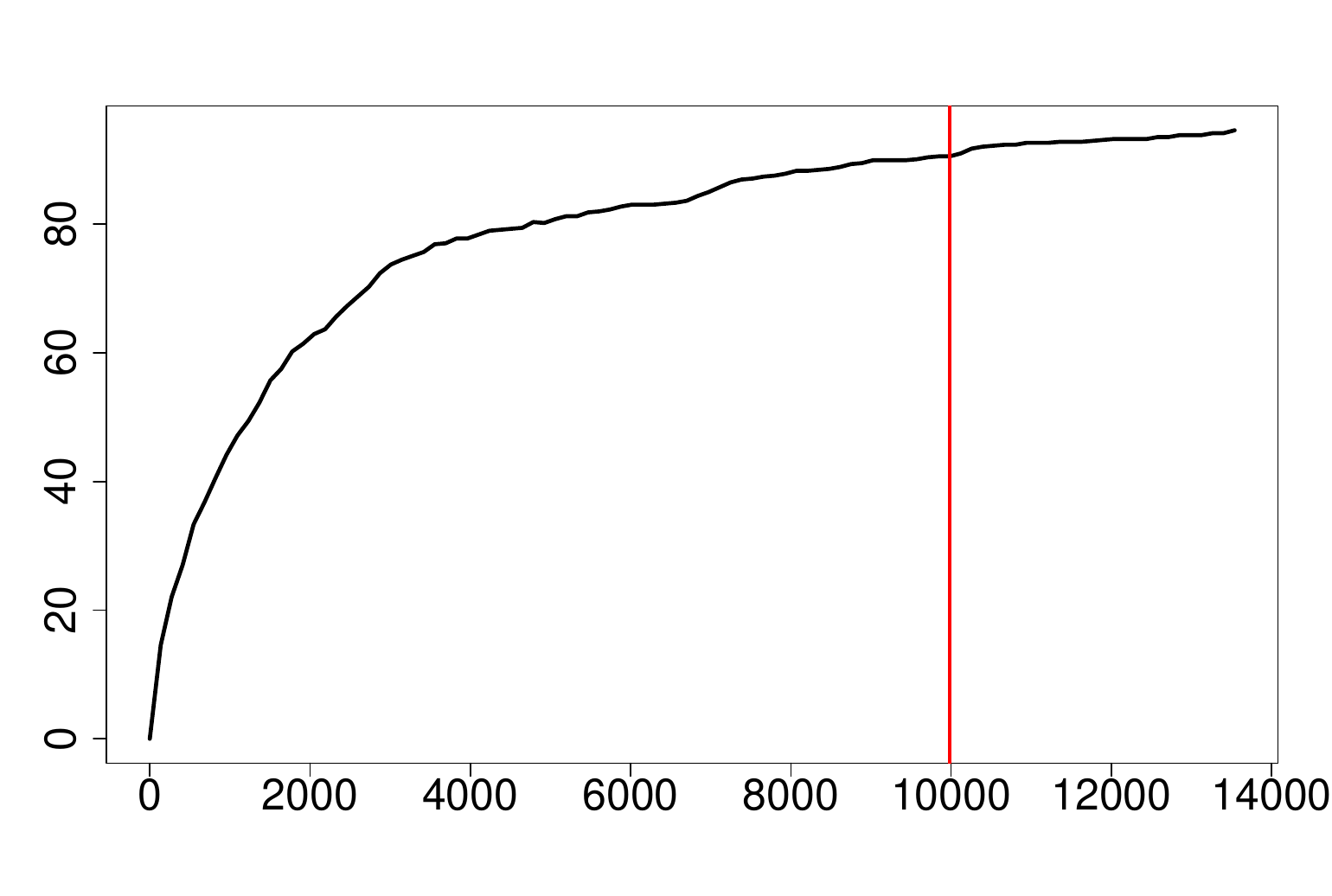}
    \end{minipage}%
    \begin{minipage}{.32\textwidth}
        \centering
        \includegraphics[scale=0.15]{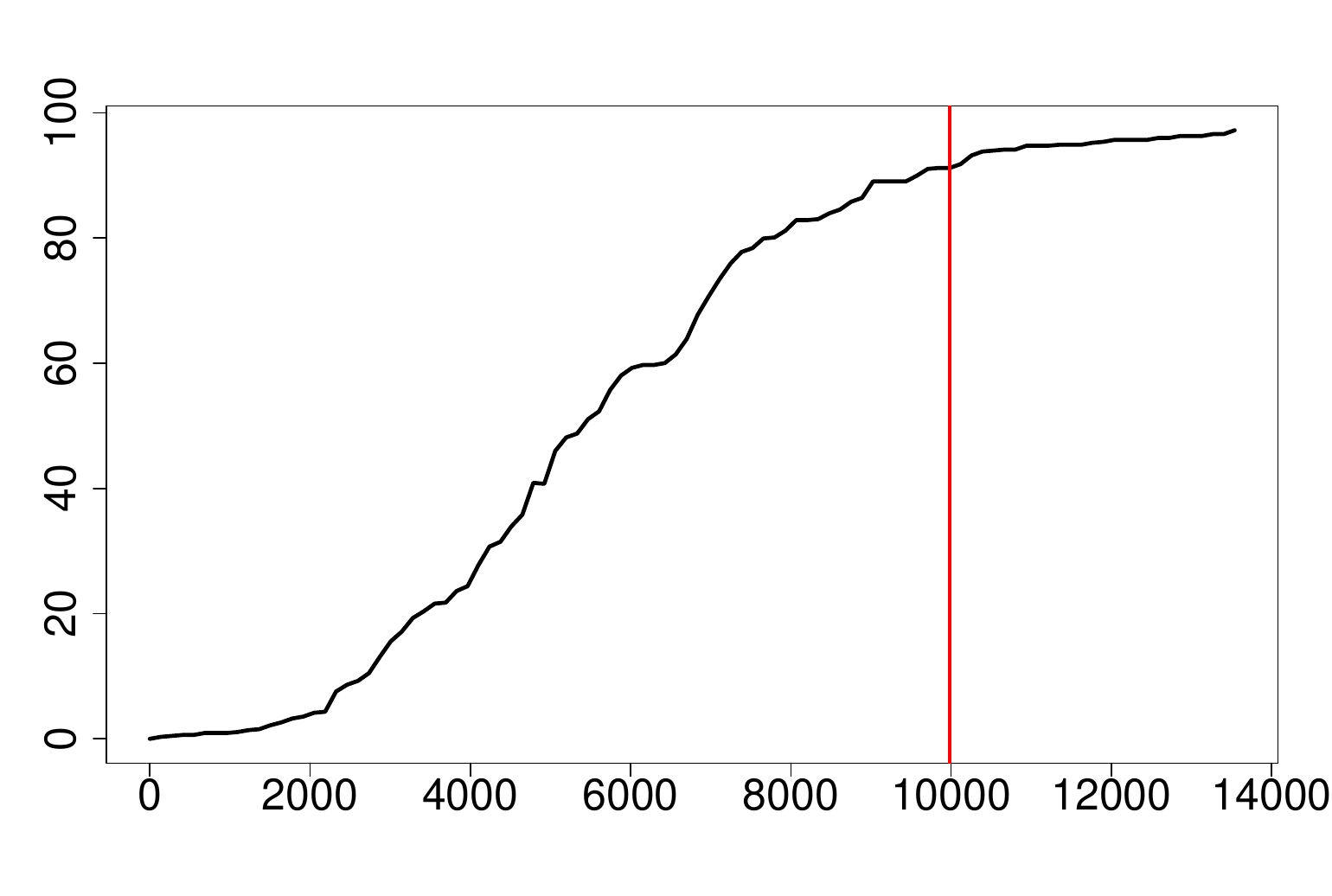}
    \end{minipage}
    \begin{minipage}{.32\textwidth}
        \includegraphics[scale=0.15]{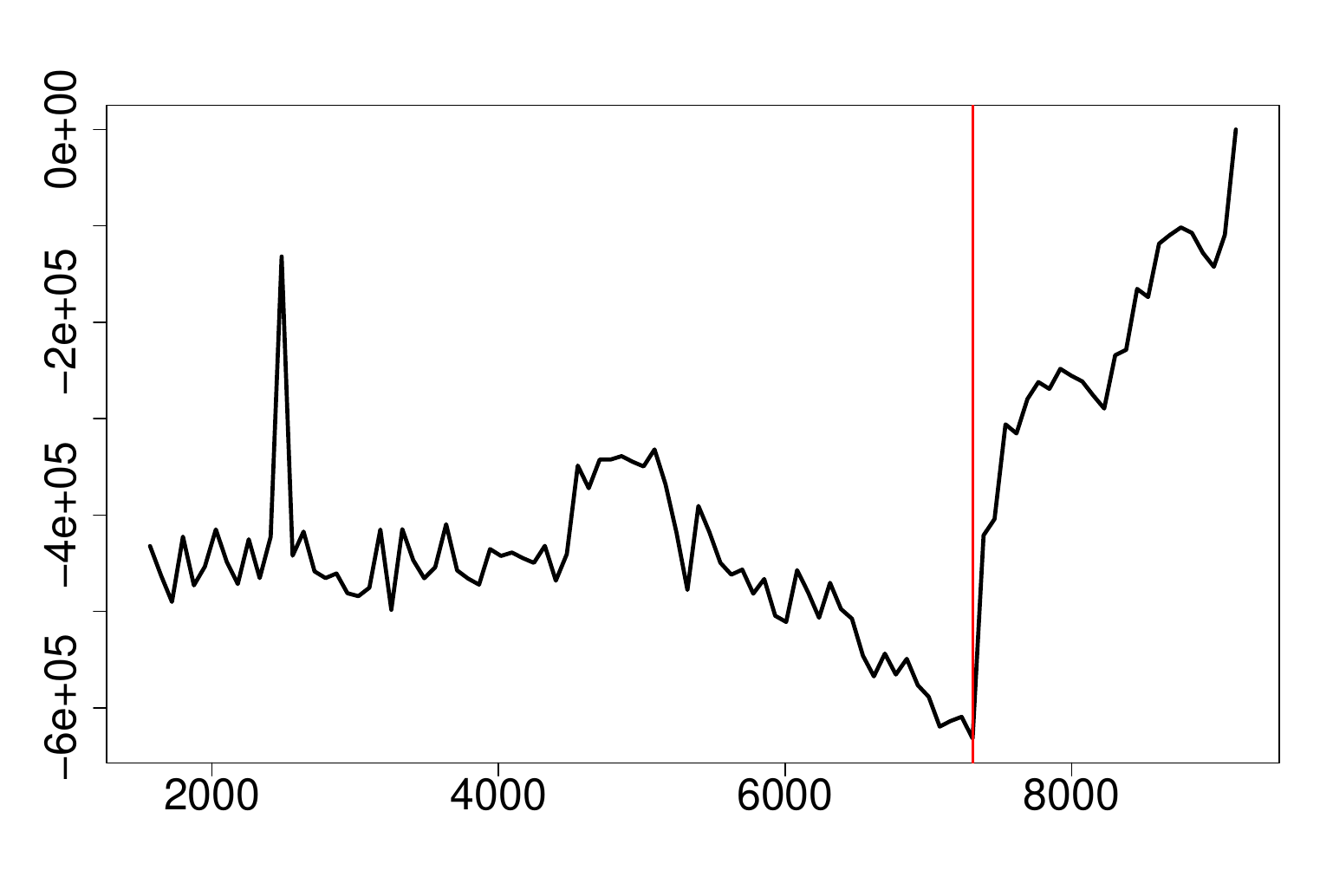}
    \end{minipage}
    \caption{Left: Percentage of zeros in matrix $L$ relative to $\lambda$. Center: 
    Percentage of zeros in matrix $\Psi$ relative to $\lambda$. Right: CLIC criterion 
    relative to $\lambda$. The optimal $\lambda$ is shown in red.}\label{Sparsidad_tot}
\end{figure}

After determining the zero values in matrix $L$, cokriging was performed using the retention method on the previously defined subset of the data, with 800 spatial locations used for prediction of the four variables of interest and the remaining data used for estimation.

Table \ref{tab:rmse_varianza} shows the root mean square errors (RMSE) and estimated 
cokriging error variances for the four variables of interest using the composite 
likelihood method with the CLIC-selected $\lambda$. The cokriging results for copper 
(Cu) are shown in Figure \ref{predicciones_final}.

\begin{table}[H]
    \centering
    \begin{tabular}{|c|c|c|c|c|}
        \hline
        \textbf{Variables} & \textbf{Cu} & \textbf{Fe} & \textbf{Co} & \textbf{Al} \\
        \hline
        \textbf{RMSE} & 0.7357 & 0.8918 & 0.7631 & 1.2909 \\
        \hline
        \textbf{Standard Deviation} & 0.1412 & 0.0991 & 0.1019 & 0.1648 \\
        \hline
    \end{tabular}
    \caption{Root Mean Square Error (RMSE) and standard deviation of cokriging 
    predictions for the variables of interest at the test set.}\label{tab:rmse_varianza}
\end{table}

\begin{figure}[H]
    \centering
    \begin{minipage}{0.33\textwidth}
        \centering
        \includegraphics[width=\linewidth]{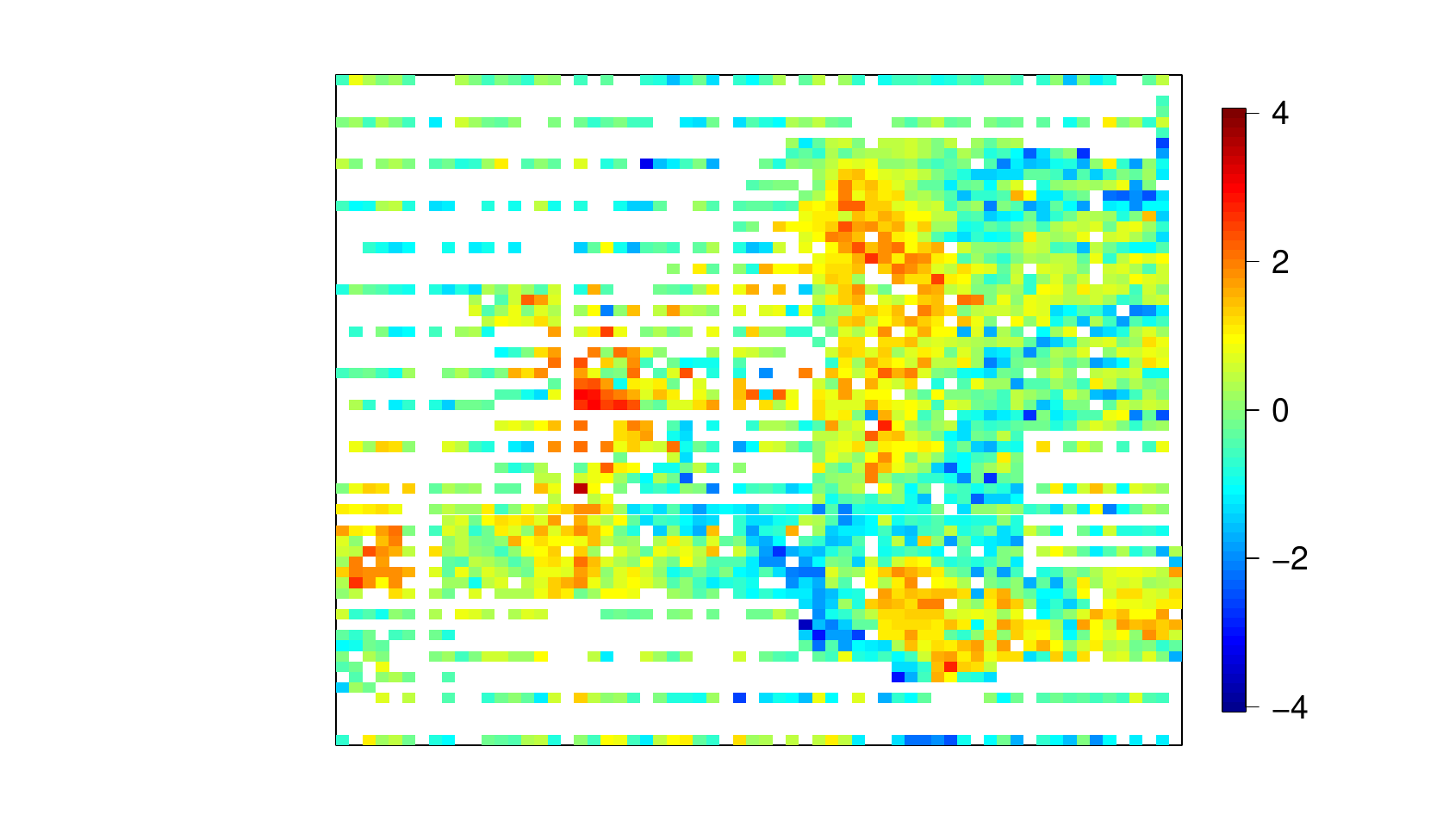}
        (a)
    \end{minipage}\hfill
    \begin{minipage}{0.33\textwidth}
        \centering
        \includegraphics[width=\linewidth]{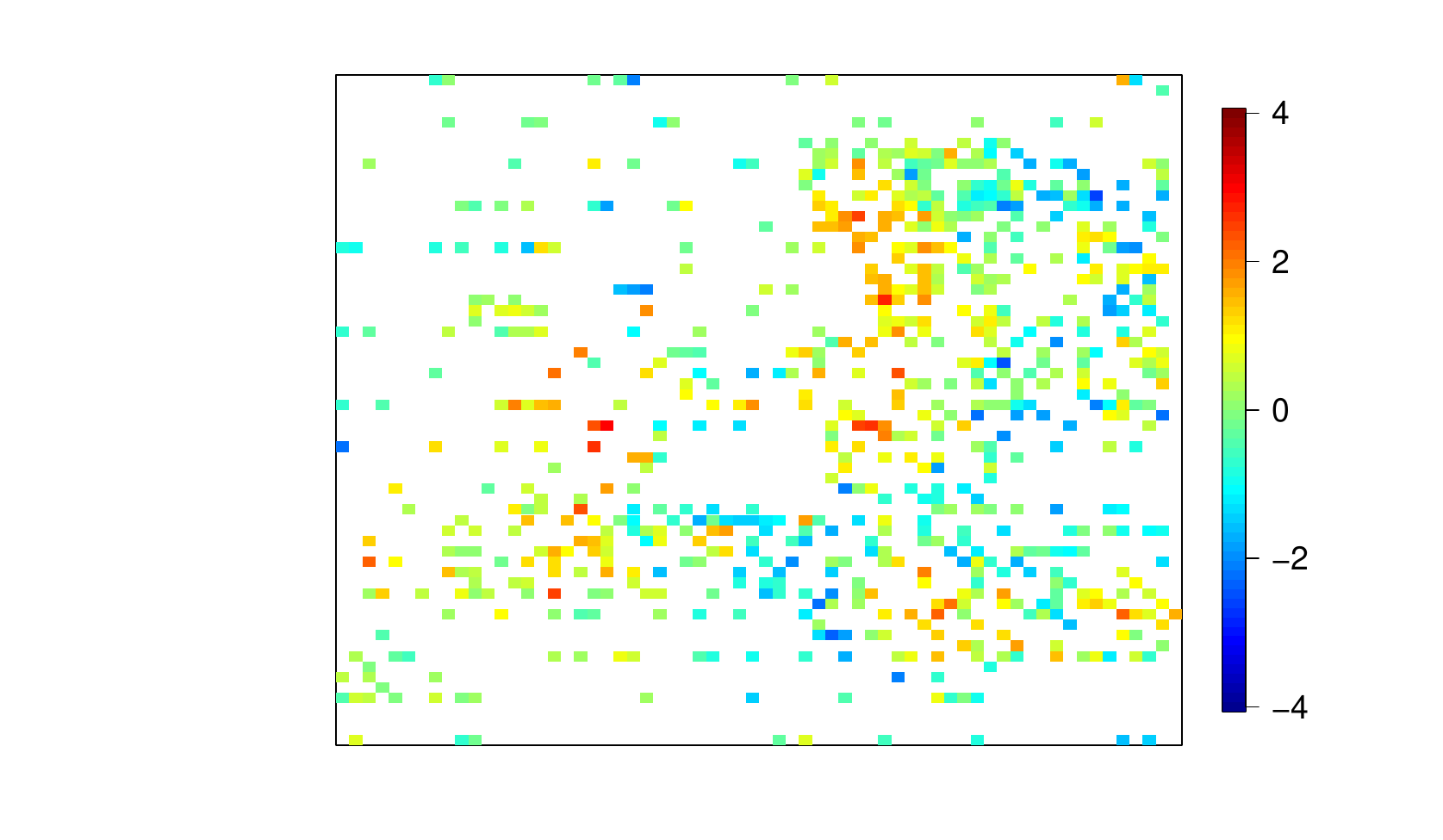}
        (b)
    \end{minipage}\hfill
    \begin{minipage}{0.33\textwidth}
        \centering
        \includegraphics[width=\linewidth]{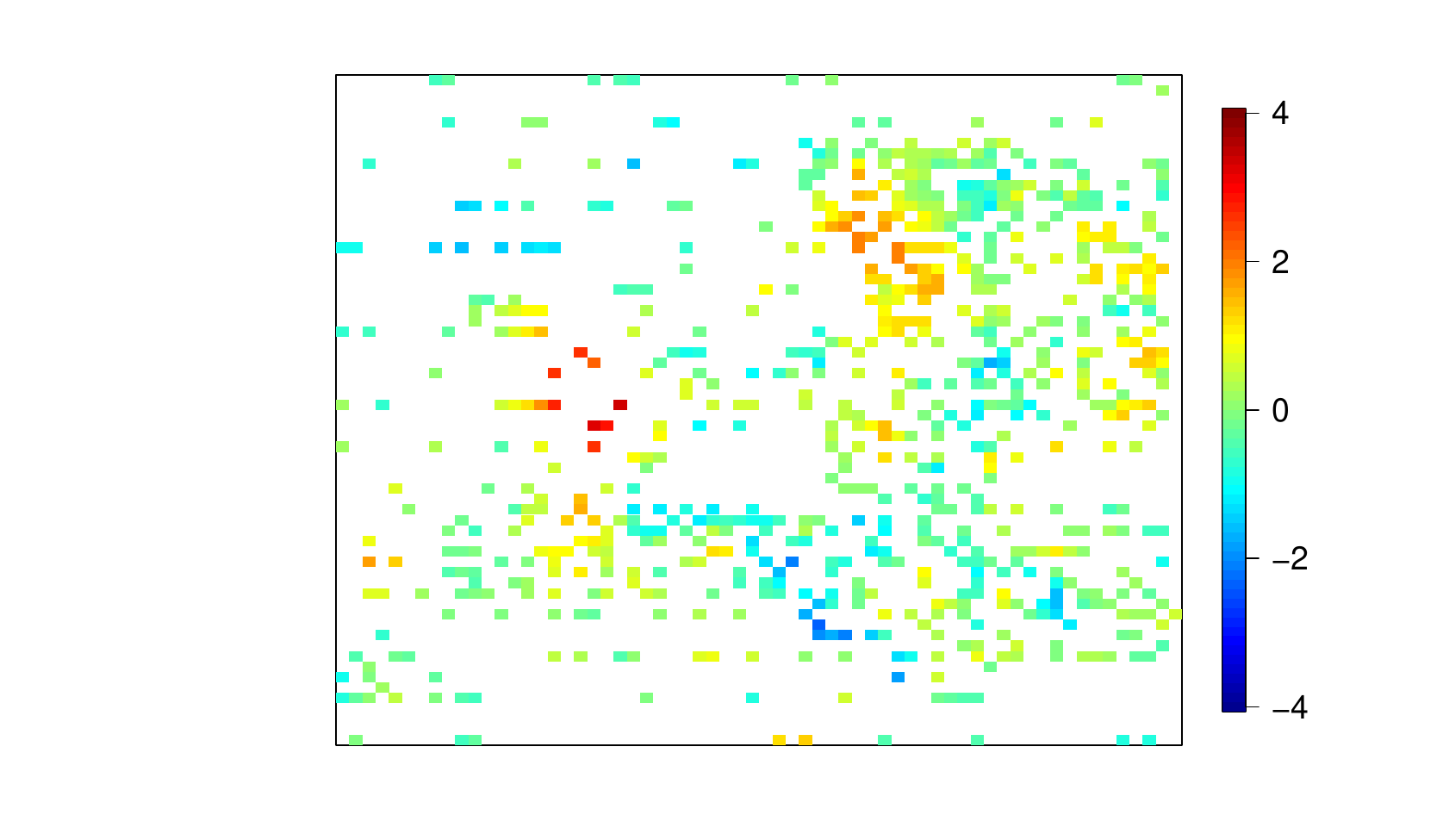}
        (c)
    \end{minipage}
     
    \caption{Cokriging predictions for copper (Cu). (a) Training data. (b) Test 
    data. (c) Prediction using the model selected by CLIC.}
    \label{predicciones_final}
\end{figure}

\section{Conclusions}

This paper has addressed the challenging problem of estimating covariance parameters for highly multivariate spatial Gaussian random fields, where traditional maximum likelihood approaches become computationally infeasible due to the rapid growth in the number of parameters and the high dimensionality of the problem. We have proposed a regularized estimation framework that combines LASSO-type penalization with a projected block coordinate descent algorithm, enabling both parameter estimation and correlation structure identification in high-dimensional settings.

We developed a tailored optimization algorithm that exploits the natural block structure of multivariate Matérn covariance parameters, ensuring that parameter constraints are satisfied at each iteration through appropriate projections onto convex parameter spaces. This approach maintains computational tractability while guaranteeing valid covariance models. Second, we introduced a LASSO penalty on the Cholesky decomposition matrix $L$ of the covariance matrix $\Psi$, which directly induces sparsity in the correlation structure and allows for automatic identification of significant cross-variable relationships. Third, we adapted information criteria (AIC for likelihood, CLIC for composite likelihood) to select the optimal regularization parameter $\lambda$, providing a data-driven approach to balance model complexity and fit.

Our simulation studies demonstrate the effectiveness of the proposed methodology. In the 5-variable experiments, the LASSO-penalized estimator successfully identified zero correlations with high accuracy, particularly when using the full likelihood function, though the composite likelihood approach showed slightly higher false negative rates while offering substantial computational savings. Importantly, our results show that penalized estimation can reduce the overall RMSE by removing spurious correlations, leading to more parsimonious and interpretable models. 

Several limitations and opportunities for future research merit discussion. First, our approach fixes the smoothness parameter $\nu$ to avoid identifiability issues, which is appropriate for many applications but may be restrictive in others. Future work could explore data-driven methods for smoothness parameter selection or develop adaptive penalization schemes that account for varying degrees of smoothness across variables. Second, while our LASSO penalty effectively identifies zero correlations, it tends to produce conservative estimates that may include some spurious non-zero correlations (false positives) to avoid missing true relationships. Alternative penalty functions, such as adaptive LASSO or SCAD, might improve selection accuracy. Third, the selection of weights in the composite likelihood function deserves further investigation, as different weighting schemes may affect both computational efficiency and statistical properties.

Extensions of this work could consider non-stationary multivariate models, where parameters vary spatially, or space-time settings where temporal dependencies must also be modeled. The block coordinate descent framework could naturally accommodate such extensions through appropriate parameterizations. Additionally, incorporating prior knowledge about correlation structures, such as through graphical model constraints or structured penalties, could further improve estimation in specific application domains.

\if0\blind
{
\section{Acknowledgements}
Francisco Cuevas-Pacheco was funded and supported by Agencia Nacional de Investigaci\'on y Desarrollo (Chile), grant ANID/FONDECYT/INICIACION 11240330. Gabriel Riffo acknowledge funding by ANID BECAS/MAGISTER NACIONAL 22230570, by UTFSM through the Programa de Iniciación a la Investigación Científica No. 017/23 and by the Deutsche Forschungsgemeinschaft (DFG, German Research Foundation) – CRC/TRR 388 ``Rough Analysis, Stochastic Dynamics and Related Fields'' – Project ID 516748464. Xavier Emery was funded and supported by Agencia Nacional de Investigaci\'on y Desarrollo (Chile), grant ANID AMTC CIA250010.
} \fi

\bibliographystyle{apalike}
\bibliography{biblioord}

\bigskip
\begin{center}
{\large\bf SUPPLEMENTARY MATERIAL}
\end{center}

\appendix

\section{Derivative Calculation}

If we differentiate the likelihood function with respect to $\theta$, we obtain the 
following expression:
\begin{align}
    \frac{\partial \ell}{\partial \theta}(\theta;\bold{z}_n) &= -\frac{1}{2} 
    \text{tr}\left(\Sigma^{-1}(\theta) \dfrac{\partial \Sigma}{\partial \theta}(\theta)
    \right) - \bold{z}_n^{\top} \Sigma^{-1} \dfrac{\partial \Sigma}{\partial \theta}
    (\theta) \Sigma^{-1}(\theta) \bold{z}_n.
\end{align}
On the other hand, if we differentiate the function $cl^{p}_{kl}$ with respect to 
$\theta$, we get:
\begin{align*}
    \dfrac{\partial cl^{p}_{kl}}{\partial \theta}(\theta;\bold{z}_n) &= -\frac{1}{2} 
    \text{tr}\left(\left(\Sigma^{(kl)}(\theta)\right)^{-1} \Sigma^{(kl)}(\theta)
    \right) \\
    &+ (\bold{z}_n(s_k), \bold{z}_n(s_l))^{\top} \left(\Sigma^{(kl)}(\theta)\right)^{-1} 
    \dfrac{\partial \Sigma}{\partial \theta}(\theta) \left(\Sigma^{(kl)}(\theta)\right)^{-1} 
    (\bold{z}_n(s_k), \bold{z}_n(s_l)).
\end{align*}
Since the function $pl_p$ is the sum of the $cl^{p}_{kl}$, to calculate the derivative 
of $pl_p$ with respect to $\theta$, we simply sum the derivatives of the different 
$cl^{p}_{kl}$.

For the standard and composite likelihood functions, we obtain similar expressions, both 
for the function itself and its derivative. Therefore, in the rest of this section, we 
will assume that we are optimizing the standard likelihood function, though the procedure 
is analogous for the case of the composite likelihood, unless otherwise specified.

Assuming we know $\dfrac{\partial \Sigma}{\partial \theta}(\theta)$ to calculate 
$\dfrac{\partial \ell}{\partial \theta}$, we must compute:
\begin{align}
    \text{tr}\left(\Sigma^{-1}(\theta) \dfrac{\partial \Sigma}{\partial \theta}(\theta)
    \right) \label{trace} \\
    \bold{z}_n^{\top} \Sigma^{-1}(\theta) \dfrac{\partial \Sigma}{\partial \theta}(\theta) 
    \Sigma^{-1}(\theta) \bold{z}_n. \label{quadraticform}
\end{align}

Let $\left(\Sigma(\theta)^{-1}\right)_{kl}$ and $\left(\dfrac{\partial \Sigma}{\partial 
\theta}(\theta)\right)_{kl}$ be the $(k,l)$ entries of $\Sigma(\theta)^{-1}$ and 
$\dfrac{\partial \Sigma}{\partial \theta}(\theta)$.

On the one hand, the expression \eqref{trace} can be computed as follows:
\begin{align}
    \text{tr}\left(\Sigma^{-1}(\theta) \dfrac{\partial \Sigma}{\partial \theta}(\theta)
    \right) &= \sum_{k=1}^{np} \sum_{l=1}^{np} 
    \left(\Sigma(\theta)^{-1} \right)_{kl} \left(\dfrac{\partial \Sigma}{\partial \theta}
    (\theta)\right)_{lk}  \nonumber \\
    &= \sum_{k=1}^{np} \sum_{l=1}^{np} \left(\Sigma(\theta)^{-1} 
    \right)_{kl} \left(\dfrac{\partial \Sigma}{\partial \theta}(\theta)\right)_{kl} 
    \quad \quad \text{by symmetry of } \Sigma(\theta) \nonumber \\
    &= \sum_{k=1}^{np} \sum_{l=1}^{n^2p} \left(\Sigma(\theta)^{-1} 
    \odot \dfrac{\partial \Sigma}{\partial \theta}(\theta)\right)_{kl},  
    \label{hadamard}
\end{align}
where $\odot$ is the Hadamard product. Instead of the matrix product $\Sigma^{-1}(\theta) 
\dfrac{\partial \Sigma}{\partial \theta}(\theta)$, which requires $\mathcal{O}((np)^{3})$ 
operations, the Hadamard product requires $\mathcal{O}((np)^2)$ operations, thus reducing 
the computational complexity of this operation. Nevertheless, it is still necessary to 
calculate the inverse of matrix $\Sigma(\theta)$. If we use the Cholesky decomposition 
to calculate this inverse, we still need a total of $\mathcal{O}((np)^3)$ operations, so 
the problem, in general, retains cubic complexity. Specifically, we first calculate the 
Cholesky factor $U(\theta)$ such that $\Sigma(\theta) = U(\theta)^{\top} U(\theta)$, 
with $U(\theta)$ being an upper triangular matrix. We then calculate $U(\theta)^{-1}$, 
the inverse of $U(\theta)$, and finally compute $\Sigma(\theta)^{-1} = U(\theta)^{-1}
(U(\theta)^{-1})^{\top}$.

On the other hand, to calculate \eqref{quadraticform}, we use the inverse 
$\Sigma(\theta)^{-1}$ obtained when calculating \eqref{trace}. Our goal, just as in 
the previous expression, will be to avoid matrix multiplication, as this has cubic 
complexity.

To achieve this, we first calculate vectors $\bold{z}_1 = \Sigma(\theta)^{-1} \bold{z}_n$ 
and $\bold{z}_2 = \dfrac{\partial \Sigma}{\partial \theta}(\theta) \bold{z}_1$, and 
finally compute:
\begin{equation*}
    \bold{z}_n^{\top} \Sigma(\theta)^{-1} \dfrac{\partial \Sigma}{\partial \theta}
    (\theta) \Sigma(\theta)^{-1} \bold{z}_n = \bold{z}_1^{\top} \bold{z}_2.
\end{equation*}

This gives us that, in total, the derivative calculation has a cubic cost, due to the 
calculation of the Cholesky decomposition $\Sigma(\theta) = U(\theta)^{\top} U(\theta)$ 
and its subsequent inverse. The LASSO penalty produces sparse solutions in the matrix 
$\Psi = \mathbf{LL}^{\top}$, and thus in the matrix $\Sigma(\theta)$. Therefore, to reduce 
computation time, we can apply methods that perform Cholesky decomposition for sparse 
matrices, such as the methods available in the libraries \texttt{Matrix} \citep{Matrix}, \texttt{spam} \citep{spam}, 
and \texttt{SparseM} \citep{SparseM}.

\subsection{Objective Function Calculation}

As in the previous section, the calculation of the objective function is quite similar 
for both the likelihood case and the composite likelihood case. For this reason, we will 
focus on the former case and later explain the differences to consider for the latter case.

The function to compute is:
\begin{equation*}
    \ell(\theta;\bold{z}_n) = -\frac{np}{2} \ln(2\pi) - \frac{1}{2} \ln(\left|
    \Sigma(\theta)\right|) - \frac{1}{2}\bold{z}_n^{\top} \Sigma(\theta)^{-1} 
    \bold{z}_n.
\end{equation*}

Assume we have the Cholesky decomposition $\Sigma(\theta) = U(\theta)^{\top} U(\theta)$, 
with $U(\theta)$ being an upper triangular matrix. To compute the determinant, note that:
\begin{equation*}
    \left|\Sigma(\theta)\right| = \left| U(\theta)^{\top} U(\theta) \right| = \left| 
    U(\theta)^{\top} \right| \left| U(\theta) \right| = \prod_{i=1}^{np} 
    U_{ii}(\theta)^2.
\end{equation*}

To calculate the quadratic form $\bold{z}_n^{\top} \Sigma(\theta)^{-1} \bold{z}_n$, 
assume we already have the matrix $U(\theta)$, which was calculated in the previous part. 
Then, we want to compute:
\begin{equation*}
\bold{z}_n^{\top} \Sigma(\theta)^{-1} \bold{z}_n = \bold{z}_n^{\top} U^{-1} \left( 
U(\theta)^{\top} \right)^{-1} \bold{z}_n = \left|\left| \left( U(\theta)^{\top} 
\right)^{-1} \bold{z}_n \right|\right|^2.
\end{equation*}
To avoid inverting $U(\theta)$, we solve the following system $U(\theta)^{\top} z = Z_n$, 
which requires $\mathcal{O}((np)^2)$ operations insofar as the matrix $U(\theta)^{\top}$ 
is lower triangular.

\subsection{Updating $L$}

Using the parameterization from Equation (2) of the main document, the multivariate Matérn model can be written as
\begin{equation*}
    C_{ij}(h) := \sqrt{\sigma_i^2 \sigma_j^2} \, \rho_{ij} \, 
    \frac{\alpha_{ii}^{\nu_{ii}}\alpha_{jj}^{\nu_{jj}}}{\alpha_{ij}^{2\nu_{ij}}} 
    \frac{\Gamma(\nu_{ij})}{\sqrt{\Gamma(\nu_{ii}) \Gamma(\nu_{jj})}} 
    \mathcal{M}(h; \alpha_{ij}, \nu_{ij}),
\end{equation*}
for $i,j=1,\dots,p$.

Define matrices $\mathcal{M}_{ij}$ as
\[
(\mathcal{M}_{ij})_{kl} = \mathcal{M}(\|s_k-s_l\|;\alpha_{ij},\nu_{ij}), \quad k,l=1,\dots,n,
\]
and the following matrices:
\[
\begin{aligned}
&\Psi_{ij} = \sigma_{ij}, \quad
A_{ii} = \alpha_{ii}^{\nu_{ii}}, \, A_{ij}=0 \ (i\neq j), \\
&\Gamma_{ii} = \frac{1}{\sqrt{\Gamma(\nu_{ii})}}, \, \Gamma_{ij}=0 \ (i\neq j), \quad
\overline{A}_{ij} = \frac{1}{\alpha_{ij}^{2\nu_{ij}}}, \quad
\overline{\Gamma}_{ij} = \Gamma(\nu_{ij}), \\
&\mathcal{M} = 
\begin{pmatrix}
\mathcal{M}_{11} & \cdots & \mathcal{M}_{1p} \\
\vdots & \ddots & \vdots \\
\mathcal{M}_{p1} & \cdots & \mathcal{M}_{pp}
\end{pmatrix}.
\end{aligned}
\]

Let $L$ be the Cholesky factor of $\Psi$, $\Psi = \mathbf{LL}^\top$, and let $\theta$ denote the vector of model parameters. Then the covariance matrix of the observations is
\begin{equation*}
    \Sigma(\theta) = \big([\Gamma A \mathbf{L L}^\top A \Gamma \odot \overline{A} \odot \overline{\Gamma}] 
    \otimes 1_{n\times n} \big) \odot \mathcal{M}.
\end{equation*}

To ensure that the matrix $\Psi=\mathbf{LL}^{\top}$ is positive semidefinite, we parameterize $\mathbf{L}$ 
as follows:
\begin{itemize}
    \item $\mathbf{L}_{ii}=\exp(l_{ii})$.
    \item $\mathbf{L}_{ij}=l_{ij}$ for $i\neq j$.
\end{itemize}

We then define the matrix $l$ formed by $l_{ij}$.

We have:
\begin{equation*}
     \frac{\partial \Sigma}{\partial l_{ij}}(\theta)=\left(\left[\Gamma A \left(
     \frac{\partial L}{\partial l_{ij}} \mathbf{L}^{\top}+\left(\frac{\partial \mathbf{L}}{\partial l_{ij}} 
     \mathbf{L}^{\top}\right)^{\top}\right) A\Gamma \odot \overline{A}\odot 
     \overline{\Gamma} \right] \otimes 1_{n\times n}\right) \odot 
     \mathcal{M}.
\end{equation*}

On the one hand, if $i=j$, we have:
\begin{equation}
    \left(\frac{\partial L}{\partial l_{ij}} \mathbf{L}^{\top}\right)_{k_1k_2}=\left\{ 
    \begin{array}{lcc} l_{ii} \mathbf{L}_{k_1 i} & \text{if} & k_2=i  \\ \\ 0  & 
    \text{otherwise.}  \end{array} \right.
\end{equation}
On the other hand, if $i\neq j$, we have:
\begin{equation}
    \left(\frac{\partial \mathbf{L}}{\partial l_{ij}} \mathbf{L}^{\top}\right)_{k_1k_2}=\left\{ 
    \begin{array}{lcc} \mathbf{L}_{k_1 j} & \text{if} & k_2=i   \\ \\ 0  & 
    \text{otherwise.}  \end{array} \right. 
\end{equation}

Furthermore, note that the function $R(l)=\lambda \sum_{i\neq j}|l_{ij}|$ is not 
differentiable, but it is subdifferentiable. Therefore, we need to perform a proximal 
gradient step using the Soft thresholding operator. In particular, we denote the Soft thresholding operator of 
dimension $p \times p$ as $S_{\lambda}$ as follows: For $i<j$, we have:
\begin{equation}
    (S_{\lambda}(u))_{ij}=\left(\text{prox}_{\lambda R}(l)\right)_{ij} =\left\{ 
    \begin{array}{lcc}
        l_{ij}- \lambda  & \text{if}  & l_{ij} >  \lambda \\
        \\ 0 & \text{if} &   |l_{ij}| \leq  \lambda \\
        \\ l_{ij}+ \lambda  & \text{if}  & l_{ij}  <  \lambda,
    \end{array}
    \right. 
\end{equation}
\noindent additionally, $(S_\lambda(l))_{ii})=l_{ii}$, and for $i>j$, we define 
$(S_\lambda(l))_{ij})=0$.

Now that the Soft thresholding operator is well defined for our problem, and considering 
that we can compute $\frac{\partial\Sigma}{\partial l_{ij}}(\theta)$ and, therefore, 
$\frac{\partial \ell}{\partial l_{ij}}(\theta;Z_n)$ for all $i,j$, we can calculate the 
proximal gradient. We then denote:
\begin{equation*}
    \left(\frac{\partial \ell}{\partial l}(\theta;\bold{z}_n)\right)_{ij}=
    \frac{\partial \ell}{\partial l_{ij}}(\theta;\bold{z}_n),
\end{equation*}
and let $l^{(k)}$ represent the $k$-th step of the iteration. Furthermore, assume we 
use a gradient step length $b_{step}$. We can then update the parameters $l^{(k)}$ 
through Algorithm \ref{alg:l}.

\begin{algorithm}
    \caption{Update of $l^{(k)}$}\label{alg:l}
    \begin{algorithmic}[1]
        \State Compute $\dfrac{\partial \ell}{\partial l^{(k)}}(\theta^{(k)};\bold{z}_n)$.
        \State Define $l'=S_{b_{step}\lambda}\left( l^{(k)}+b_{step} \dfrac{\partial \ell}
        {\partial l^{(k)}}(\theta^{(k)};\bold{z}_n)\right)$.
        \State Define $\mathbf{L}'$ as: 
        \begin{equation}
            (\mathbf{L}')_{ij}=\left\{ \begin{array}{lcc}
                e^{((l')_{ii})}  & \text{if}  & i=j \\
                \\(l')_{ij} & \text{otherwise}
            \end{array}
            \right.
        \end{equation}
        \State Define the matrix $D$ such that:
        \begin{equation}
            D_{ij}=\left\{ \begin{array}{lcc}
                \sqrt{\sum_{k=1}^{i} (\mathbf{L}_{ik}')^2}  & \text{if}  & i=j \\
                \\ 0 & \text{otherwise}
            \end{array}
            \right.
        \end{equation}  
        \State Update $\mathbf{L}^{(k)}=D\mathbf{L}'$. 
    \end{algorithmic}
\end{algorithm}

\subsection{Updating $\Delta_B$}
To update $\Delta_B$, we implement a projected gradient step with a line search, i.e., 
given $\Delta_B^{(k)}$, we use Algorithm \ref{alg:delb}.

\begin{algorithm}
    \caption{Update of $\Delta_B^{(k)}$}\label{alg:delb}
    \begin{algorithmic}[1]
        \State Define $\theta^{(act)}=(\mathbf{L}^{(k+1)},\Delta_B^{(k)},\mathbf{R}_B^{(k)})$.
        \State Compute $\ell_{\lambda}(\theta^{(act)})$.
        \State Compute $\dfrac{\partial \ell_{\lambda}}{\partial\Delta_B}(\theta^{(act)})$.
        \For{ $i=1,\dots,M$ }
            \State Define $\Delta_B'=\max\Big(0,\Delta_B^{(k)}+(a_{\Delta_B})^{i}
            \dfrac{\partial \ell_{\lambda}}{\partial \Delta_B}(\Delta_B^{(k)})\Big)$. 
            \State Define $\theta'=(\mathbf{L}^{(k+1)},\Delta_B',\mathbf{R}_B^{(k)})$.
            \State Compute $\ell_{\lambda}(\theta')$.
            \If{$\ell_{\lambda}(\theta') > \ell_{\lambda}(\theta^{(act)})$}
                \State $\Delta_B^{(k+1)}=\Delta_B'$.
                \State Break
            \EndIf
        \EndFor 
    \end{algorithmic}
\end{algorithm}

\subsection{Updating $\mathbf{R}_B$}

Similarly to $\Delta_B$, we implement a projected gradient step with a line search, i.e., 
given $\mathbf{R}_B^{(k)}$, we use Algorithm \ref{alg:R_B}.

\begin{algorithm}
\caption{Update of $\mathbf{R}_B^{(k)}$}\label{alg:R_B}
\begin{algorithmic}[1]
    \State Define $\theta^{(act)}=(\mathbf{L}^{(k+1)},\Delta_B^{(k+1)},\mathbf{R}_B^{(k)})$ and compute 
    $\ell_{\lambda}(\theta^{(act)})$.
    \State Compute $\dfrac{\partial \ell_{\lambda}}{\partial \mathbf{R}_B}(\theta^{(act)})$.
    \For{$i=1,\dots,M$}
        \State Define $\mathbf{R}_B' = \mathbf{R}_B^{(k)} + (b_{\mathbf{R}_B})^{i} \dfrac{\partial \ell_{\lambda}}{\partial R_B} (\theta^{(act)})$.
        \State Project $\mathbf{R}_B'$ to the space of correlation matrices with non-negative entries.
        \State Compute $\ell_{\lambda}(\mathbf{R}_B')$.
        \If{$\ell_{\lambda}(\mathbf{R}_B') > \ell_{\lambda}(\theta^{(act)})$} 
            \State $\mathbf{R}_B^{(k+1)} = \mathbf{R}_B'$.
            \State Break
        \EndIf
    \EndFor 
\end{algorithmic}
\end{algorithm}

\section{Composite Likelihood: Asymptotic Theorems}\label{asintotics theorem}

In this section, we will demonstrate the consistency of the composite likelihood 
estimator for multivariate random fields, that is, the estimator produced by minimizing 
the function $pl_{p}$. The proof is analogous to the proof for composite likelihood 
functions presented in \citet{bevilacqua2016composite} and \citet{bevilacqua2012estimating}. 
Specifically, we assume that $\bold{Z}=\{(Z_1(s), \dots, Z_p(s))^{\top}: s \in \mathcal{D} 
\subsetneq \mathbb{R}^d\}$ where $\mathcal{D}$ is a grid, not necessarily regular, of 
$\mathbb{R}^d$. The results are based on assumptions about the domain $\mathcal{D}$. We 
first assume that $\mathcal{D}$ is infinite in cardinality and locally finite, i.e., 
$|\mathcal{D}|=\infty$ and for every $s \in \mathcal{D}$ and for every $r>0$ we have 
$|\mathcal{B}(s,r) \cup \mathcal{D}|=\mathcal{O}(r^d)$, with $B(s,r)$ being the ball of 
dimension $d$, centered at $s$ with radius $r$.

We will use results on least squares estimators presented in \citet{gaetan2010spatial}. 
Specifically, we assume the asymptotic behavior of increasing domains, meaning that the 
number of observations grows as the size of the domain increases. This behavior is 
represented through an increasing sequence of sets $\mathcal{R}_{N} \subset \mathcal{D}$, 
with $\left|\mathcal{R}_{N}\right| \rightarrow \infty$ as $N \rightarrow \infty$.

If we consider the weights $w_{kl}\in \{0, 1\}$, it is not necessary to include all 
possible pairs of spatial locations. In fact, we define the following finite set 
$\mathbb{M}=\left\{\boldsymbol{h}_{1}, \ldots, \boldsymbol{h}_{q}\right\}$, with 
$\boldsymbol{h}_{j} \in \mathbb{R}^{d}$. The set $\mathbb{M}$ does not contain the 
origin and determines the pairs of geographic locations involved in the estimation 
equation. The set $\mathbb{M}$ represents the neighborhood of each spatial location. 
Additionally, we will denote $\stackrel{\circ}{\mathcal{R}}_{N}=\left\{s \in 
\mathcal{R}_{N}: s+h \in \mathcal{R}_{N}, \text{ for every } \boldsymbol{h} \in 
\mathbb{M}\right\}$, the interior of $\mathcal{R}_{N}$. Note that $\left|\mathfrak{R}_{N}
\right| / \left|\mathcal{R}_{N}\right| \rightarrow 1$ as $N \rightarrow \infty$. 
Furthermore, we will denote $\boldsymbol{X}_{s(s+h)}=\left(\bold{Z}(s), \bold{Z}(s+h)
\right)^{\top}$, to which we associate the covariance matrix $Q(h ;\theta)$. We then 
consider the objective function:

\begin{equation}
    \mathrm{CL}(\theta) = -\frac{1}{2} \sum_{s \in \mathfrak{\mathcal{R}}_{N}}\left[
    \sum_{h \in \mathbb{M}} \log \left|Q(h; \theta)\right|+
    \boldsymbol{X}_{s(s+h)}^{\top} Q^{-1}(h ; \theta) \boldsymbol{X}_{s(s+h)}\right]. 
    \label{CL} 
\end{equation}

\subsection{Consistency}

To demonstrate the consistency and asymptotic normality of the composite likelihood 
estimator, we write the composite likelihood \eqref{CL} as a contrast process:
\begin{equation}
    \tilde{U}_{N}(\theta) = -\dfrac{2}{\left|\stackrel{\circ}{\mathcal{R}}_{N}\right|} 
    \mathrm{CL}(\theta),
\end{equation}
with $\text{CL}(\cdot)$ defined in equation \eqref{CL}. Therefore, the estimator 
$\text{CL}$ is defined as the value $\widehat{\theta}_{N}$, such that $\widehat{\theta}_{N}
=\argmin_{\theta \in \Theta} \tilde{U}_{N}(\theta)$.

\textbf{Proposition:} Assume that:
\begin{itemize}
\item[C1] $\bold{Z}$ is a second-order Gaussian, stationary, and ergodic random field.
\item[C2] $\Theta$ is a compact subset of $\mathbb{R}^{k}$ and the true value $\theta_{0}$ 
belongs to the interior of $\Theta$.
\item[C3] For every $h \in \mathbb{M}$, the mappings $\theta \mapsto C_{ij}(h, \theta)$ 
and $\theta \mapsto C_{ij}(0,\theta)$ have continuous partial derivatives with 
$\inf_{\theta \in \Theta} \left|Q(h ; \theta)\right| > 0$.
\item[C4] The function
\begin{equation}
    \theta \mapsto K\left(\theta, \theta_{0}\right) = \sum_{h \in \mathcal{M}}\left\{
    \log \left|Q(h;{\theta})\right| + \text{tr}\left(Q\left(h;\theta_{0}\right) 
    Q^{-1}(h;\theta)\right)\right\} \label{Kfuncion}
\end{equation}
has a unique global minimum over $\Theta$ at $\theta_{0}$.
\end{itemize}
Then, the estimator $\widehat{\theta}_{N} \to \theta_{0}$ as $N \to \infty$, in 
probability.

\textbf{Remark:} C4 is an identifiability condition. Under Gaussianity, for every 
$\boldsymbol{h} \in \mathbb{M}$, the function $\theta \mapsto \log \left|Q(h ; \theta)
\right| + \text{tr}\left(Q\left(h; \theta_{0}\right) Q^{-1}(h; \theta)\right)$ has a 
global minimum at $\theta_{0}$ due to the Kullback-Leibler inequality, but in the 
multidimensional case $(k>1)$, we may have that $\theta_{0}$ is not necessarily the 
only minimum.

\begin{proof}
Let us denote the empirical covariance matrix by 
\[
\widehat{Q}(h) = \left|\stackrel{\circ}{\mathcal{R}}_{N}\right|^{-1} 
\sum_{s \in \stackrel{\circ}{\mathcal{R}}_{N}} X_{s(s+h)} X_{s(s+h)}^{T},
\] 
with components $\widehat{C}_{ij}(\cdot)$, $i,j=1,\dots,p$. Then, the contrast process can be written as
\begin{equation*}
\tilde{U}_N(\boldsymbol{\theta}) = \sum_{h \in \mathcal{M}} \log \left| Q(h;\theta) \right| 
+ \operatorname{tr}\left( \widehat{Q}(h) Q^{-1}(h;\theta) \right).
\end{equation*}

Following the Theorem 3.4.1 in \citet{guyon1995random} the consistency of the $CL$ estimator can be established by showing the following properties:

\begin{enumerate}
    \item By conditions (C2) and (C3), the mapping $\theta \mapsto Q(h;\theta)$ is continuous with continuous derivative. Since $U_N(\theta)$ and $K(\theta,\theta_0)$ (defined in Equation \eqref{Kfuncion}) are compositions of continuous functions, they are themselves continuous.
    
    \item The weak law of large numbers together with condition C1 implies that, as $N \to \infty$, 
    \[
    \widehat{Q}(h) \rightarrow Q(h;\theta_0) \quad \text{in probability}.
    \]
    Therefore, $\tilde{U}_N$ converges in probability to $K(\theta,\theta_0)$ (defined in \eqref{Kfuncion}), which has a unique minimum at $\theta_0$ by condition C4.
    
    \item Define the modulus of continuity of $\tilde{U}_N(\theta)$ as
    \begin{equation}
        W_N(\eta) = \sup_{\|\alpha-\beta\|\leq \eta} \left| \tilde{U}_N(\alpha) - \tilde{U}_N(\beta) \right|, \quad \alpha, \beta \in \Theta.
    \end{equation}
    Then we have
    \begin{equation}
        \left| \tilde{U}_N(\alpha) - \tilde{U}_N(\beta) \right| 
        \leq \sum_{h \in \mathcal{M}} 
        \left| \log \frac{|Q(h;\alpha)|}{|Q(h;\beta)|} \right|
        + \left| \operatorname{tr} \left( \widehat{Q}(h) \left[ Q^{-1}(h;\alpha) - Q^{-1}(h;\beta) \right] \right) \right|.
    \end{equation}
    By conditions (C2) and (C3), the functions $\theta \mapsto Q(h;\theta)$ are Lipschitz; then, by composition of Lipschitz functions, we have
    \begin{equation}
        \sum_{h \in \mathcal{M}} \left| \log \frac{|Q(h;\alpha)|}{|Q(h;\beta)|} \right| 
        \leq \left[ L_1 + \sum_{h \in \mathcal{M}} L_2 \right] \| \alpha - \beta \|.
    \end{equation}
    Moreover, by the Cauchy-Schwarz inequality and Lipschitz continuity, we have
    \begin{align}
        \sum_{h \in \mathcal{M}} \left| \operatorname{tr} \left( \widehat{Q}(h) \left[ Q^{-1}(h;\alpha) - Q^{-1}(h;\beta) \right] \right) \right|
        &\leq \sum_{h \in \mathcal{M}} \left| \operatorname{tr}(\widehat{Q}(h)) \right| \left| Q^{-1}(h;\alpha) - Q^{-1}(h;\beta) \right| \\
        &\leq \sum_{h \in \mathcal{M}} \left| \operatorname{tr}(\widehat{Q}(h)) \right| L_3 \| \alpha - \beta \|.
    \end{align}
    Noting that $\sum_{h \in \mathcal{M}} \left| \operatorname{tr}(\widehat{Q}(h)) \right| \leq L_4 \sum_{i=1}^p \widehat{C}_{i}(0)$, combining with the previous results we obtain constants $L_1,L_2,L_5$ such that
    \begin{equation}
        \left| \tilde{U}_N(\alpha) - \tilde{U}_N(\beta) \right| 
        \leq \left[ L_1 + \sum_{h \in \mathcal{M}} L_2 + L_5 \sum_{i=1}^p \widehat{C}_i(0) \right] \| \alpha - \beta \|. \label{Lipch}
    \end{equation}
    Define the random variable
    \begin{equation}
        M_N = L_1 + \sum_{h \in \mathcal{M}} L_2 + L_5 \sum_{i=1}^p \widehat{C}_i(0),
    \end{equation}
    which satisfies $E[M_N] < \infty$ by condition (C1). Furthermore, equation \eqref{Lipch} guarantees the existence of a sequence $\epsilon_k > 0$ with $\epsilon_k \to 0$ such that
    \begin{equation}
        \lim_{l \to \infty} P\left( W_N(1/l) \geq \epsilon_l \right) = 0,
    \end{equation}
    analogous to the univariate case (see details in \cite{bevilacqua2016composite}).
\end{enumerate}
\end{proof}

\subsection{Asymptotic Normality}

\textbf{Proposition:} Assume that conditions (C1)-(C4) hold. Additionally, suppose that:
\begin{itemize}
    \item[AN.1] $Z$ is a multivariate $\alpha$-mixing random field with mixing coefficient 
    $\alpha(\cdot)=\alpha_{\infty, \infty}(\cdot)$ such that:
    \begin{itemize}
        \item[AN.1-1] There exists $\delta>0$ such that:
        \begin{equation*}
            \sum_{s_{m}, s_{n} \in \stackrel{\circ}{\mathcal{R}}_{N}} \alpha\left(
            \text{dist}\left(s_{m}, s_{n}\right)\right)^{\delta /(2+\delta)}=\mathcal{O}
            \left(\left| \mathscr{\mathcal{R} }_{N} \right|\right).
        \end{equation*}
        \item[AN.1-2] $\sum_{l \geq 0} l^{d-1} \alpha(l)<\infty $.
    \end{itemize}
    \item[AN.2] For all $i, j=1, \ldots, p$, and $h \in \mathbb{M}_{i j}$, the mappings 
    $\theta \to C_{i j}(h,\theta)$ and $\theta \to C_{i j}({0}, {\theta})$ have 
    second-order partial derivatives that are continuous with respect to $\theta$.
    \item[AN.3] There exist $i, j$ and $h \in \mathbb{M}_{i j}$ such that the matrix of 
    order $(p^2 \times k)$ whose $l$-th column is given by $vec\left[\partial \Sigma_{i j}
    ({h} ; {\theta}) / \partial \theta_{l}\right]$, for $l=1, \ldots, k$, has rank $k$, 
    where $vec[\cdot]$ is the vectorization of a matrix that converts the matrix into a 
    column vector. Furthermore, suppose there exists a positive semidefinite matrix $J$ 
    such that for large $N$, we have:
    \begin{equation}
        \left|\stackrel{\circ}{\mathcal{R}}_{N}\right|E\left[\nabla \tilde{Q}_{N}(\theta) 
        \nabla \tilde{Q}_{N}(\theta)\right]\geq J,
    \end{equation}
    where if $A$ and $B$ are two symmetric matrices, $A \geq B$ means that $A-B$ is a 
    positive semidefinite matrix.
    
    Then,
    \begin{equation*}
        \left|\stackrel{\circ}{\mathcal{R}}_{N}\right|^{1 / 2} J_{N}^{-1 / 2}\left(
        \theta_{0}\right) H_{N}\left(\theta_{0}\right)\left(\widehat{\theta}_{N}-\theta_{
        \mathbf{0}}\right) \rightarrow \mathcal{N}\left(0, I_{p}\right),
    \end{equation*}
    in distribution, where $J_{N}(\theta):=\left|\stackrel{\circ}{\mathcal{R}}_{N}\right| 
    E\left[\nabla \tilde{Q}_{N}(\theta) \nabla \tilde{Q}_{N}(\theta)^{\top}\right]$ and 
    $\mathcal{H}_{N}(\theta):=\mathbb{E}\left[\nabla^{2} \tilde{Q}_{N}(\theta)\right]$.
\end{itemize}

\begin{proof}
The asymptotic normality is obtained by verifying the general conditions (AN.1-AN.2-AN.3) 
of Theorem 3.4.5 in \citet{guyon1995random}:

\begin{itemize}
    \item[N.1]  By condition AN.2 and Equation \eqref{Lipch}, there exist three positive constants $L_{1}, L_{2}, L_{3}$ such that
        \begin{equation*}
        \sup_{\theta \in \Theta} \left|\frac{\partial^{2} \tilde{Q}_{N}(\theta)}{\partial \theta_{l} \partial \theta_{r}}\right| \leq L_{1} + \sum_{h \in \mathcal{M}} L_{2} + L_{5} \sum_{i=1}^{p} \widehat{C}_{i}(0),
        \end{equation*}
    where the right-hand side is a real-valued random variable, since we work under second-order properties.
    
    \item[N.2] Notice that by AN.1 and the Central Limit Theorem for weakly dependent random variables \citep{spodarev2013stochastic}, we have
    \[
        |\stackrel{\circ}{\mathcal{R}}_{N}|^{1/2} J_{N}^{-1/2} \nabla \tilde{Q}_{N}(\boldsymbol{\theta}_0) \rightarrow N(0, I_p)
    \]
    in distribution as $N \to \infty$.
    
    \item[N.3] Consider the $(l,r)$ elements of the matrix $H_{N}(\theta_0) = E\left[\nabla^2 \tilde{Q}_N(\theta_0)\right]$:
    \begin{align*}
        H_{N, lr}(\theta_0) &= E\left[\frac{\partial^{2} \tilde{Q}_{N}(\boldsymbol{\theta}_0)}{\partial \theta_l \partial \theta_r}\right] \\
        &= \sum_{h \in \mathcal{M}} \operatorname{tr}\left(\Sigma_{ij}^{-1}(h;\theta_0) \frac{\partial \Sigma_{ij}(h;\theta_0)}{\partial \theta_l} \Sigma_{ij}^{-1}(h;\theta_0) \frac{\partial \Sigma_{ij}(h;\theta_0)}{\partial \theta_r}\right) \\
        &= \sum_{h \in \mathcal{M}} \operatorname{tr}\left(A_{ij}^{(l)}(h;\theta_0) A_{ij}^{(r)}(h;\theta_0)\right),
    \end{align*}
    where $A_{ij}^{(l)}(h;\theta_0) = B_{ij}(h;\theta_0) \frac{\partial \Sigma_{ij}(h;\theta_0)}{\partial \theta_l} B_{ij}(h;\theta_0)$ is a symmetric matrix for $l = 1, \dots, p$, and $B_{ij}(h;\theta_0)$ is a positive definite matrix such that $B_{ij}(h;\theta_0)^2 = \Sigma_{ij}^{-1}(h;\theta_0)$.  
    
    Using the identity $\operatorname{tr}(X^\top Y) = \operatorname{vec}(X)^\top \operatorname{vec}(Y)$, we can write
    \begin{equation*}
        \mathcal{H}_{N, lr}(\theta_0) = \sum_{h \in \mathcal{M}} \operatorname{vec}[A_{ij}^{(l)}(h;\theta_0)]^\top \operatorname{vec}[A_{ij}^{(r)}(h;\theta_0)],
    \end{equation*}
    where $\operatorname{vec}[A_{ij}^{(r)}(h;\theta_0)] \in \mathbb{R}^{4p^2}$. Moreover, using $\operatorname{vec}(ABC) = (C^\top \otimes A) \operatorname{vec}(B)$, we have
    \begin{equation*}
        H_{N, lr}(\theta_0) = \sum_{h \in \mathcal{M}} \operatorname{vec}\left[\frac{\partial \Sigma_{ij}(h;\theta_0)}{\partial \theta_l}\right]^\top \left(B_{ij}(h;\theta_0) \otimes B_{ij}(h;\theta_0)\right)^2 \operatorname{vec}\left[\frac{\partial \Sigma_{ij}(h;\theta_0)}{\partial \theta_r}\right].
    \end{equation*}
    Hence, $H_N(\theta_0)$ can be written as
    \begin{equation*}
        \mathcal{H}_N(\theta_0) = \sum_{h \in \mathcal{M}} V_{ij}^\top(h;\theta_0) U_{ij}(h;\theta_0) V_{ij}(h;\theta_0),
    \end{equation*}
    where $U(h;\theta_0) := \left(B_{ij}(h;\theta_0) \otimes B_{ij}(h;\theta_0)\right)^2$ is a positive definite matrix of size $4p^2 \times 4p^2$, and $V^\top(h;\theta_0)$ is a matrix of size $q \times 4p^2$, with row $i$ given by $\operatorname{vec}\left[\frac{\partial \Sigma(h;\theta_0)}{\partial \theta_i}\right]^\top$, for $i = 1, \dots, q$.
    
    Therefore, for any $x \in \mathbb{R}^q$ with $x \neq \mathbf{0}$:
    \begin{align*}
        x^\top H_N(\theta_0) x &= \sum_{h \in \mathcal{M}} x^\top V^\top(h;\theta_0) U(h;\theta_0) V(h;\theta_0) x \\
        &= \sum_{h \in \mathcal{M}} \left[V(h;\theta_0) x\right]^\top U(h;\theta_0) \left[V(h;\theta_0) x\right].
    \end{align*}
\end{itemize}

Since $U(h;\theta_0)$ is positive definite, it follows that $x^\top H_N(\theta_0) x > 0$ if and only if $V(h;\theta_0)x \neq 0$ for all $x \neq 0$, which holds by assumption AN.3. Therefore, by the weak law of large numbers, $\nabla \tilde{Q}_N - H_N(\theta_0)$ converges in probability to $0$.
\end{proof}

\section{Calculation of Matrices $J$ and $H$}\label{Eq:H and J definition}

\begin{equation}
    H(\theta)=E\left[\nabla^2 pl_{p}\right]  \quad \quad   \text{and} \quad \quad  
    J(\theta)=E\left[\nabla pl_{p}(\theta)\nabla pl_{p}(\theta)^{\top}\right],\quad \quad  
    \theta \in \Theta \subset \mathbb{R}^{q}.
\end{equation}

\begin{equation}
    \nabla cl_{kl}^{p}= \left[  -\frac{1}{2}\left(\operatorname{tr}\left(Q_{kl} 
    \frac{\partial Q_{kl}}{\partial \theta_{h}}\right)+(\bold{Z}(s_k),\bold{Z}(s_l))^{
    \top}Q_{kl}^{-1}\frac{\partial Q_{kl}}{\partial \theta_{h}}Q_{kl}^{-1}(\bold{Z}(s_k),
    \bold{Z}(s_l))\right)\right]_{h=1}^{q}.
\end{equation}

\begin{equation}
    H(\theta)=\left[\frac{1}{2}\sum_{k=1}^{n-1}\sum_{l=k+1}^{n} \operatorname{tr}\left(
    Q_{kl}^{-1}\frac{\partial Q_{kl}}{\partial \theta_{h_1}} Q_{kl}^{-1}\frac{\partial 
    Q_{kl}}{\partial \theta_{h_2}}\right)\right]_{h_1,h_2=1}^{q}
\end{equation}

\begin{align}
    J(\theta)&=\left[\frac{1}{4} \sum_{k=1}^{n-1}\sum_{l=k+1}^{n} \sum_{k'=1}^{n-1}
    \sum_{l'=k'+1}^{n} \left[ \operatorname{tr}\left(\left[ Q_{kl}-
    Q_{kl}^{-1}\right]\frac{\partial Q_{kl}}{\partial \theta_{h_1}} \right)  
    \operatorname{tr}\left(\left[Q_{k'l'}-Q_{k'l'}^{-1}\right]
    \frac{\partial Q_{k'l'}}{\partial \theta_{h_2}} \right) \right. \right. \\ 
    & + 2 \operatorname{tr}\left( \left[Q_{kl}^{-1}\frac{\partial Q_{kl}}{\partial 
    \theta_{h_1}}Q_{kl}^{-1} \right] \operatorname{cov}\left((\bold{Z}(s_k),\bold{Z}
    (s_l)),(\bold{Z}(s_{k'}),\bold{Z}(s_{l'}))\right) \right. \\ 
    & \left. \left. \left[Q_{k'l'}^{-1} \frac{\partial Q_{k'l'}}{\partial \theta_{h_2}} 
    Q_{k'l'}^{-1}  \right] \operatorname{cov}\left((\bold{Z}(s_k),\bold{Z}(s_l)),
    (\bold{Z}(s_{k'}),\bold{Z}(s_{l'}))\right) \right)  \right]_{h_1,h_2=1}^{q}.
\end{align}

\section{Data analysis}\label{appendix:data_analysis}

\begin{table}[H]
    \centering
    \begin{tabular}{@{}lllllll@{}}
        \toprule
        & Variable & Symbol & Variable & Symbol & Variable & Symbol \\ 
        \midrule
        Major elements & Iron & Fe & Magnesium & Mg & Sodium & Na \\
        & Calcium & Ca & Titanium & Ti & Potassium & K \\
        Trace elements & Phosphorus & P & Aluminum & Al & Sulfur & S \\
        & Molybdenum & Mo & Uranium & U & Barium & Ba \\
        & Copper & Cu & Thorium & Th & Boron & B \\
        & Lead & Pb & Strontium & Sr & Tungsten & W \\
        & Zinc & Zn & Cadmium & Cd & Scandium & Sc \\
        & Silver & Ag & Antimony & Sb & Thallium & Tl \\
        & Nickel & Ni & Bismuth & Bi & Mercury & Hg \\
        & Cobalt & Co & Vanadium & V & Selenium & Se \\
        & Manganese & Mn & Lanthanum & La & Tellurium & Te \\
        & Arsenic & As & Chromium & Cr & Gallium & Ga \\
        \bottomrule
    \end{tabular}
    \caption{Assayed elements with their chemical symbols.}
    \label{cuadrito}
\end{table}

\begin{table}[H]
    \centering
    \begin{tabular}{|l|lll|l|lll|}
        \hline
        Variables & $\hat{\tau}^2$ & $\hat{\sigma}^2$ & $\hat{\tau}^2/(\hat{\tau}^2+\hat{\sigma}^2)$ & Variables 
        & $\hat{\tau}^2$ & $\hat{\sigma}^2$ & $\hat{\tau}^2/(\hat{\tau}^2+\hat{\sigma}^2)$ \\
        \hline
        Fe & 0.730 & 0.381 & 65.7\% & U & 0.293 & 0.747 & 28.2\% \\
        Ca & 0.588 & 0.494 & 54.3\% & Th & 0.445 & 0.667 & 40.0\% \\
        P & 0.579 & 0.428 & 57.5\% & Sr & 0.460 & 0.554 & 45.3\% \\
        Mg & 0.451 & 0.801 & 36.0\% & Cd & 0.705 & 0.379 & 65.0\% \\
        Ti & 0.483 & 0.604 & 44.4\% & Sb & 0.162 & 0.924 & 14.9\% \\
        Al & 0.376 & 0.670 & 35.9\% & Bi & 0.369 & 0.710 & 34.2\% \\
        Na & 0.366 & 0.425 & 46.2\% & V & 0.518 & 0.549 & 48.6\% \\
        K & 0.195 & 0.887 & 18.0\% & La & 0.320 & 0.813 & 28.3\% \\
        S & 0.721 & 0.428 & 62.7\% & Cr & 0.639 & 0.378 & 62.8\% \\
        Mo & 0.219 & 0.883 & 19.9\% & Ba & 0.447 & 0.720 & 38.3\% \\
        Cu & 0.071 & 1.055 & 6.3\% & B & 1.096 & 0.122 & 90.0\% \\
        Pb & 0.279 & 0.739 & 27.4\% & W & 0.263 & 1.201 & 18.0\% \\
        Zn & 0.519 & 0.690 & 42.9\% & Sc & 0.013 & 1.016 & 1.3\% \\
        Ag & 0.463 & 0.623 & 42.6\% & Tl & 0.671 & 0.467 & 59.0\% \\
        Ni & 0.618 & 0.454 & 57.7\% & Hg & 0.350 & 0.718 & 32.8\% \\
        Co & 0.165 & 1.020 & 13.9\% & Se & 0.241 & 0.846 & 22.2\% \\
        Mn & 0.590 & 0.542 & 52.2\% & Te & 0.279 & 0.828 & 25.2\% \\
        As & 0.335 & 0.721 & 31.7\% & Ga & 0.570 & 0.517 & 52.4\% \\
        \hline
    \end{tabular}
    \caption{Estimated nugget ($\hat{\tau}^2$), partial sill ($\hat{\sigma}^2$), and nugget-to-sill ratio for each variable obtained from the marginal random field estimate.}
    \label{Tab:Varianza_y_Nugget}
\end{table}

\subsection{Exploratory data analysis}

Figure \ref{visualizacion_variables} shows the spatial distribution of the four 
variables of interest after normal score transformation.

\begin{figure}[H]
    \centering
    \begin{minipage}{.5\textwidth}
        \centering
        \includegraphics[scale=0.25]{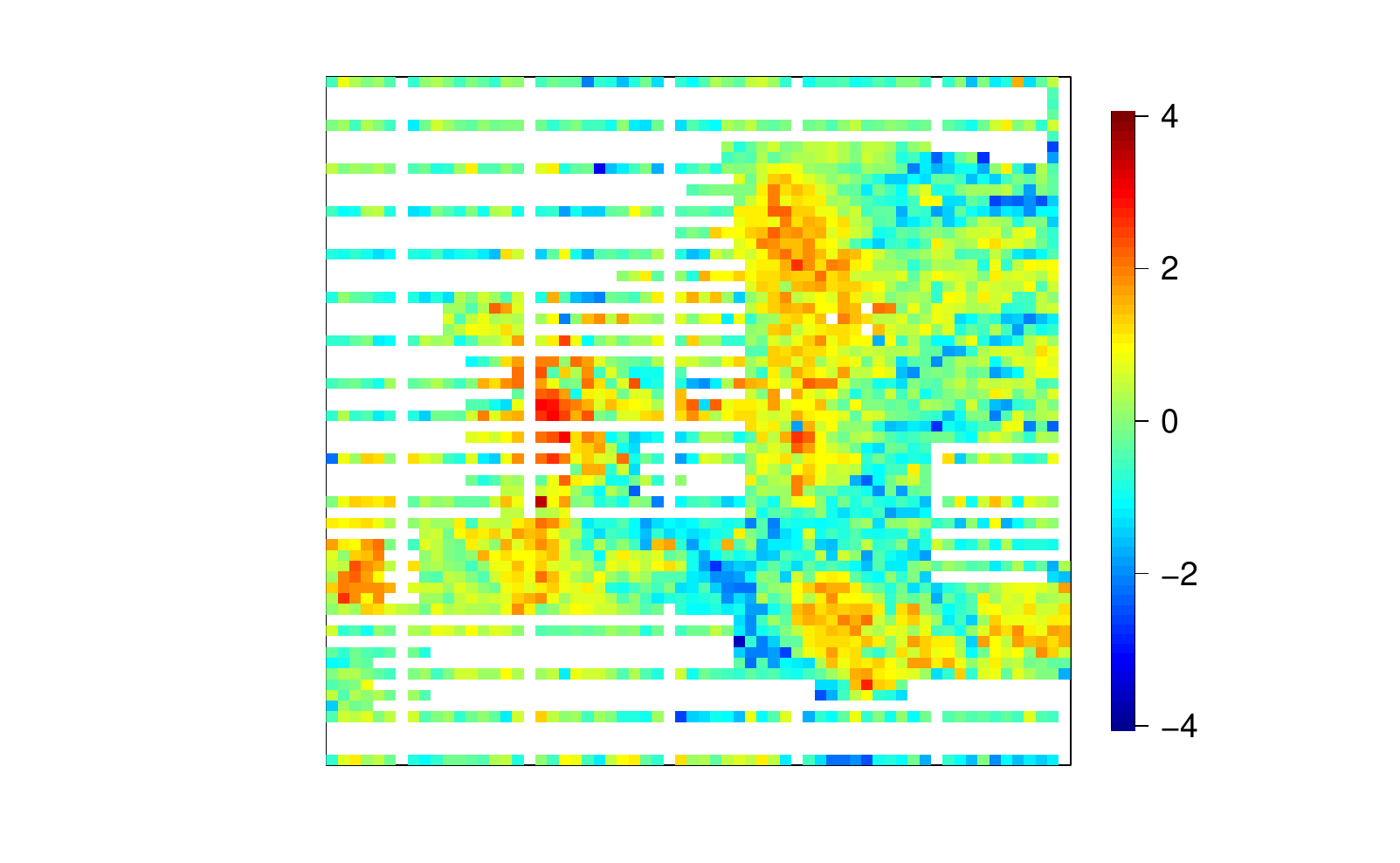}
        (a) Cu
    \end{minipage}%
    \begin{minipage}{.5\textwidth}
        \centering
        \includegraphics[scale=0.25]{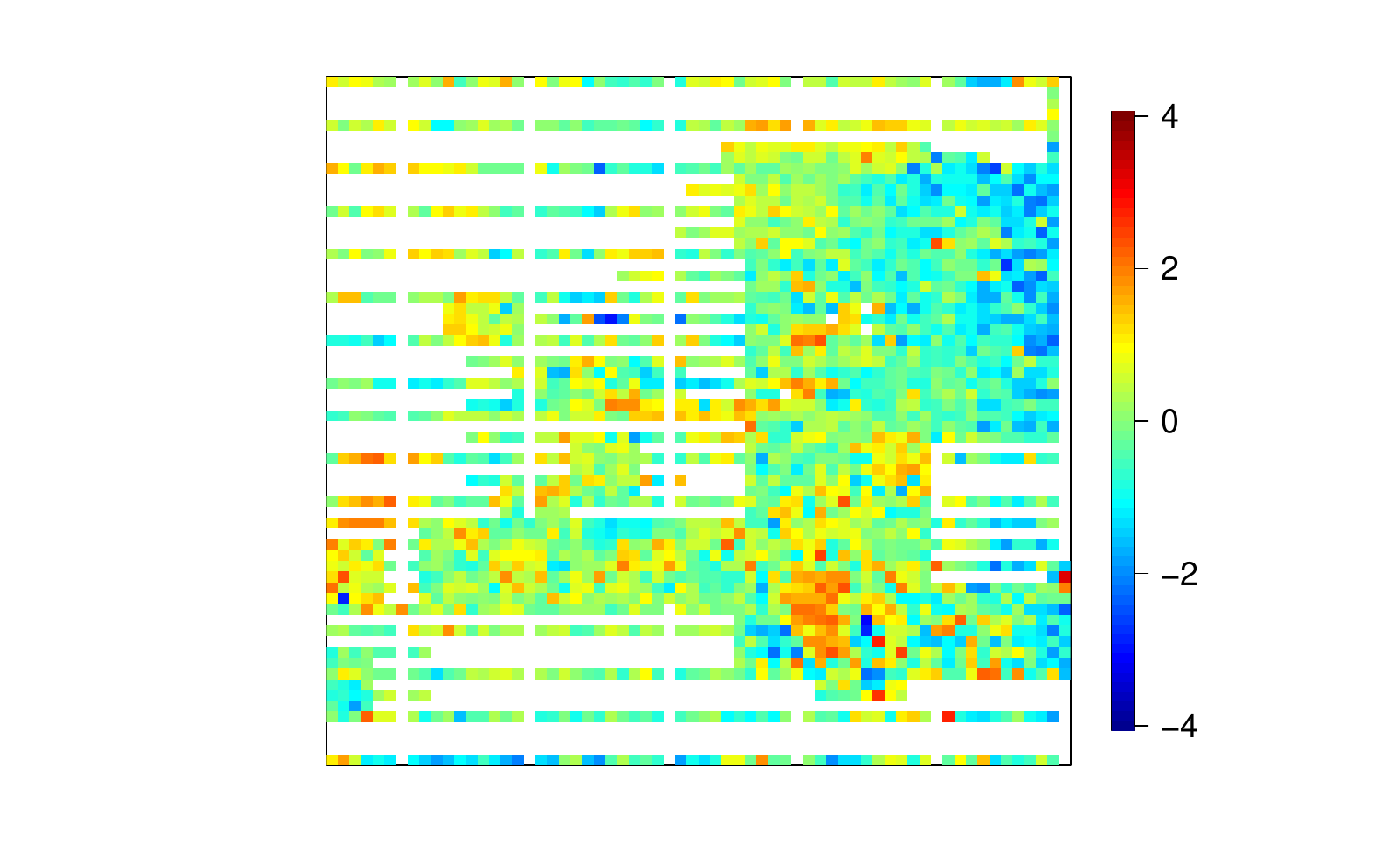}
        (b) Fe
    \end{minipage}
    
    \begin{minipage}{.5\textwidth}
        \centering
        \includegraphics[scale=0.25]{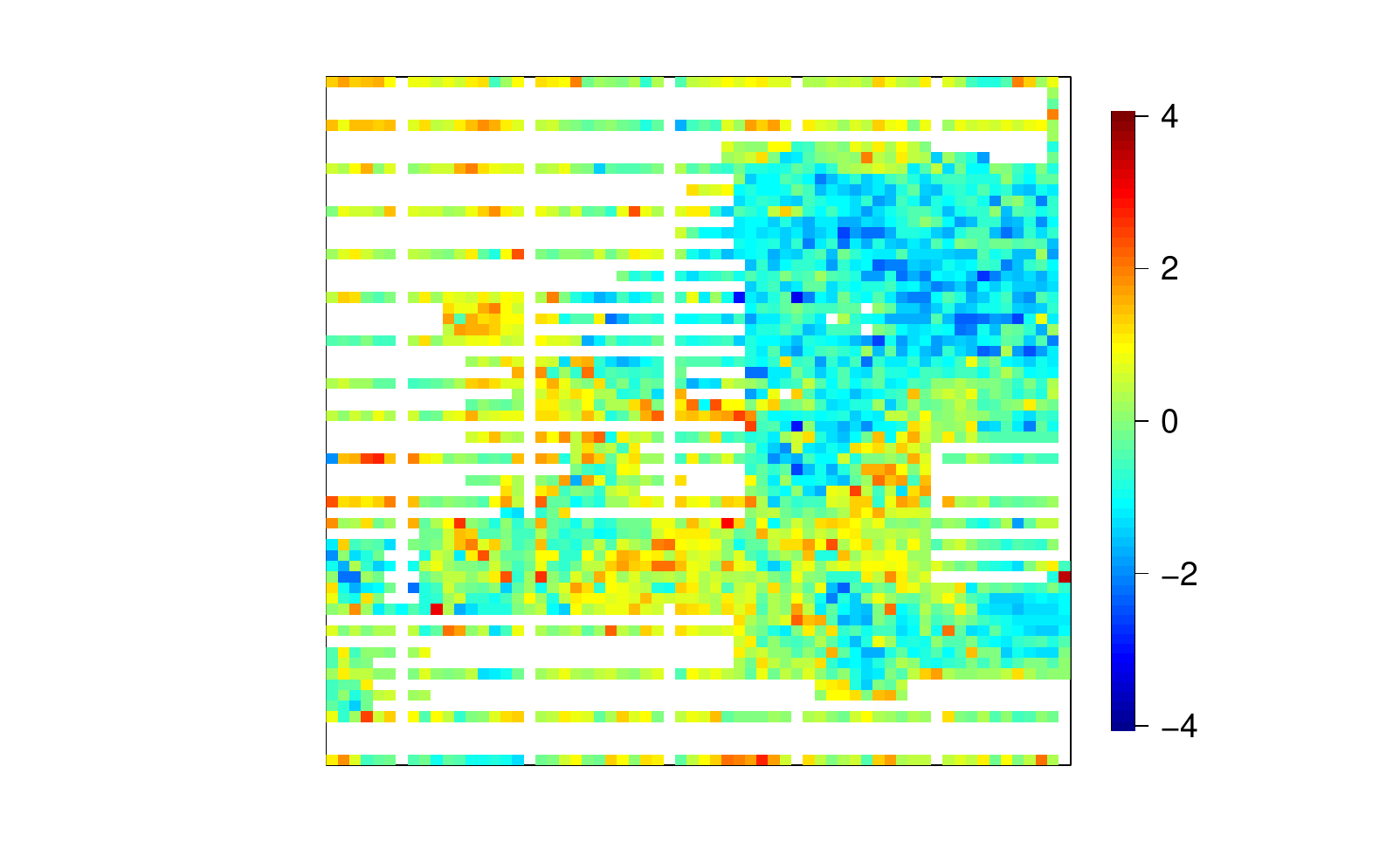}
        (c) Co
    \end{minipage}%
    \begin{minipage}{.5\textwidth}
        \centering
        \includegraphics[scale=0.25]{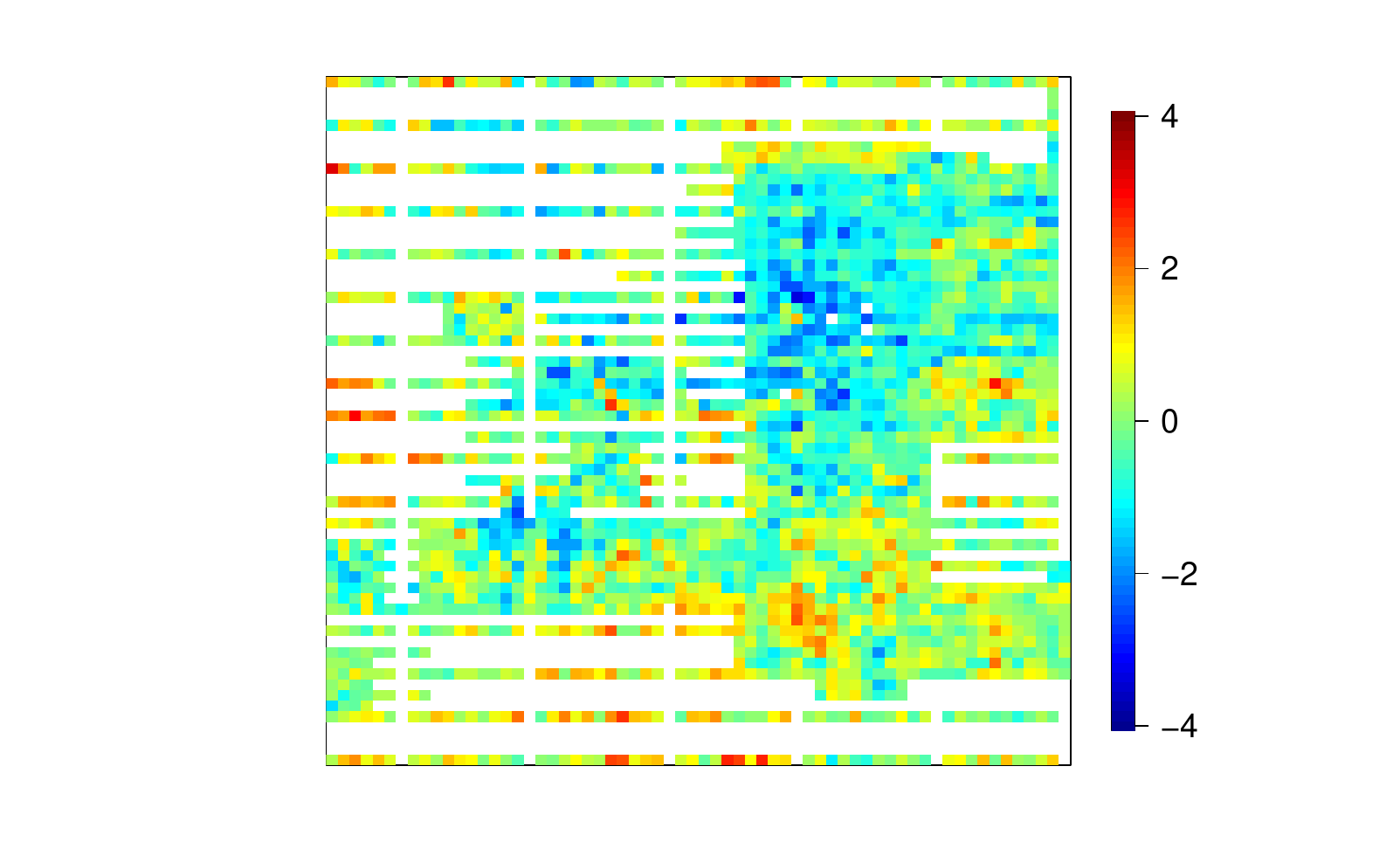}
        (d) Al
    \end{minipage}
    
    \caption{Spatial distribution of the variables of interest after normal score 
    transformation.}\label{visualizacion_variables}
\end{figure}

Figure \ref{visualizacion_variogramas} shows the empirical cross-variograms of the 
variables of interest with the other variables. Irregular behavior can be observed 
in several cross-variograms, with some taking both positive and negative values and 
others remaining close to zero, suggesting that some variables are spatially 
uncorrelated with one another.

\begin{figure}[H]
    \centering
    \begin{minipage}{.3\textwidth}
        \centering
        \includegraphics[scale=0.23]{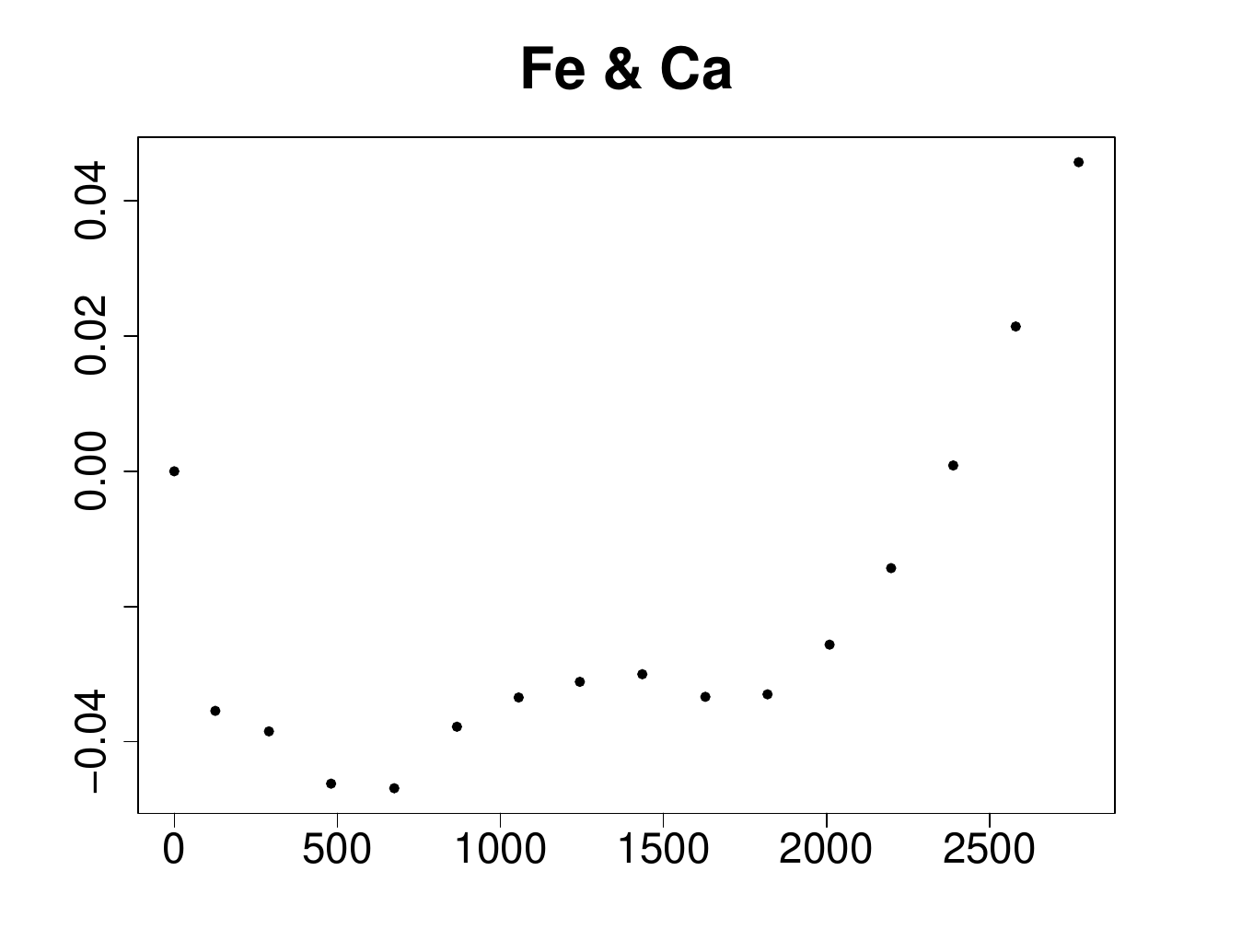}
    \end{minipage}%
    \begin{minipage}{.3\textwidth}
        \centering
        \includegraphics[scale=0.23]{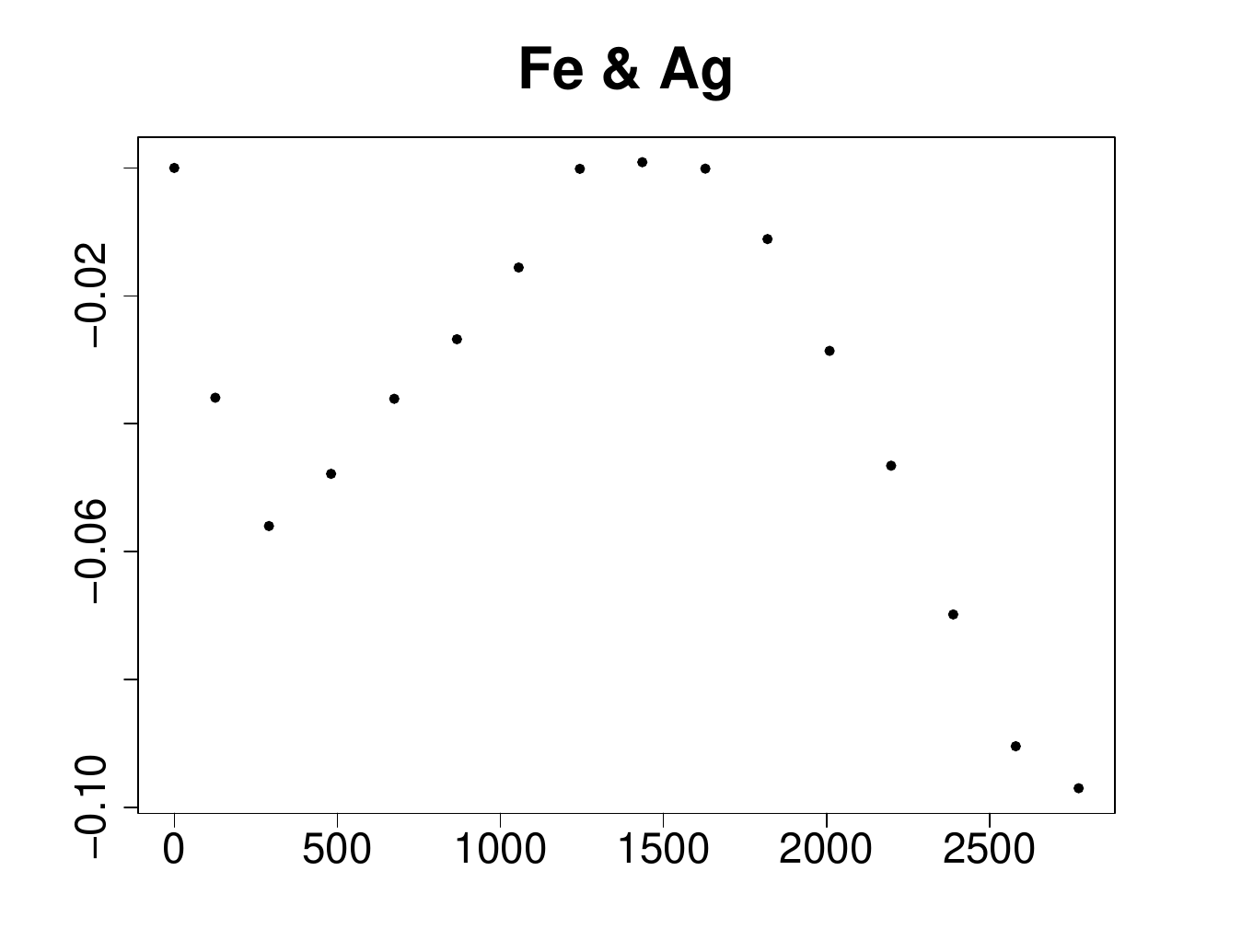}
    \end{minipage}%
    \begin{minipage}{.3\textwidth}
        \centering
        \includegraphics[scale=0.23]{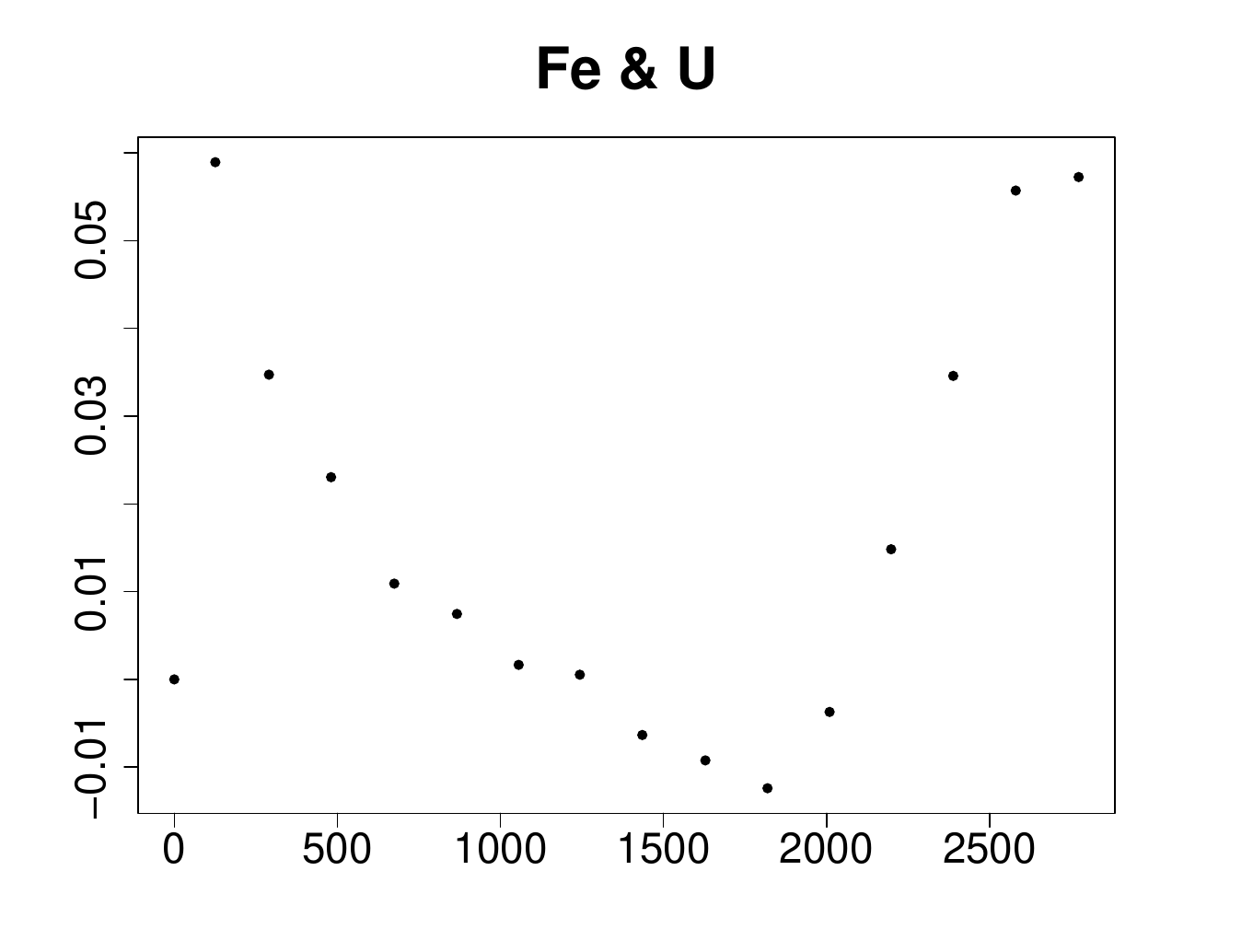}
    \end{minipage}
    
    \begin{minipage}{.3\textwidth}
        \centering
        \includegraphics[scale=0.23]{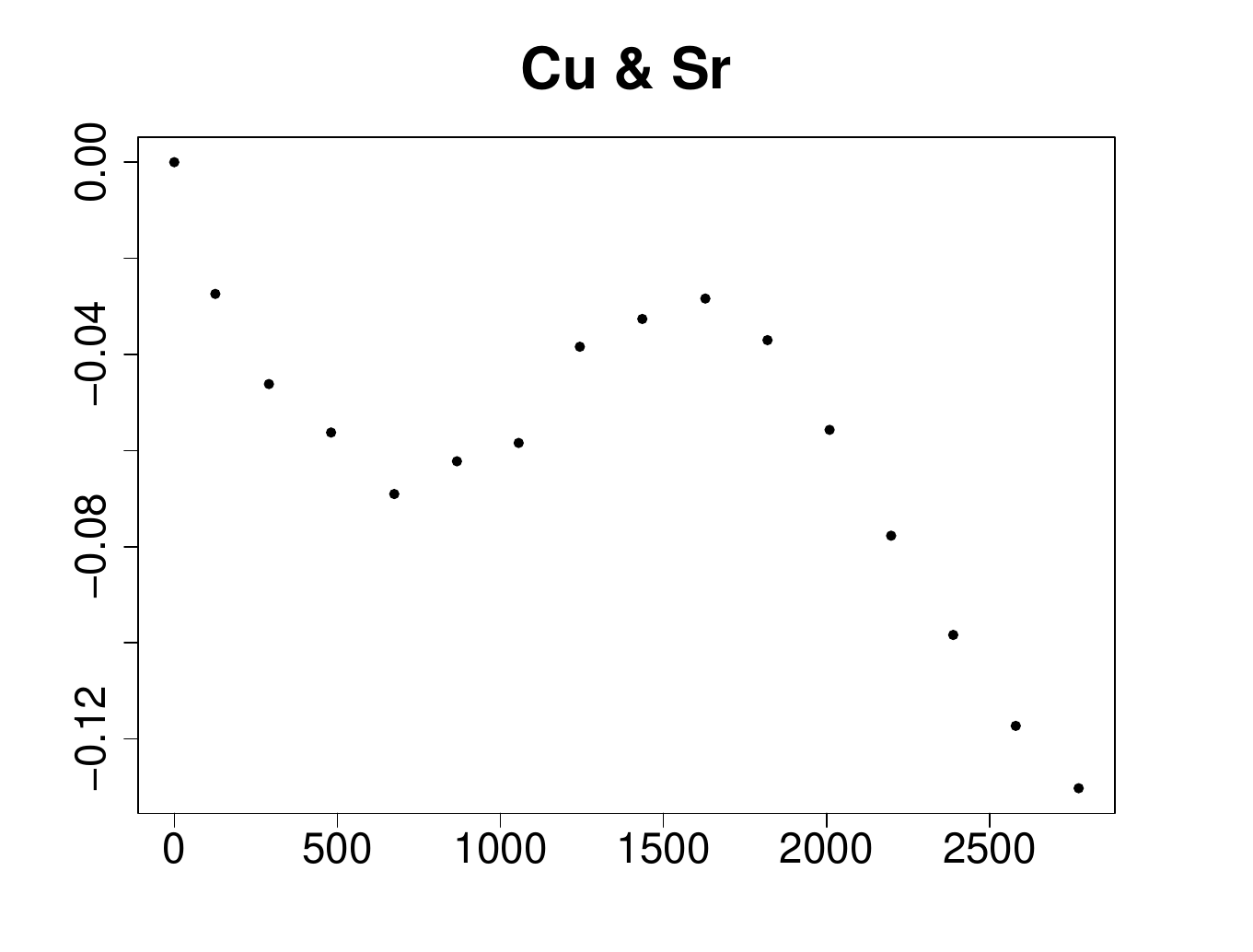}
    \end{minipage}%
    \begin{minipage}{.3\textwidth}
        \centering
        \includegraphics[scale=0.23]{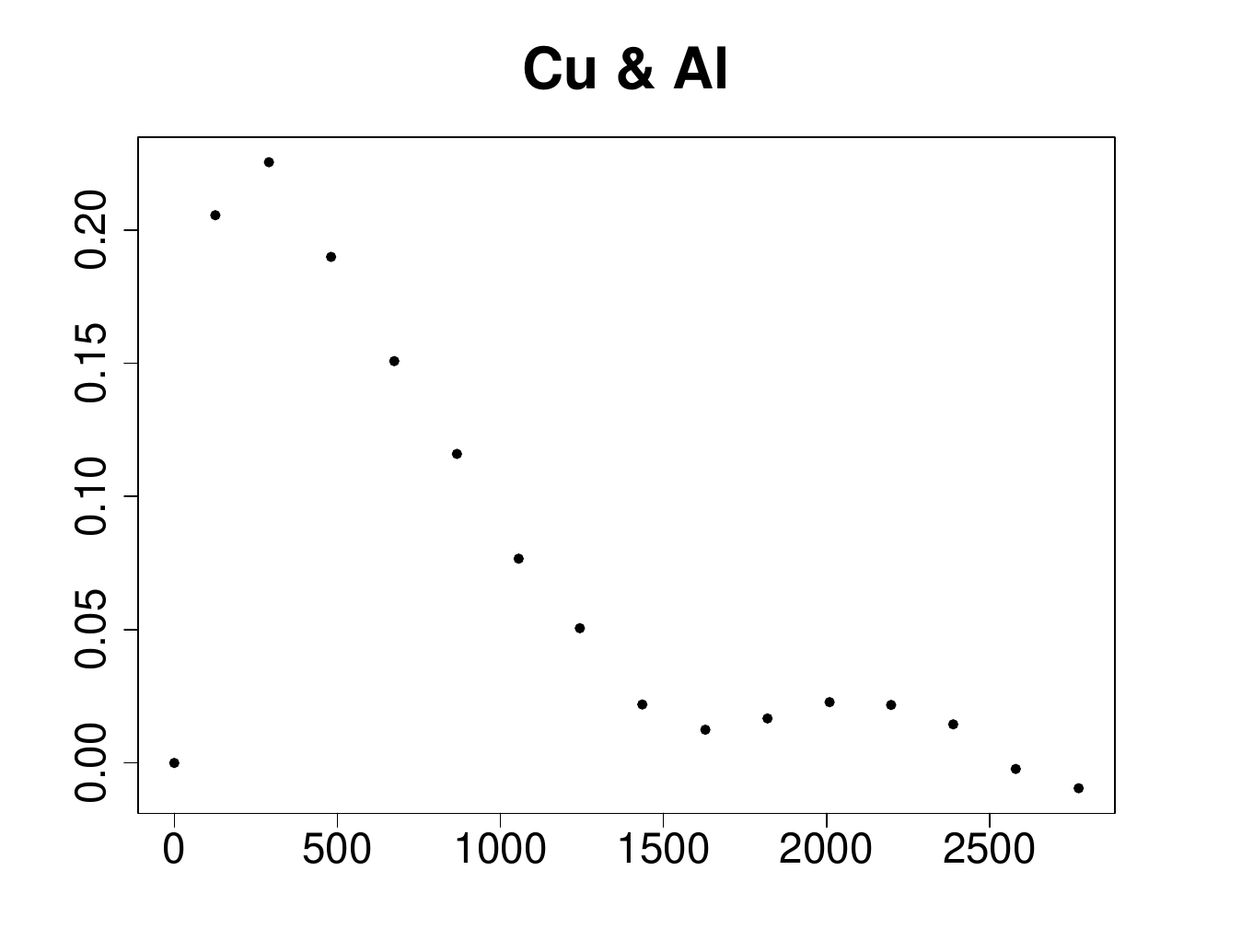}
    \end{minipage}%
    \begin{minipage}{.3\textwidth}
        \centering
        \includegraphics[scale=0.23]{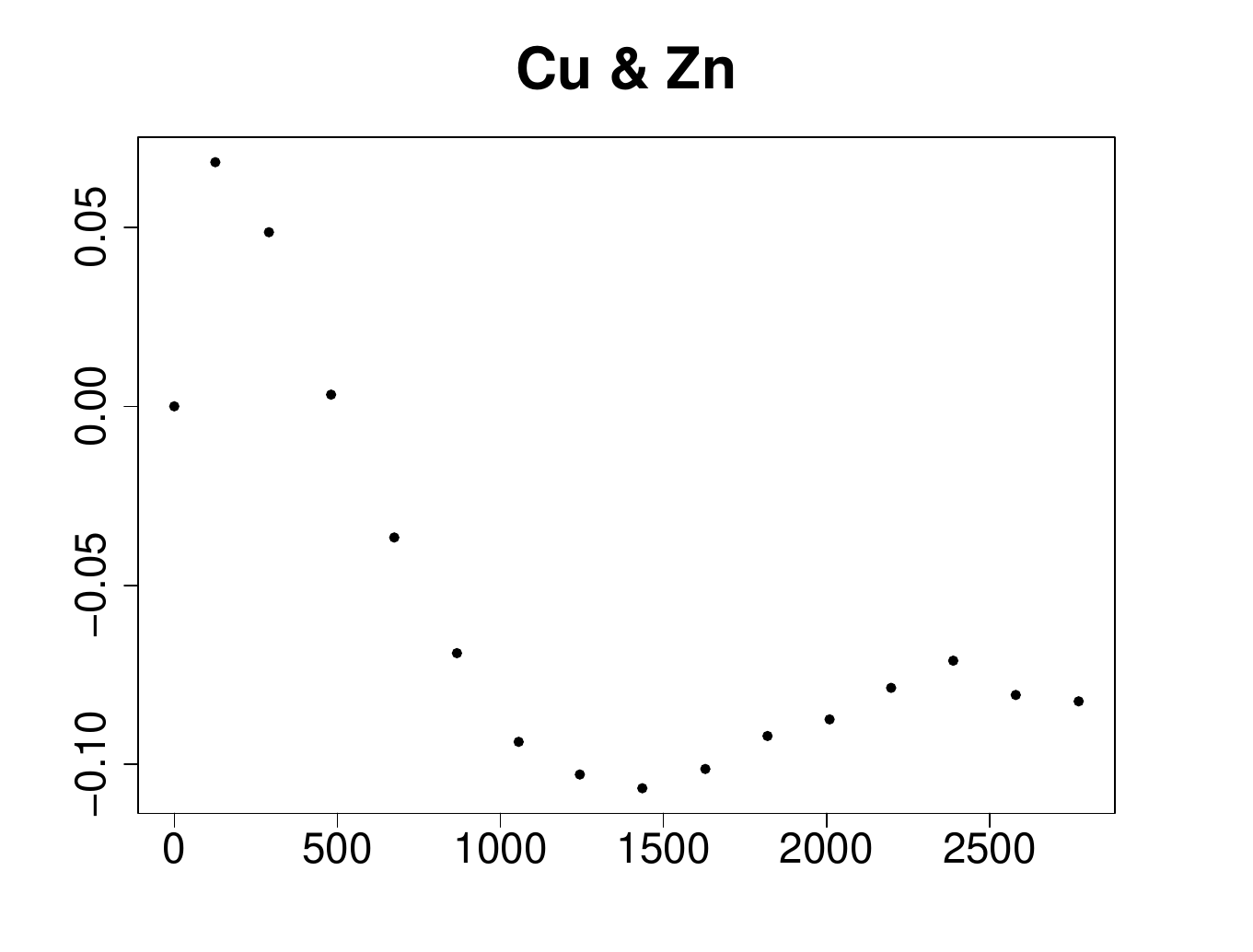}
    \end{minipage}
    
    \begin{minipage}{.3\textwidth}
        \centering
        \includegraphics[scale=0.23]{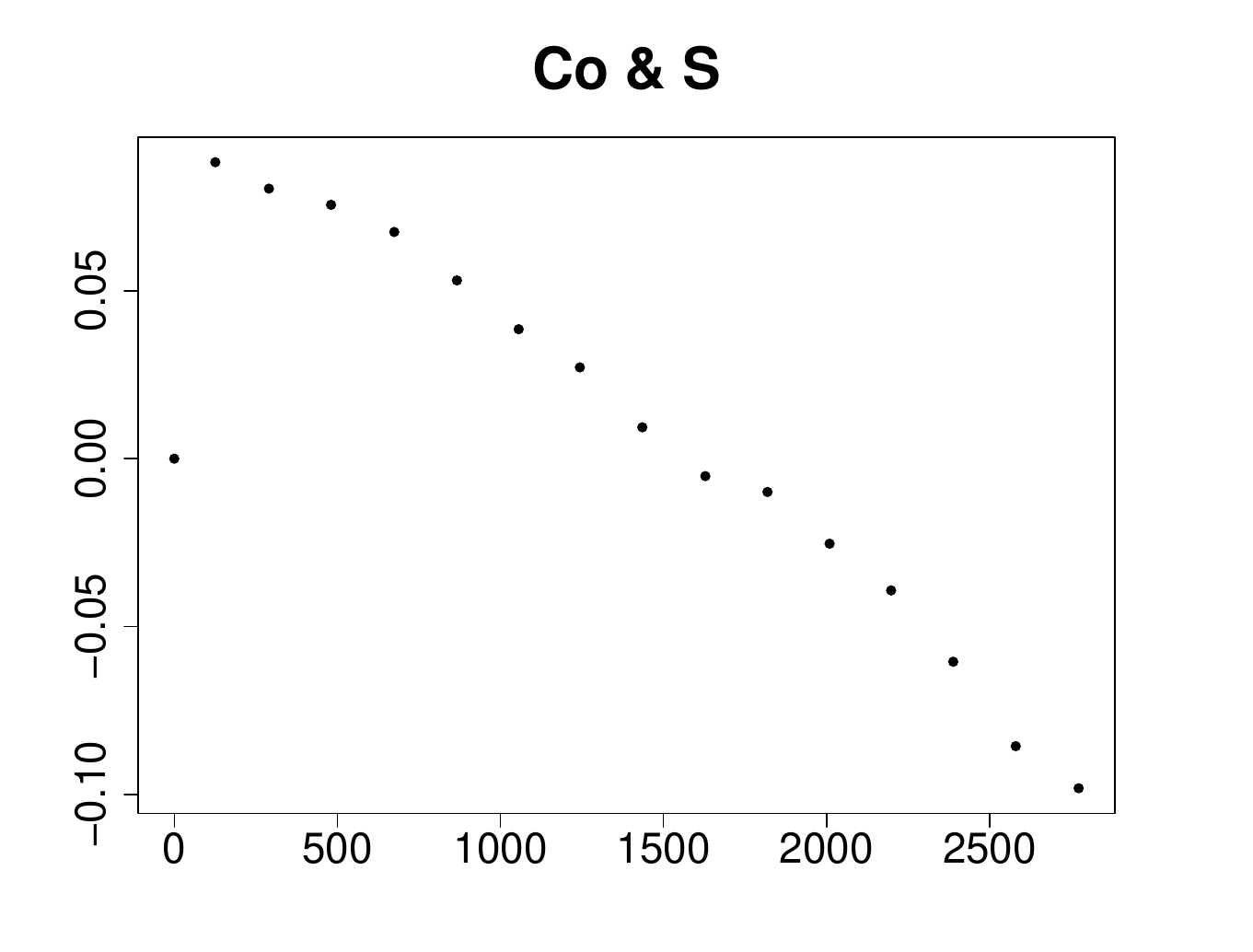}
    \end{minipage}%
    \begin{minipage}{.3\textwidth}
        \centering
        \includegraphics[scale=0.23]{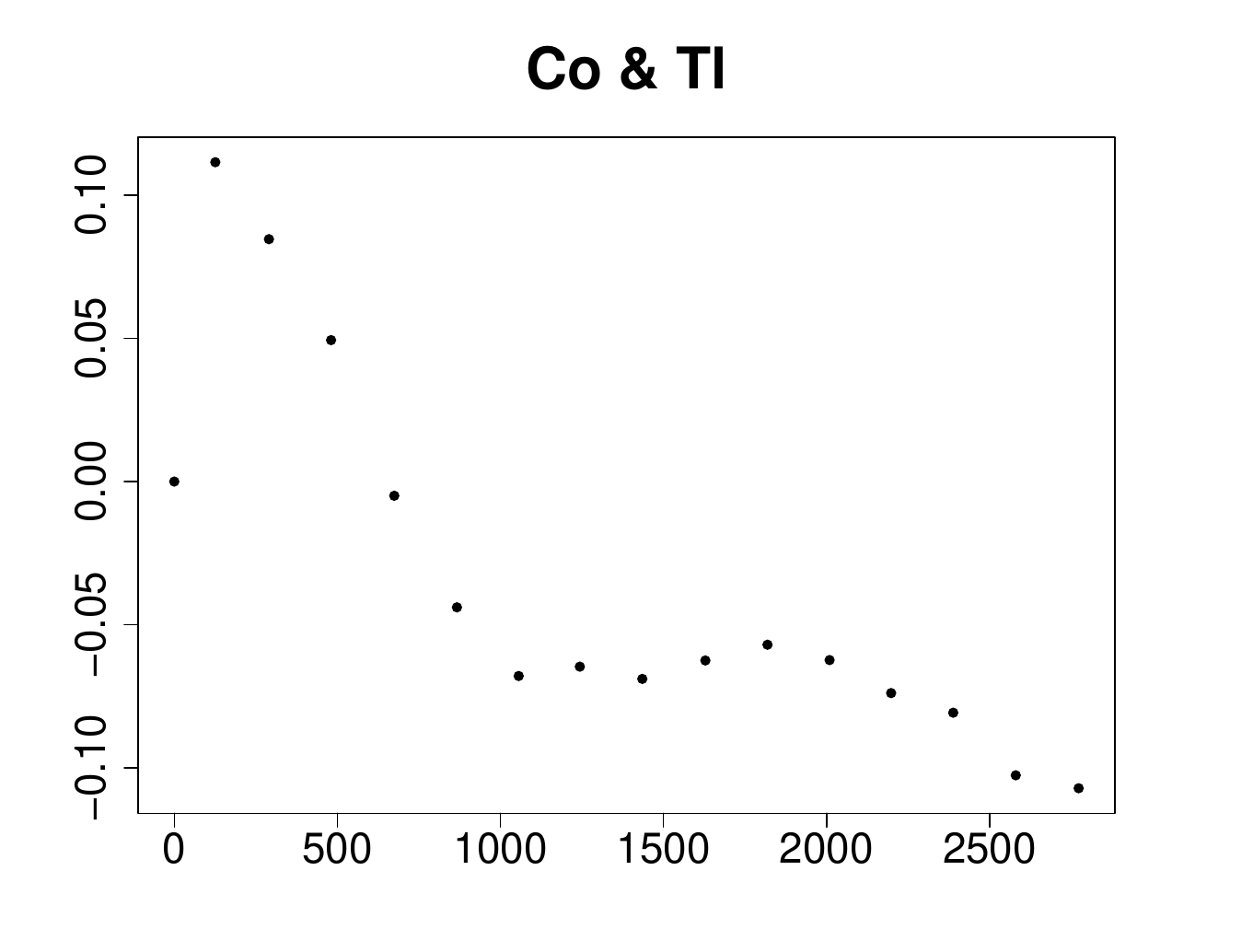}
    \end{minipage}%
    \begin{minipage}{.3\textwidth}
        \centering
        \includegraphics[scale=0.23]{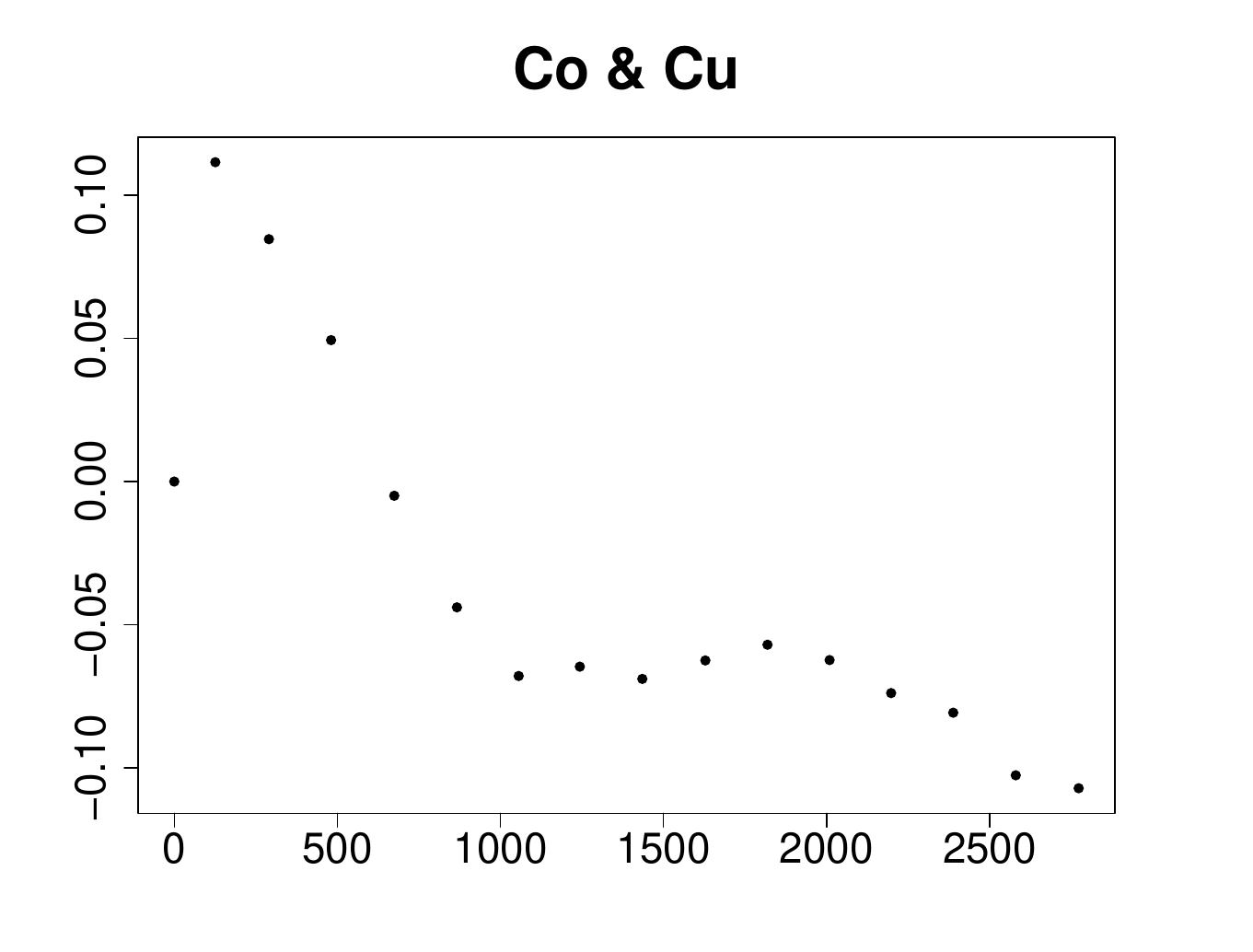}
    \end{minipage}
    
    \begin{minipage}{.3\textwidth}
        \centering
        \includegraphics[scale=0.23]{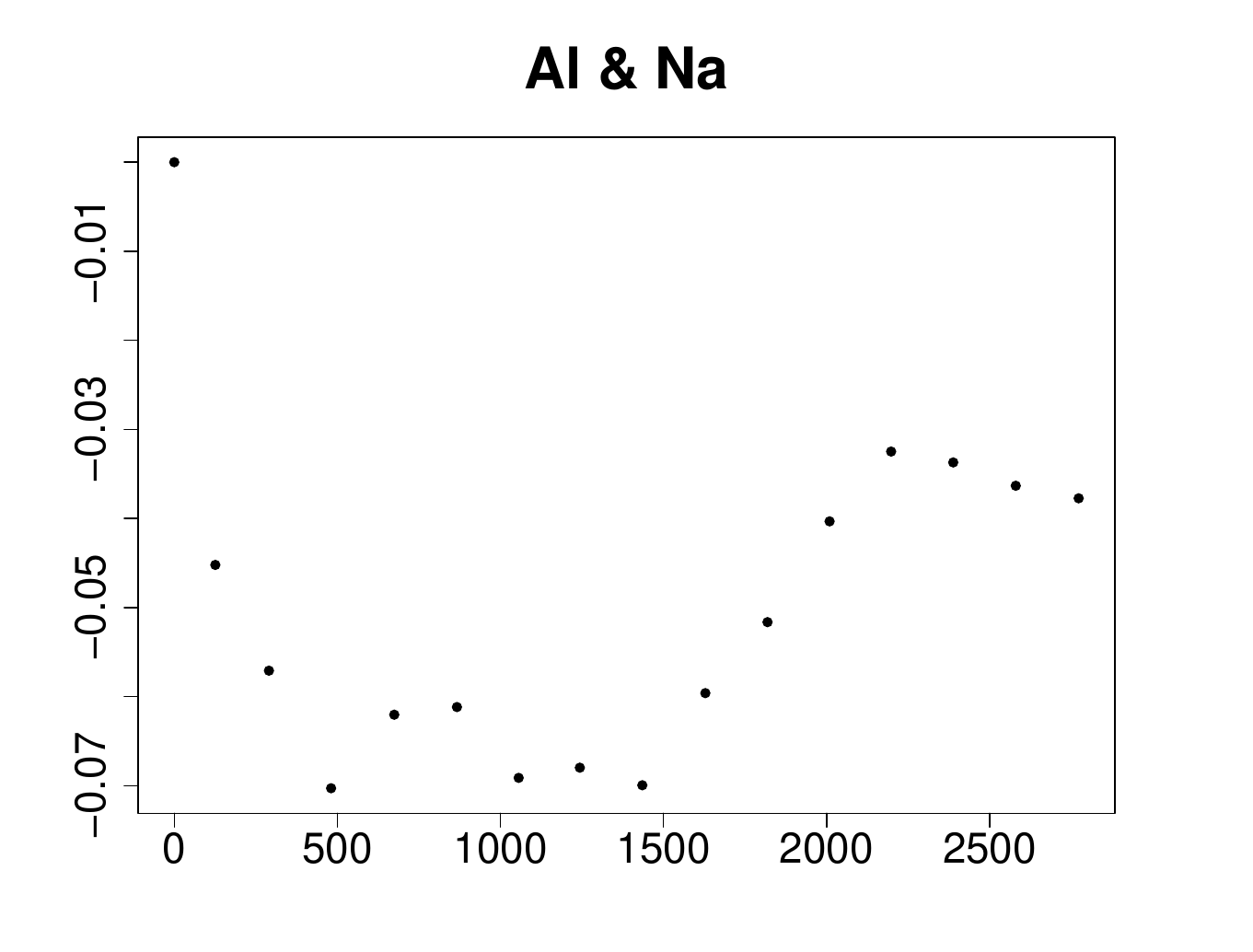}
    \end{minipage}%
    \begin{minipage}{.3\textwidth}
        \centering
        \includegraphics[scale=0.23]{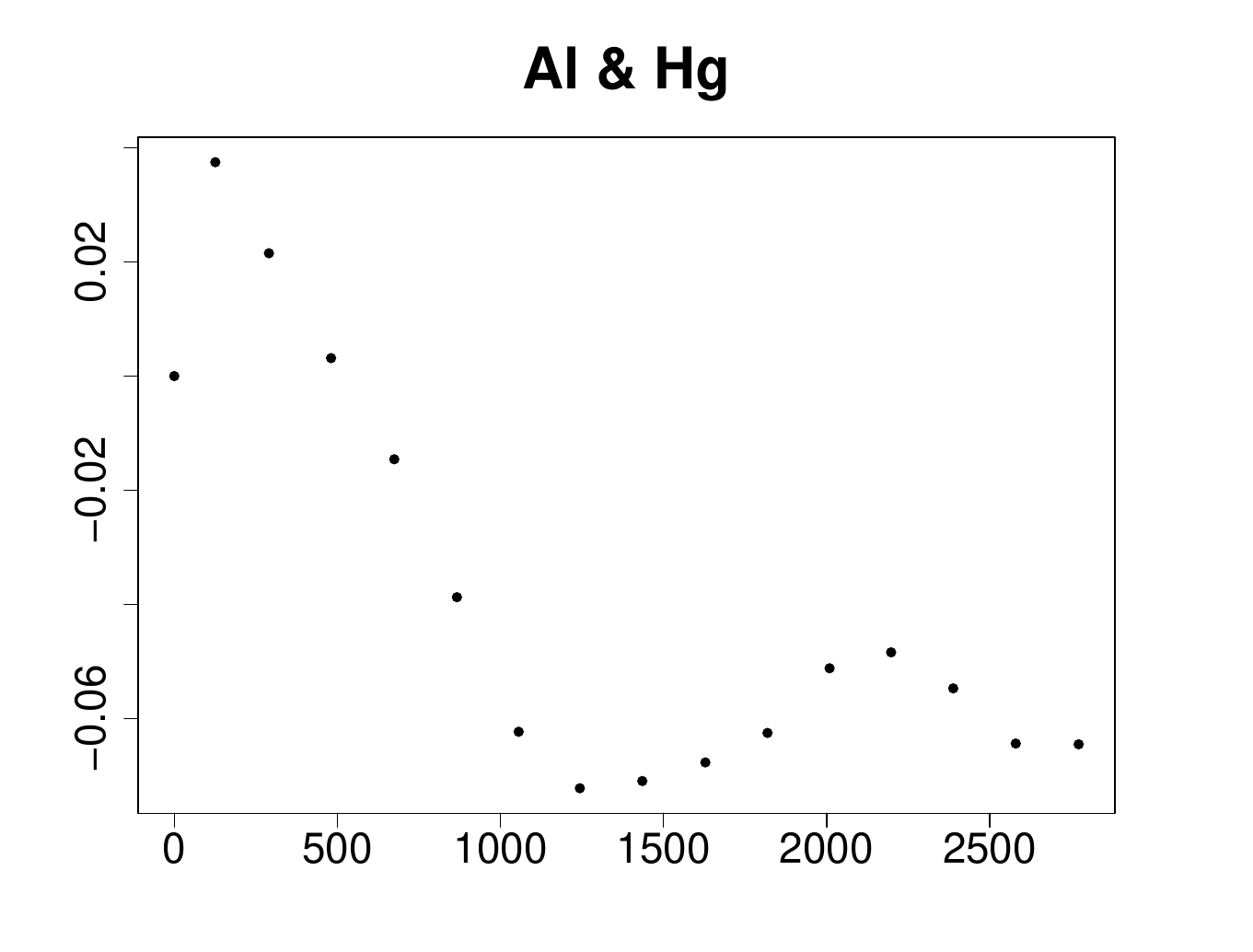}
    \end{minipage}%
    \begin{minipage}{.3\textwidth}
        \centering
        \includegraphics[scale=0.23]{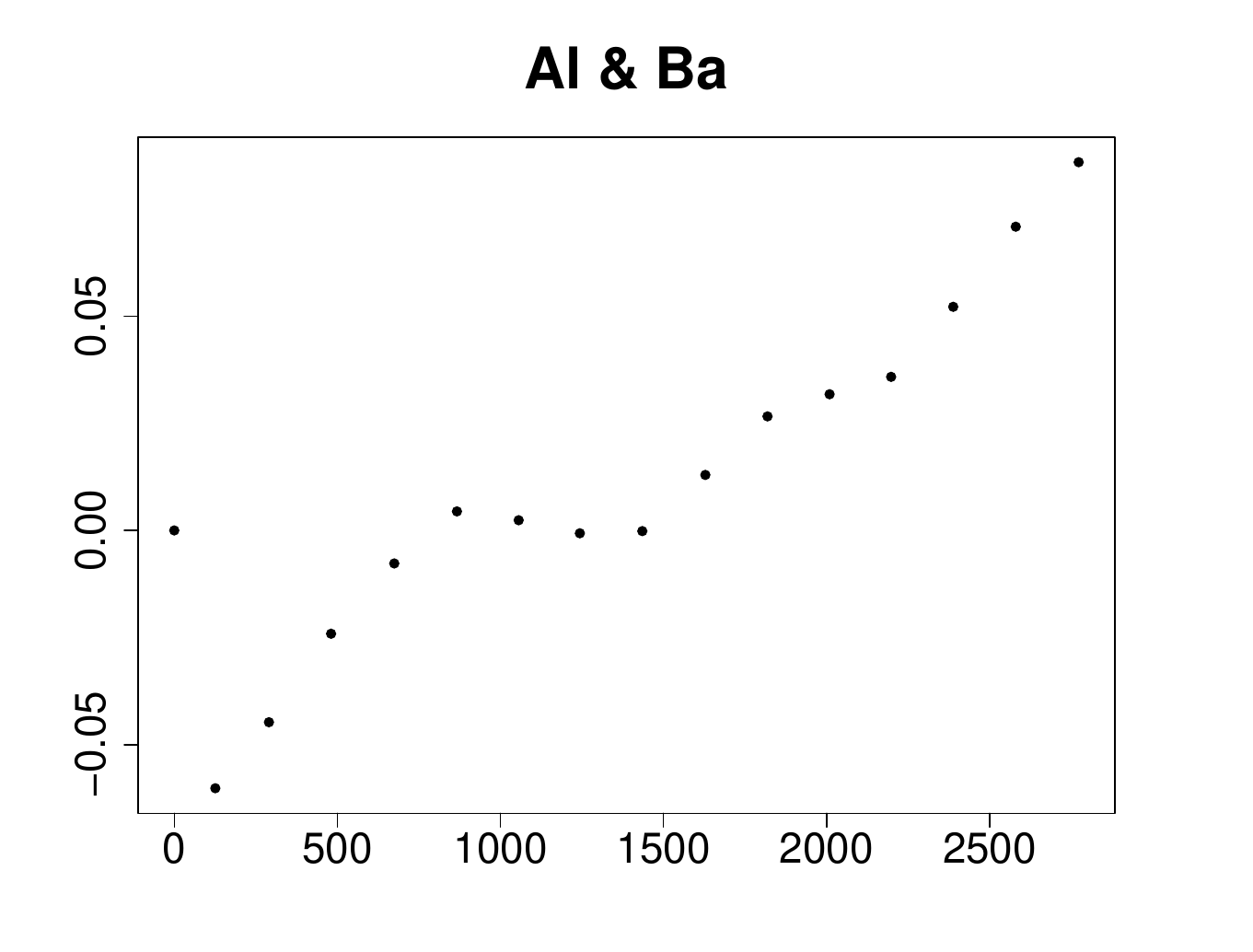}
    \end{minipage}
    
    \caption{Empirical cross-variograms of the variables of interest with the other 
    variables.}\label{visualizacion_variogramas}
\end{figure}

\begin{figure}[H]
    \centering
    \begin{minipage}{.40\textwidth}
        \centering
        \includegraphics[width=0.9\linewidth]{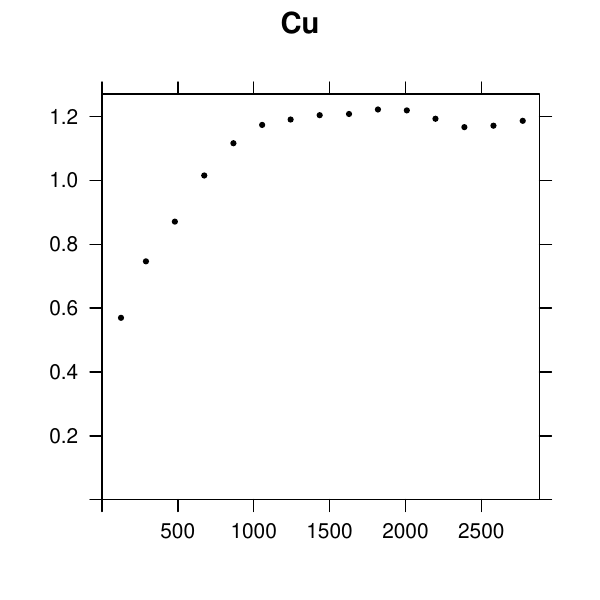}
    \end{minipage}
    \begin{minipage}{.40\textwidth}
        \centering
        \includegraphics[width=0.9\linewidth]{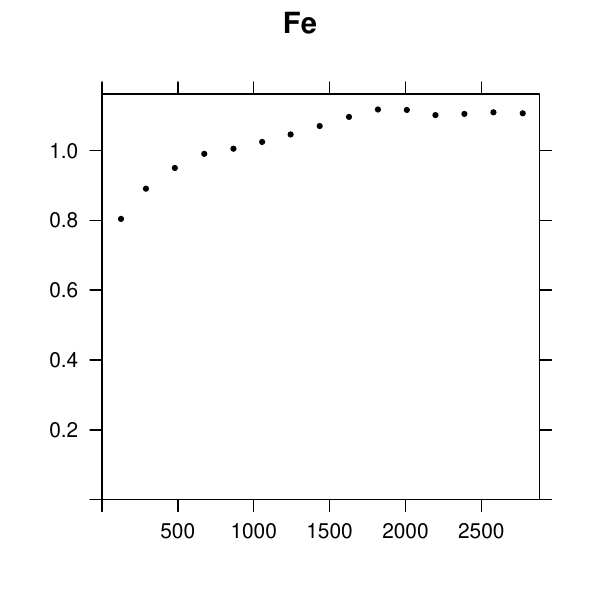}
    \end{minipage}

    \vspace{0.3cm}

    \begin{minipage}{.40\textwidth}
        \centering
        \includegraphics[width=0.9\linewidth]{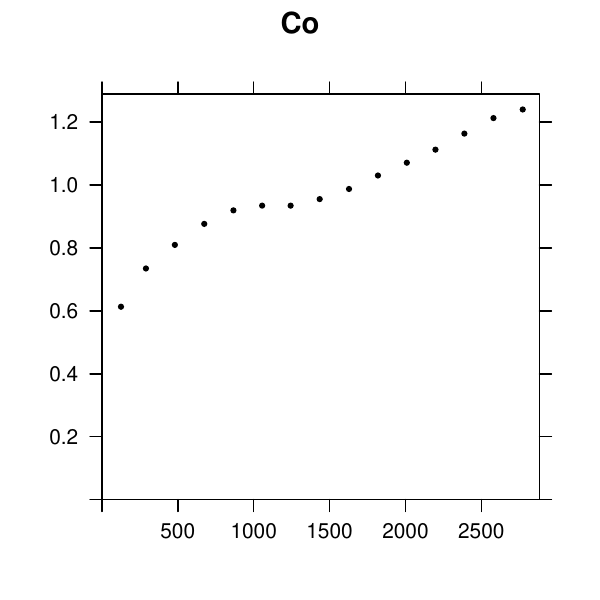}
    \end{minipage}
    \begin{minipage}{.40\textwidth}
        \centering
        \includegraphics[width=0.9\linewidth]{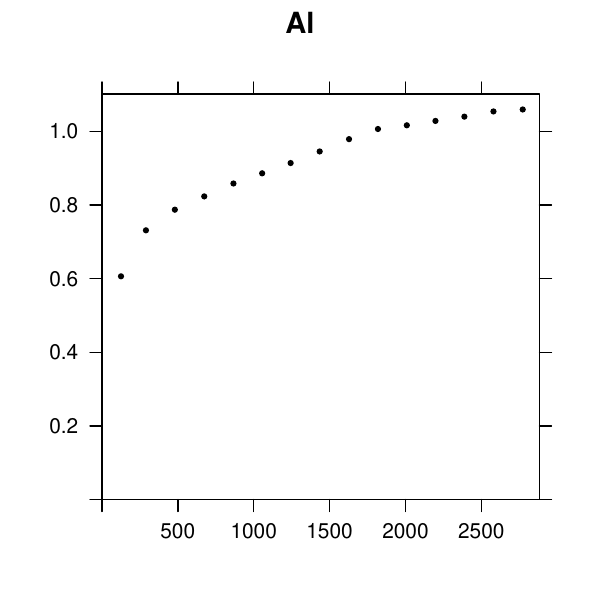}
    \end{minipage}

    \caption{Sample variograms of the variables of interest.}
    \label{variograms_interest}
\end{figure}

\subsection{Cross-validation results}

Table \ref{tab:rmse_vario} presents the RMSE, cokriging error standard deviations, 
and computation times for copper (Cu) across a subset of $\lambda$ values, illustrating 
the trade-off between model sparsity and prediction accuracy.

\begin{table}[H]
    \raggedleft
    \begin{tabular}{|c|c|c|c|c|c|c|c|c|c|}
        \hline
        $\lambda$ & 2.245 & 2.106 & 2.037 & 1.967 & 1.898 & 1.829 & 1.759 
        & 1.689 & 1.620 \\
        \hline
        \textbf{N° Variables} & 1  & 1 & 4 & 4 & 5 & 6 & 9 & 12 & 13 \\
        \hline
        \textbf{RMSE} & 0.7345  & 0.7345 & 0.7357 & 0.7357 & 0.7340 & 0.7331 
        & 0.7311 & 0.7306 & 0.7309 \\
        \hline
        \textbf{\makecell{Standard \\ Deviation}} & 0.1418  & 0.1418 & 0.1414 
        & 0.1414 & 0.1410 & 0.1407 & 0.1403 & 0.1401 & 0.1396 \\
        \hline
        \textbf{Time (min)} & 6.9  & 7.4 & 36.9 & 36.8 & 50.5 & 67.0 & 138.1 
        & 255.8 & 302.7 \\
        \hline
    \end{tabular}
    \caption{RMSE, standard deviation, and computation time for cokriging predictions 
    of Copper (Cu) across a subset of penalty values $\lambda$.}
    \label{tab:rmse_vario}
\end{table}

\end{document}